\documentclass[reqno]{amsart}
\usepackage{fullpage,amsfonts,amssymb,amsthm,mathrsfs,graphicx,color,framed,esint,stackrel,amsmath}
\usepackage{cancel,bm}
\usepackage{tikz}
\usetikzlibrary{decorations.markings}
\usepackage{booktabs,longtable}
\usepackage[initials]{amsrefs}

\DeclareMathAlphabet\mathbfcal{OMS}{cmsy}{b}{n}

\usepackage[colorlinks=true, pdfstartview=FitV, linkcolor=blue, citecolor=blue, urlcolor=blue]{hyperref}
\definecolor{ao}{rgb}{0.0, 0.4, 0.0}

\numberwithin{equation}{section}

\renewcommand{\d}{{\mathrm d}}
\newcommand{\im}{\mathrm{i}}
\newcommand{\e}{\mathrm{e}}

\def\tr{\mathop{\mathrm{tr}}\limits}
\newcommand{\tw}{\stackrel{\boldsymbol{\cdot}}{=}}

\newtheorem{theo}{Theorem}[section]

\newtheorem{lem}[theo]{Lemma}
\newtheorem{rem}[theo]{Remark}

\newtheorem{prop}[theo]{Proposition} 
\newtheorem{cor}[theo]{Corollary} 
\newtheorem{definition}[theo]{Definition}

\usepackage{mdframed}

\begin{document}

\title[Bulk level spacings in the eGinUE]{The complex elliptic Ginibre ensemble at weak non-Hermiticity: bulk spacing distributions}

\author{Thomas Bothner}
\address{School of Mathematics, University of Bristol, Fry Building, Woodland Road, Bristol, BS8 1UG, United Kingdom}
\email{thomas.bothner@bristol.ac.uk}

\author{Alex Little}
\address{School of Mathematics, University of Bristol, Fry Building, Woodland Road, Bristol, BS8 1UG, United Kingdom}
\email{al17344@bristol.ac.uk}
\date{\today}

\keywords{Complex elliptic Ginibre Ensemble, Fredholm determinants, Gaudin-Mehta gap statistics, integro-differential Painlev\'e functions, Gaudin-Mehta and Poisson gap distributions.}

\subjclass[2020]{Primary 60B20; Secondary 60G55, 33E17, 47B35}
\thanks{This work is supported by the Engineering and Physical Sciences Research Council through grant EP/T013893/2.}

\begin{abstract}
We show that the distribution of bulk spacings between pairs of adjacent eigenvalue real parts of a random matrix drawn from the complex elliptic Ginibre ensemble is asymptotically given by a generalization of the Gaudin-Mehta distribution, in the limit of weak non-Hermiticity. The same generalization is expressed in terms of an integro-differential Painlev\'e function and it is shown that the generalized Gaudin-Mehta distribution describes the crossover, with increasing degree of non-Hermiticity, from Gaudin-Mehta nearest-neighbor bulk statistics in the Gaussian Unitary Ensemble to Poisson gap statistics for eigenvalue real parts in the bulk of the Complex Ginibre Ensemble. 

\end{abstract}

\maketitle


\section{Introduction and statement of results}\label{sec1}
Given a matrix $M_n\in\mathbb{C}^{n\times n}$, we let $\{\lambda_j(M_n)\}_{j=1}^n\subset\mathbb{C}$ be the eigenvalues of $M_n$ counted with algebraic multiplicity. The purpose of this paper is to study nearest-neighbor bulk spacings of elements in the set $\{\Re\lambda_j(M_n)\}_{j=1}^n\subset\mathbb{R}$ when $M_n$ is drawn from the complex elliptic Ginibre ensemble (eGinUE), in a suitable scaling limit coined the limit of weak non-Hermiticity, cf. \cite{FKS1,FKS2}. To begin with, we set our notational conventions for the eGinUE, as well as for the Gaussian Unitary Ensemble (GUE) and the Complex Ginibre Ensemble (GinUE).
\begin{definition}[eGinUE and GinUE, \cite{Gir,Gi,LS}] Let $\mathcal{M}_n$ denote the space of matrices $M_n=[\xi_{jk}]_{j,k=1}^n\in\mathbb{C}^{n\times n}$ equipped with the measure
\begin{equation*}
	\d M_n=\prod_{j,k=1}^n\d\Re \xi_{jk}\,\d\Im \xi_{jk}.
\end{equation*}
The \textnormal{eGinUE} is defined as the following probability law on $\mathcal{M}_n$,
\begin{equation}\label{r1}
	\mathbb{P}_{n,\tau}(\d M_n)=Z_{n,\tau}^{-1}\exp\left\{-\frac{n}{1-\tau^2}\tr\Big(M^{\dagger}_nM_n-\frac{\tau}{2}\big(M_n^2+(M_n^{\dagger})^2\big)\Big)\right\}\d M_n,\ \ \ \tau\in[0,1),
\end{equation}
where $Z_{n,\tau}$ is a normalizing constant and $M_n^{\dagger}$ denotes the adjoint of $M_n$. The \textnormal{GinUE} is defined on $\mathcal{M}_n$ by taking $\tau=0$ in \eqref{r1}.
\end{definition}
\begin{definition}[GUE, \cite{Po}] Let $\mathcal{H}_n$ denote the space of Hermitian matrices $M_n=[\zeta_{jk}]_{j,k=1}^n\in\mathbb{C}^{n\times n}$ equipped with the measure
\begin{equation*}
	\d M_n=\prod_{j=1}^n\d\zeta_{jj}\prod_{1\leq j<k\leq n}\d\Re\zeta_{jk}\,\d\Im\zeta_{jk}.
\end{equation*}
The \textnormal{GUE} is defined as the following probability law on $\mathcal{H}_n$,
\begin{equation}\label{r1a}
	\mathbb{P}_n(\d M_n)=Z_n^{-1}\exp\big\{-n\tr(M_n^2)\big\}\d M_n,
\end{equation}
where $Z_n$ is a normalizing constant.
\end{definition}


\subsection{Spacings in the GUE} The bulk distribution of eigenvalues in the GUE is a classical topic in Random Matrix Theory (RMT), cf. \cite{AGZ,F1,M,PS}, which is driven by the Wigner semicircle law. To the point, if $\mathcal{N}_n(\Delta)$ denotes the number of eigenvalues of $M_n\in\textnormal{GUE}$ in a bounded interval $\Delta\subset\mathbb{R}$, then under the law \eqref{r1a},
\begin{equation}\label{r2}
	\mathbb{E}_n\{\mathcal{N}_n(\Delta)\}=n\int_{\Delta}\rho_{\textnormal{sc}}(x)\,\d x+o(n)\ \ \ \ \textnormal{as}\ \ n\rightarrow\infty,
\end{equation}
uniformly in $\Delta$, with $\rho_{\textnormal{sc}}$ equal to the semi-circular density
\begin{equation*}
	\rho_{\textnormal{sc}}(x):=\frac{1}{\pi}\sqrt{(2-x^2)_+}\,,\ \ \ \ \ \ x_+:=\max\{x,0\},\ \ \ x\in\mathbb{R}.
\end{equation*}
Informally, the mean law \eqref{r2} says that the eigenvalues $\lambda_j(M_n)$ concentrate in the interval $[-\sqrt{2},\sqrt{2}]\subset\mathbb{R}$ and for a level $\lambda_0\in\mathbb{R}$ in the interior of the bulk region $-\sqrt{2}+\epsilon\leq \lambda_0\leq \sqrt{2}-\epsilon$ with $\epsilon\in(0,\sqrt{2})$ fixed, the expected eigenvalue spacing should be $1/(n\rho_{\textnormal{sc}}(\lambda_0))$ near $\lambda_0$. In more detail, if we first define the eigenvalue density $p_n:\mathbb{R}^n\rightarrow\mathbb{R}_+$ in the GUE to be the unique symmetric function for which under \eqref{r1a}
\begin{equation}\label{r3}
	\mathbb{E}_n\Big\{F\big(\lambda_1(M_n),\ldots,\lambda_n(M_n)\big)\Big\}=\int_{\mathbb{R}^n}F(x_1,\ldots,x_n)p_n(x_1,\ldots,x_n)\,\d x_1\cdots\d x_n,\ \ \ \ M_n\in\textnormal{GUE},
\end{equation}
for any symmetric Borel function $F:\mathbb{R}^n\rightarrow\mathbb{R}$ of compact support in the region $\{-\infty<x_1<\cdots<x_n<\infty\}$, then, by the work of Dyson \cite{Dy1},
\begin{equation*}
	p_n(x_1,\ldots,x_n)=Q_n^{-1}\exp\left\{-n\sum_{\ell=1}^nx_{\ell}^2\right\}\big|\Delta(x_1,\ldots,x_n)\big|^2,
\end{equation*}
in terms of the Vandermonde $\Delta(x_1,\ldots,x_n):=\prod_{1\leq j<k\leq n}(x_k-x_j)$ and the normalizing constant $Q_n$ given in \cite[$(4.1.29)$]{PS}. In particular, the marginal density
\begin{equation*}
	p_{\ell}^{(n)}(x_1,\ldots,x_{\ell}):=\int_{\mathbb{R}^{n-\ell}}p_n(x_1,\ldots,x_n)\,\d x_{\ell+1}\cdots\d x_n,\ \ \ \ \ 1\leq\ell\leq n,
\end{equation*}
is given by the famous Gaudin-Mehta formula, cf. \cite[$(4.2.22)$]{PS},
\begin{equation}\label{r4}
	p_{\ell}^{(n)}(x_1,\ldots,x_{\ell})=\frac{(n-\ell)!}{n!}\det\big[K_n(x_j,x_k)\big]_{j,k=1}^{\ell},
\end{equation}
where $K_n(x,y)$ is the reproducing kernel
\begin{equation*}
	K_n(x,y)=\e^{-\frac{n}{2}x^2-\frac{n}{2}y^2}\sum_{\ell=0}^{n-1}P_{\ell}^{(n)}(x)P_{\ell}^{(n)}(y),
\end{equation*}
and $\{P_{\ell}^{(n)}\}_{\ell=0}^{\infty}\subset\mathbb{R}[x]$ are the orthonormal polynomials with respect to the measure $\e^{-nx^2}\d x$ on $\mathbb{R}$, i.e.
\begin{equation*}
	\int_{-\infty}^{\infty}\e^{-nx^2}P_j^{(n)}(x)P_k^{(n)}(x)\,\d x=\delta_{jk}\ \ \ \ \ \ \ \ \textnormal{and so}\ \ \ \ \ \ \ \ \ P_{\ell}^{(n)}(x)=n^{\frac{1}{4}}h_{\ell}\big(\sqrt{n}x\big),\end{equation*}
where the family $\{h_{\ell}\}_{\ell=0}^{\infty}\subset\mathbb{R}[x]$ relates to the Hermite polynomials $\{H_{\ell}\}_{\ell=0}^{\infty}\subset\mathbb{R}[x]$, cf. \cite[$\S 18.3$]{NIST}, via
\begin{equation}\label{r5}
	h_{\ell}(x):=\gamma_{\ell}H_{\ell}(x),\ \ \ \ \ \gamma_{\ell}^{-2}:=\sqrt{\pi}\,2^{\ell}\ell!,\ \ \ \ \ \ H_{\ell}(x):=\frac{\ell!}{2\pi\im}\oint_{|\eta|=\delta>0}\e^{2x\eta-\eta^2}\frac{\d\eta}{\eta^{\ell+1}}.
\end{equation}
Furthermore, by \eqref{r3} and \eqref{r4}, the generating functional in the GUE, i.e. the quantity
\begin{equation*}
	E_n[\phi]:=\mathbb{E}_n\left\{\prod_{\ell=1}^n\Big(1-\phi\big(\lambda_{\ell}(M_n)\big)\Big)\right\}=1+\sum_{\ell=1}^n(-1)^{\ell}\binom{n}{\ell}\int_{\mathbb{R}^{\ell}}p_{\ell}^{(n)}(x_1,\ldots,x_{\ell})\phi(x_1)\cdots\phi(x_{\ell})\,\d x_1\cdots\d x_{\ell},
\end{equation*}
admits the Fredholm determinant representation, cf. \cite[$(4.2.20)$]{PS},
\begin{equation}\label{r6}
	E_n[\phi]=1+\sum_{\ell=1}^n\frac{(-1)^{\ell}}{\ell!}\int_{J^{\ell}}\phi(x_1)\cdots\phi(x_{\ell})\det\big[K_n(x_j,x_k)\big]_{j,k=1}^{\ell}\d x_1\cdots\d x_{\ell},
\end{equation}
as soon as the support $J=\textnormal{supp}(\phi)$ of the bounded test function $\phi:\mathbb{R}\rightarrow\mathbb{C}$ is compact. Hence, by \eqref{r6} and the inclusion-exclusion principle, the probability $E_n(\Delta):=\mathbb{P}_n\{\mathcal{N}_n(\Delta)=0\}$ that $M_n\in\textnormal{GUE}$ has no eigenvalues in a given bounded interval $\Delta\subset\mathbb{R}$, is equal to
\begin{equation}\label{r7}
	E_n(\Delta)=E_n[\chi_{\Delta}]=1+\sum_{\ell=1}^n\frac{(-1)^{\ell}}{\ell!}\int_{\Delta^{\ell}}\det\big[K_n(x_j,x_k)\big]_{j,k=1}^{\ell}\d x_1\cdots\d x_{\ell},
\end{equation}
with the indicator function $\chi_{\Delta}$ on $\Delta$. Seeing that the asymptotics of $K_n$ are very-well understood, the precise eigenvalue spacing near a level $\lambda_0$ in the interior of the bulk region becomes accessible: by the Plancherel-Rotach asymptotics \cite[$\S 18.15$(v)]{NIST},
\begin{equation}\label{r8}
	\lim_{n\rightarrow\infty}\big(\rho_n(\lambda_0)\big)^{-\ell}p_{\ell}^{(n)}\left(\lambda_0+\frac{\mu_1}{n\rho_n(\lambda_0)},\ldots,\lambda_0+\frac{\mu_{\ell}}{n\rho_n(\lambda_0)}\right)=\det\left[\frac{\sin\pi(\mu-\mu_k)}{\pi(\mu_j-\mu_k)}\right]_{j,k=1}^{\ell},
\end{equation}
uniformly in $(\mu_1,\ldots,\mu_{\ell})$ chosen from any compact subset of $\mathbb{R}^{\ell}$ where $\rho_n(x):=p_1^{(n)}(x)$ is the density of the averaged normalized counting measure, i.e. the density of $\mathbb{E}_n\{\frac{1}{n}\mathcal{N}_n(\cdot)\}$. Note the consistency of the normalization in the LHS of \eqref{r8} with the aforementioned heuristic that the expected eigenvalue spacing should be $1/(n\rho_{\textnormal{sc}}(\lambda_0))\sim1/(n\rho_n(\lambda_0))$ near $\lambda_0$ for large $n$. From \eqref{r8}, one obtains subsequently that
\begin{equation*}
	E_n\left(\lambda_0+\frac{\Delta}{n\rho_n(\lambda_0)}\right)=1+\sum_{\ell=1}^{\infty}\frac{(-1)^{\ell}}{\ell!}\int_{\Delta^{\ell}}\det\left[\frac{\sin\pi(\mu_j-\mu_k)}{\pi(\mu_j-\mu_k)}\right]_{j,k=1}^{\ell}\d\mu_1\cdots\d\mu_k+o(1)
\end{equation*}
as $n\rightarrow\infty$ uniformly in the bounded interval $\Delta\subset\mathbb{R}$, or in equivalent form, as $n\rightarrow\infty$,
\begin{equation}\label{r9}
	E_n\left(\lambda_0+\frac{\Delta}{n\rho_n(\lambda_0)}\right)=\prod_{j=1}^{\infty}\big(1-\omega_j(K_{\sin},\Delta)\big)+o(1),\ \ \ \ \ \ \ \lambda_0+\frac{\Delta}{c}:=\left\{\lambda_0+\frac{x}{c}:\,x\in\Delta\right\},\ c>0,
\end{equation}
where $\omega_j(K_{\sin},\Delta)$ are the non-zero eigenvalues, counted according to their algebraic multiplicities, of the trace class integral operator $K_{\sin}:L^2(\Delta)\rightarrow L^2(\Delta)$ with 
\begin{equation*}
	(K_{\sin}f)(x):=\int_{\Delta}\frac{\sin\pi(x-y)}{\pi(x-y)}f(y)\,\d y.
\end{equation*}
The asymptotic \eqref{r9} can now be used to control the spacing distribution of an individual bulk eigenvalue $\lambda_0$. Indeed, the conditional probability that the distance between an eigenvalue at $\lambda_0$ and its nearest neighbor to the right is bigger than $s/(n\rho_n(\lambda_0))$ for large $n$ and fixed $s>0$, equals
\begin{equation}\label{r10}
	q_n(s):=\lim_{h\downarrow 0}\mathbb{P}_n\left\{\textnormal{no eigenvalues in}\ \left(\lambda_0,\lambda_0+\frac{s}{n\rho_n(\lambda_0)}\right]\ \,\bigg|\,\ \textnormal{an eigenvalue in}\ \left(\lambda_0-\frac{h}{n\rho_n(\lambda_0)},\lambda_0\right]\right\},
\end{equation}
and so, by \eqref{r9}, the to \eqref{r10} corresponding limiting spacing distribution becomes the Gaudin-Mehta distribution, 
\begin{equation}\label{r11}
	F_1(s):=1-\lim_{n\rightarrow\infty}q_n(s)=\int_0^s\frac{\d^2}{\d u^2}\prod_{j=1}^{\infty}\Big(1-\omega_j\big(K_{\sin},(0,u)\big)\Big)\,\d u,\ \ \ s>0,
\end{equation}
see \cite[$(5.2.57)$]{PS}, uniformly in $s\in(0,\infty)$ on any compact set, with
\begin{equation*}
	\frac{\d}{\d s}F_1(s)=\frac{1}{3}(\pi s)^2+\mathcal{O}\big(s^4\big),\ \ \ s\downarrow 0. 
\end{equation*}
Note that \eqref{r11} also equals the limiting distribution function for the average \cite[$(1.85)$]{DKMVZ} and single \cite[$(7)$]{Ta} gap spacing in the GUE. From an integrable systems viewpoint we highlight that the Fredholm determinant in the leading order of \eqref{r9} relates to a distinguished solution of Painlev\'e-V, see \cite{JMMS}, however for us the following alternative Painlev\'e representation will be more interesting, see \cite[Lemma $3.6.6$]{AGZ}\footnote{One needs the identification $u(t)=2\pi r(2\pi t)$ to map $r=r(t)$ in \cite[$(3.6.36),(3.6.37)$]{AGZ} to our $u=u(t)$ in \eqref{r12}.}. For any $t>0$,
\begin{equation}\label{r12}
	D(t):=\prod_{j=1}^{\infty}\Big(1-\omega_j\big(K_{\sin},(-t,t)\big)\Big)=\exp\left[-2t-\int_0^t(t-s)\big(u(s)\big)^2\,\d s\right],
\end{equation}
where $u=u(t)$ solves the differential equation
\begin{equation}\label{r13}
	\left[\frac{\d^2}{\d t^2}\big(tu(t)\big)+4\pi^2tu(t)\right]^2=4\big(u(t)\big)^2\left[\bigg(\frac{\d}{\d t}\big(tu(t)\big)\bigg)^2+4\pi^2\big(tu(t)\big)^2\right]
\end{equation}
with boundary constraint $u(t)\sim 2+\mathcal{O}(t)$ 
as $t\downarrow 0$. Equation \eqref{r13} also relates to a Painlev\'e-V equation, see \cite[$(21.3.19),(21.3.20)$]{M} and in particular \cite[Chapter $21.5$]{M}.
\begin{center}
\begin{figure}[tbh]
\resizebox{0.456\textwidth}{!}{\includegraphics{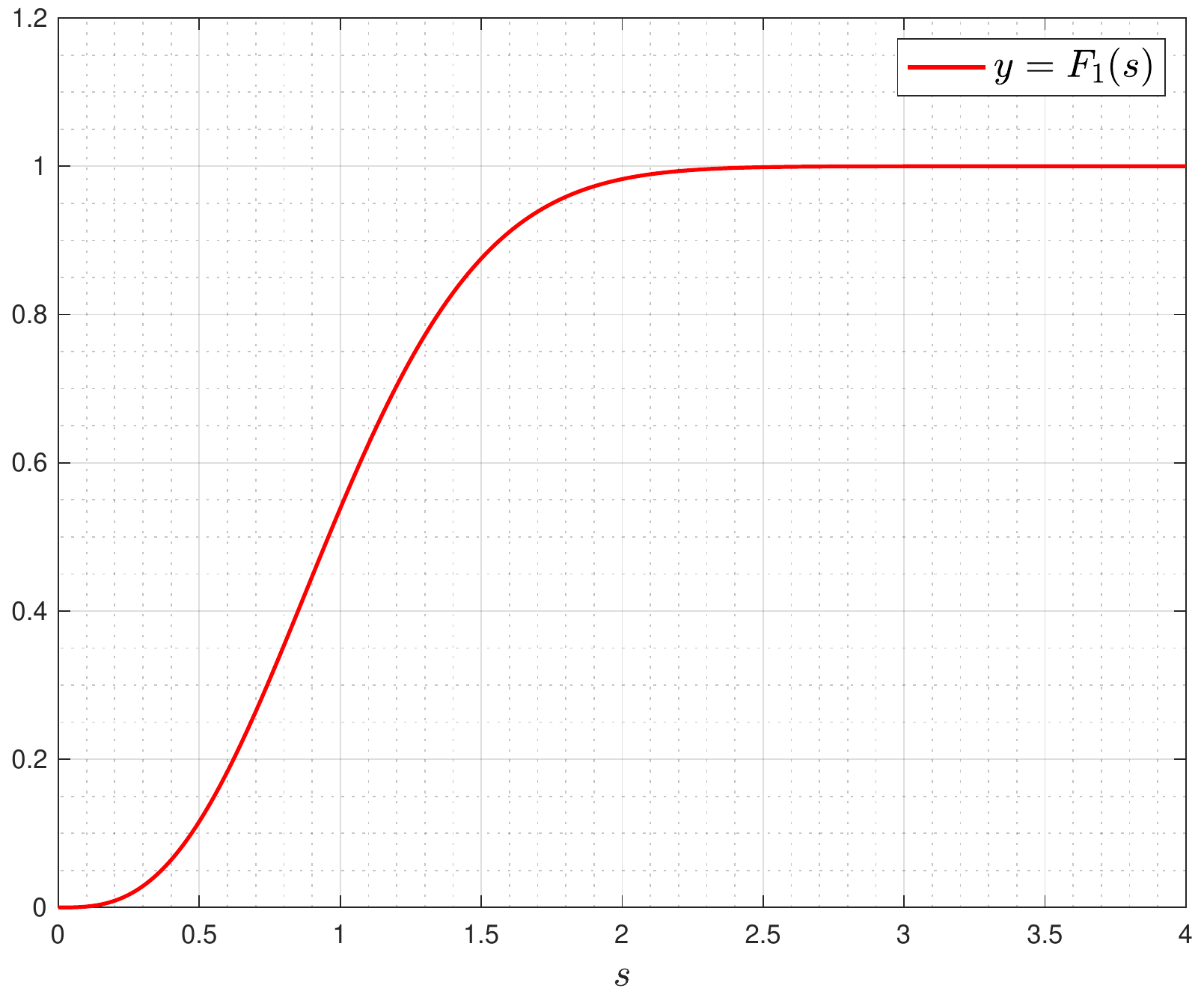}}\ \ \ \ \ \ \ \ \resizebox{0.455\textwidth}{!}{\includegraphics{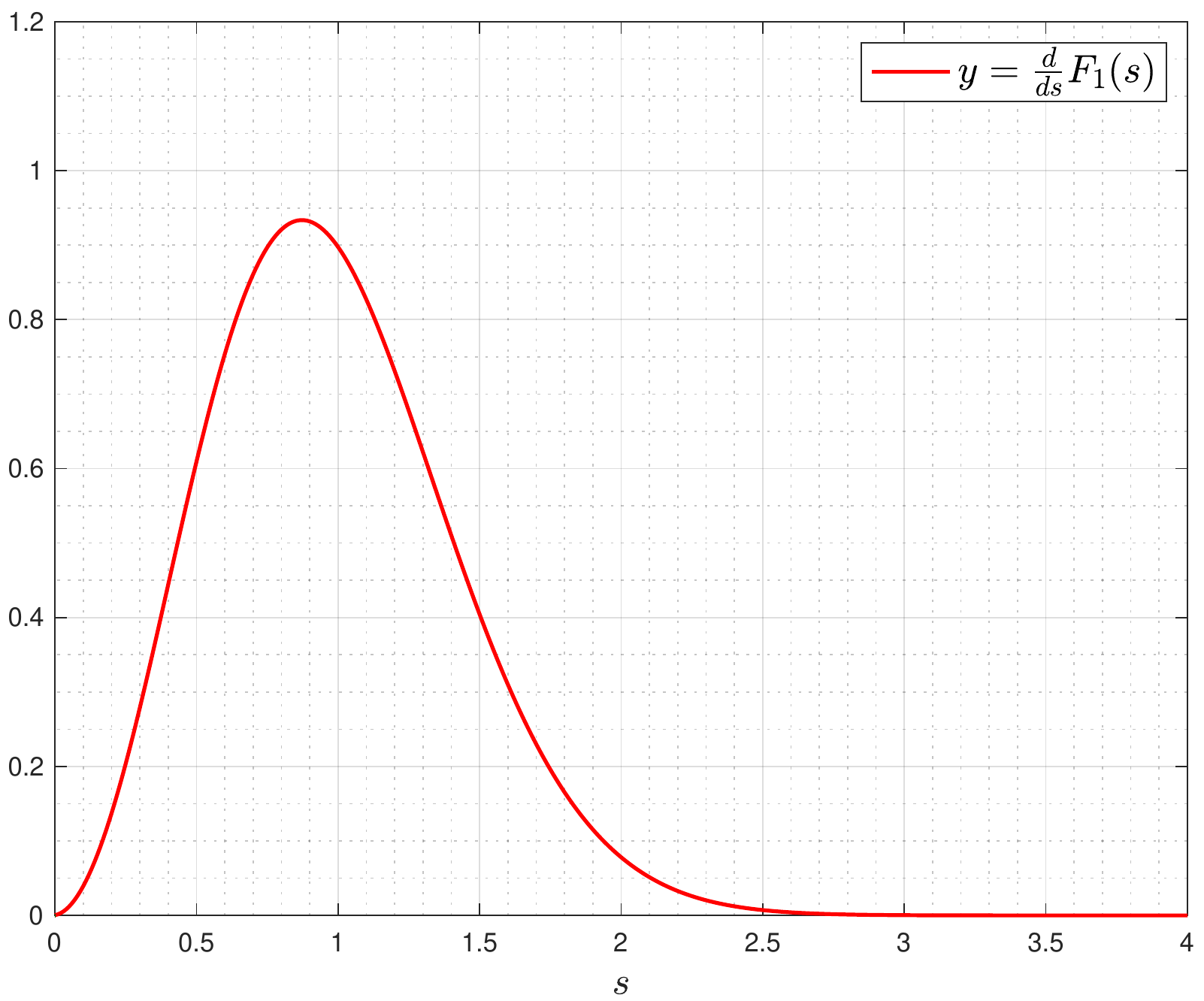}}
\caption{The Gaudin-Mehta bulk spacing distribution \eqref{r11} in the GUE. The plots were generated in MATLAB with $m=80$ quadrature points using the Nystr\"om method with Gauss-Legendre quadrature. On the left the distribution function, on the right the density function.}
\label{figure1}
\end{figure}
\end{center}

\subsection{Spacings in the GinUE} Moving to the GinUE, bulk spacings for eigenvalue real parts should also be considered classical, albeit the same topic is somewhat under appreciated in RMT literature. Namely, if $\mathcal{R}_n(\Delta)$ denotes the number of eigenvalue real parts of $M_n\in\textnormal{GinUE}$ in a bounded interval $\Delta\subset\mathbb{R}$, then our workings in Section \ref{sec2} below show that under the law \eqref{r1},
\begin{equation}\label{r14}
	\mathbb{E}_{n,0}\{\mathcal{R}_n(\Delta)\}=n\int_{\Delta}\varrho_0(x)\,\d x+o(n)\ \ \ \ \textnormal{as}\ \ n\rightarrow\infty,
\end{equation}
uniformly in $\Delta$, with $\varrho_0(x)$ equal to the density 
\begin{equation*}
	\varrho_0(x):=\frac{2}{\pi}\sqrt{(1-x^2)_+},\ \ \ x\in\mathbb{R}.
\end{equation*}
Note that the integral in the right hand side of \eqref{r14} is the integral of the uniform density $\rho_{\textnormal{c}}$ on the unit disk over $\Omega_{\Delta}:=\Delta\times\mathbb{R}\subset\mathbb{R}^2\simeq\mathbb{C}$, i.e.
\begin{equation}\label{r15}
	\mathbb{E}_{n,0}\{\mathcal{R}_n(\Delta)\}=n\int_{\Omega_{\Delta}}\rho_{\textnormal{c}}(z)\,\d^2 z+o(n)\ \ \ \ \ \textnormal{as}\ \ n\rightarrow\infty.
\end{equation}
Hence, eigenvalue real parts of $M_n\in\textnormal{GinUE}$ concentrate predictably, in view of the circular law \cite{Gi}, in the interval $[-1,1]\subset\mathbb{R}$ and for a level $\lambda_0\in\mathbb{C}$ in the interior of the bulk region $|\lambda_0|\leq 1-\epsilon$ with $\epsilon\in(0,1)$ fixed, the expected \textit{horizontal} eigenvalue spacing should be $1/(n\varrho_0(\Re\lambda_0))$ near $\lambda_0$. 
\begin{rem}\label{rem:1}
The circular law foretells the expected radial eigenvalue spacing to be of order $1/(\sqrt{n}\rho_{\textnormal{c}}(\lambda_0))$ near a bulk level $\lambda_0\in\mathbb{C}$ and the same spacing yields the Ginibre bulk kernel, cf. \cite[$(1.42)$]{Gi}. Indeed, if 
\begin{equation*}
	p_{\ell}^{(n,0)}(z_1,\ldots,z_{\ell}):=\frac{(n-\ell)!}{n!}\det\big[K_{n,0}(z_j,z_k)\big]_{j,k=1}^{\ell},\ \ \ K_{n,0}(z,w):=\frac{n}{\pi}\e^{-\frac{n}{2}|z|^2-\frac{n}{2}|w|^2}\sum_{\ell=0}^{n-1}\frac{n^{\ell}}{\ell!}(z\overline{w})^{\ell},
\end{equation*}
denotes the \textnormal{GinUE} marginal density of its symmetric joint eigenvalue density, see e.g. \cite[$(15.1.31)$]{M}, then
\begin{equation}\label{r16}
	\lim_{n\rightarrow\infty}p_{\ell}^{(n,0)}\left(\lambda_0+\frac{w_1}{\sqrt{n}},\ldots,\lambda_0+\frac{w_{\ell}}{\sqrt{n}}\right)=\det\left[\frac{1}{\pi}\e^{-\frac{1}{2}|w_j|^2-\frac{1}{2}|w_k|^2+w_j\overline{w_k}}\,\right]_{j,k=1}^{\ell},
\end{equation}
uniformly in $(w_1,\ldots,w_{\ell})$ chosen from any compact subset of $\mathbb{C}^{\ell}$, and valid for any $\lambda_0$ in the interior of the unit disk. 
As we are interested in spacings of eigenvalue real parts, \eqref{r16} will not be useful in the following.
\end{rem}
Aiming at a better understanding of bulk spacings of eigenvalue real parts in the GinUE, we now formulate an analogue of \eqref{r9} which will prove useful later on. Similarly to \eqref{r7} and the GUE, we let $E_{n,0}(\Omega):=\mathbb{P}_{n,0}\{\mathcal{N}_n(\Omega)=0\}$ denote the probability, under the law \eqref{r1} with $\tau=0$, that $M_n\in\textnormal{GinUE}$ has no eigenvalues in a given rectangle $\Omega\subset\mathbb{C}$.
\begin{lem}\label{lem:1} Assume $\lambda_0\in\mathbb{C}$ belongs to the interior of the support $\{z\in\mathbb{C}:\,|z|\leq 1\}$ of the circular law and suppose $\Delta\subset\mathbb{R}$ is a bounded interval of length $|\Delta|\in(0,\infty)$. Then
\begin{equation}\label{r17}
	\lim_{n\rightarrow\infty}E_{n,0}\left(\lambda_0+\frac{\Omega_{\Delta}}{n\varrho_0(\Re\lambda_0)}\right)=\e^{-|\Delta|},\ \ \ \ \ \Omega_{\Delta}=\Delta\times\mathbb{R}\subset\mathbb{R}^2\simeq\mathbb{C},
\end{equation}
where $\varrho_0(x)$ appeared in \eqref{r14} 
and where we abbreviate $\lambda_0+\frac{\Omega}{c}:=\{\lambda_0+\frac{z}{c}:\,z\in\Omega\}$ for any $c>0$.
\end{lem}
The limit \eqref{r17} tells us that bulk eigenvalue real parts in the GinUE, as $n\rightarrow\infty$, behave like points in a Poisson process on $\mathbb{R}$ with unit rate. To the point, for large $n$, bulk eigenvalue real parts in the GinUE form a determinantal point process on $\mathbb{R}$ determined by the correlation kernel
\begin{equation*}
	K_{\textnormal{Poi}}(x,y):=\begin{cases}1,&x=y\\ 0,&x\neq y\end{cases},\ \ \ \ \ \ (x,y)\in\mathbb{R}^2.
\end{equation*}
This observation underwrites our next result about the spacing distribution of an individual bulk eigenvalue real part. We view the same result as the GinUE analogue of \eqref{r11} and it reads as follows.
\begin{center}
\begin{figure}[tbh]
\resizebox{0.456\textwidth}{!}{\includegraphics{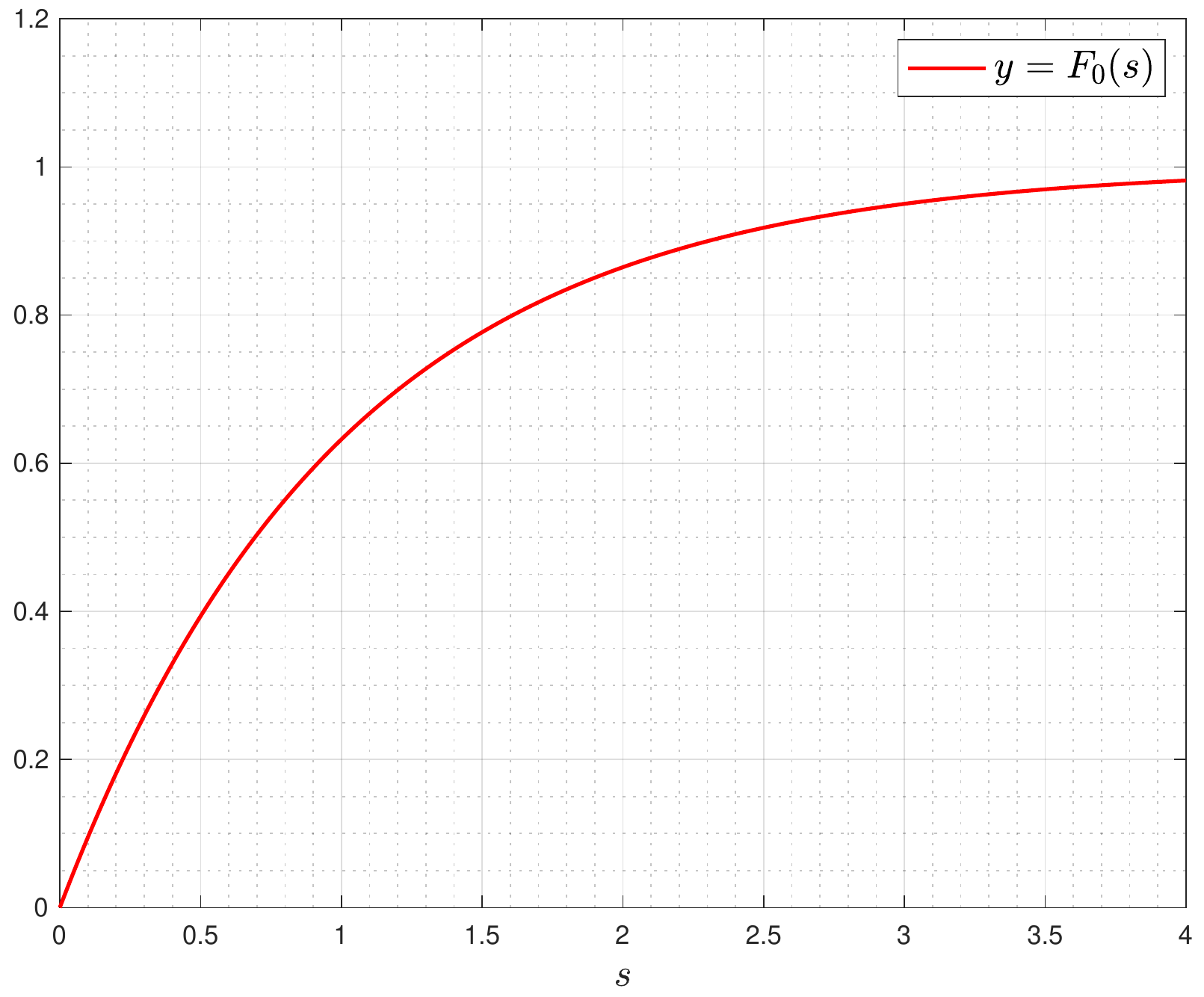}}\ \ \ \ \ \ \ \ \resizebox{0.455\textwidth}{!}{\includegraphics{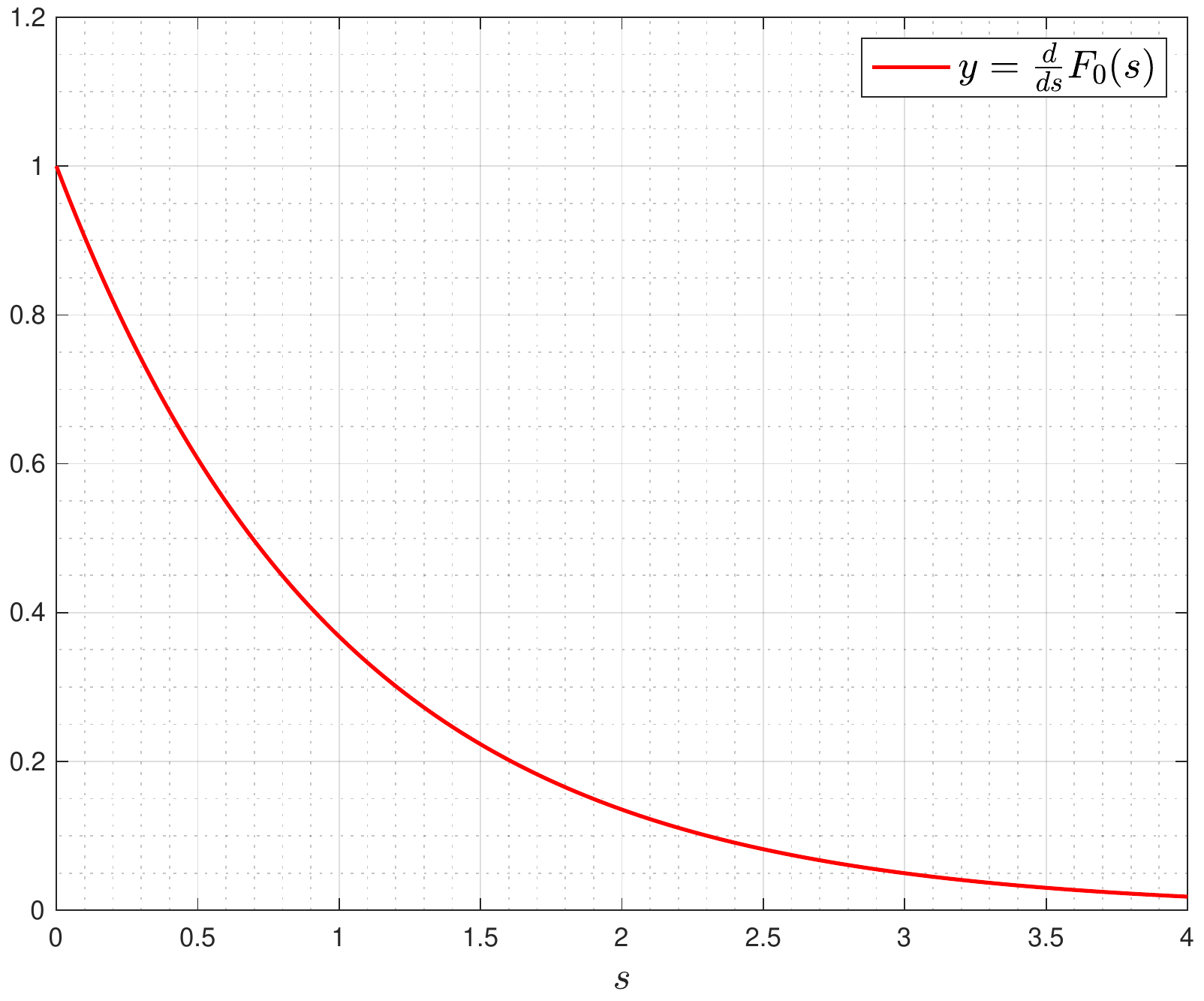}}
\caption{The Poisson bulk spacing distribution \eqref{r18} in the GinUE for real parts. The plots were generated in MATLAB. On the left the distribution function, on the right the density function.}
\label{figure2}
\end{figure}
\end{center} 
\begin{cor}\label{cor:1} Assume $\lambda_0\in\mathbb{R}$ belongs to the interior of the interval $[-1,1]\subset\mathbb{R}$ and let $q_{n,0}(s)$ denote the conditional probability that the distance between an eigenvalue real part at $\lambda_0$ and its nearest real part neighbor to the right is bigger than $s/(n\varrho_0(\lambda_0))$ for large $n$ and fixed $s>0$, under \eqref{r1} with $\tau=0$, i.e.
\begin{equation*}
	q_{n,0}(s):=\lim_{h\downarrow 0}\mathbb{P}_{n,0}\bigg\{\textnormal{no eigenvalue r.ps. in}\ \left(\lambda_0,\lambda_0+\frac{s}{n\varrho_0(\lambda_0)}\right]\,\bigg|\,\textnormal{an eigenvalue r.p. in}\ \left(\lambda_0-\frac{h}{n\varrho_0(\lambda_0)},\lambda_0\right]\bigg\}.
\end{equation*}
Then the corresponding limiting spacing distribution is the Poisson gap distribution, 
\begin{equation}\label{r18}
	F_0(s):=1-\lim_{n\rightarrow\infty}q_{n,0}(s)=\int_0^s\frac{\d^2}{\d u^2}\,\e^{-u}\,\d u=1-\e^{-s},\ \ \ s>0,
\end{equation}
uniformly in $s\in(0,\infty)$ in any compact subset with 
\begin{equation*}
	\frac{\d}{\d s}F_0(s)=1+\mathcal{O}(s),\ \ \ s\downarrow 0.
\end{equation*}
\end{cor}

When comparing \eqref{r11} and \eqref{r18} one can wonder which dynamical system interpolates between $D(\frac{t}{2})$, appearing in \eqref{r11},\eqref{r12} and upheld by the Painlev\'e type equation \eqref{r13}, and the elementary $\e^{-t}$ appearing in \eqref{r18}. We will answer this question affirmatively through our upcoming analysis of the eGinUE.
\begin{rem} The limiting behavior \eqref{r18} reaffirms the statistical independence of eigenvalue real parts in the \textnormal{GinUE} in the large $n$-limit. This was previously known to hold true near the edge, cf. \cite{Ben,CESX}, and is now known to hold true in the bulk as well. More classical is the independence of eigenvalue moduli in the \textnormal{GinUE}, see \cite{Kos}.
\end{rem}
\subsection{Spacings in the eGinUE} The statistical analysis of eigenvalue real parts in the eGinUE is made possible by the remarkable insights of \cite{FGIL,FJ} on the exact solvability of the same matrix model. Namely, if we define the eigenvalue density $p_{n,\tau}:\mathbb{C}^n\rightarrow\mathbb{R}_+$ in the eGinUE to be the unique symmetric function for which
\begin{equation*}
	\mathbb{E}_{n,\tau}\Big\{F\big(\lambda_1(M_n),\ldots,\lambda_n(M_n)\big)\Big\}=\int_{\mathbb{C}^n}F(z_1,\ldots,z_n)p_{n,\tau}(z_1,\ldots,z_n)\,\d^2z_1\cdots\d^2z_n,\ \ \ M_n\in\textnormal{eGinUE},
\end{equation*}
for any symmetric Borel function $F:\mathbb{C}^n\rightarrow\mathbb{R}$ of compact support in the region $\{-\infty<|z_1|\leq\cdots\leq|z_n|<\infty\ \textnormal{and}\ \textnormal{arg}\,z_j<\textnormal{arg}\,z_{j+1}\ \textnormal{if}\ |z_j|=|z_{j+1}|\ \textnormal{for some}\ j\}$, then by the works \cite{FKS1,FKS2,FSK} of Fyodorov, Khoruzhenko and Sommers, assuming $\tau\in(0,1)$ henceforth,
\begin{equation*}
	p_{n,\tau}(z_1,\ldots,z_n)=Q_{n,\tau}^{-1}\exp\left\{-n\sum_{\ell=1}^nV_{\tau}(z_{\ell})\right\}\big|\Delta(z_1,\ldots,z_n)\big|^2,\ \ \ V_{\tau}(z):=\frac{1}{1-\tau^2}\Big(|z|^2-\tau\,\Re(z^2)\Big),
\end{equation*}
in terms of the Vandermonde $\Delta(z_1,\ldots,z_n)$ and the normalizing constant $Q_{n,\tau}$ in \cite[$(2)$]{FKS2}. Consequently, with the help of \cite{FGIL,FJ}, the orthogonal polynomial method in RMT allows one to compute the marginal density
\begin{equation*}
	p_{\ell}^{(n,\tau)}(z_1,\ldots,z_{\ell}):=\int_{\mathbb{C}^{n-\ell}}p_{n,\tau}(z_1,\ldots,z_n)\,\d^2z_{\ell+1}\cdots\d^2z_n,\ \ \ 1\leq\ell\leq n,
\end{equation*}
in Gaudin-Mehta fashion, cf. \cite[Section $3.1$]{FGIL},
\begin{equation}\label{r19}
	p_{\ell}^{(n,\tau)}(z_1,\ldots,z_{\ell})=\frac{(n-\ell)!}{n!}\det\big[K_{n,\tau}(z_j,z_k)\big]_{j,k=1}^{\ell},
\end{equation}
where $K_{n,\tau}(z,w)$ is the reproducing kernel
\begin{equation}\label{r20}
	K_{n,\tau}(z,w)=\e^{-\frac{n}{2}V_{\tau}(z)-\frac{n}{2}V_{\tau}(w)}\sum_{\ell=0}^{n-1}P_{\ell}^{(n,\tau)}(z)\overline{P_{\ell}^{(n,\tau)}(w)},
\end{equation}
and $\{P_{\ell}^{(n,\tau)}\}_{\ell=0}^{\infty}\subset\mathbb{R}[z]$ are the orthonormal polynomials with respect to the measure $\e^{-nV_{\tau}(z)}\d^2 z$ on $\mathbb{C}$, i.e.
\begin{equation*}
	\int_{\mathbb{C}}\e^{-nV_{\tau}(z)}P_j^{(n,\tau)}(z)\overline{P_k^{(n,\tau)}(z)}\,\d^2 z=\delta_{jk}\ \ \ \ \ \ \textnormal{and so}\ \ \ \ \ \ P_{\ell}^{(n,\tau)}(z)=n^{\frac{1}{2}}h_{\ell}^{\tau}(\sqrt{n}\,z),
\end{equation*}
where the family $\{h_{\ell}^{\tau}\}_{\ell=0}^{\infty}\subset\mathbb{R}[z]$ relates again to the Hermite polynomials $\{H_{\ell}\}_{\ell=0}^{\infty}\subset\mathbb{R}[x]$ in \eqref{r5},
\begin{equation}\label{r21}
	h_{\ell}^{\tau}(z):=\gamma_{\ell}^{\tau}H_{\ell}\Big(\frac{z}{\sqrt{2\tau}}\Big),\ \ \ \ \ \ (\gamma_{\ell}^{\tau})^{-2}:=\pi\,\ell!\sqrt{1-\tau^2}\Big(\frac{2}{\tau}\Big)^{\ell}.
\end{equation}
With \eqref{r19} in place it is now straightforward to compute the generating functional in the eGinUE, i.e. the quantity
\begin{equation*}
	E_{n,\tau}[\phi]:=\mathbb{E}_{n,\tau}\left\{\prod_{\ell=1}^n\big(1-\phi(\lambda_{\ell}(M_n))\big)\right\},\ \ \ \ M_n\in\textnormal{eGinUE},
\end{equation*}
taken under the law \eqref{r1} with $\tau\in(0,1)$, as Fredholm determinant, using the general theory of determinantal point processes, see for instance \cite{Joh0}. In turn, we are able to express the gap probability $E_{n,\tau}(\Omega)$, i.e. the probability that $M_n\in\textnormal{eGinUE}$ has no eigenvalues in a rectangle $\Omega\subset\mathbb{C}$, as Fredholm determinant, too.
\begin{prop}[{\cite[Proposition $2.2$]{Joh0}}]\label{prop:1} Suppose the bounded test function $\phi:\mathbb{C}\rightarrow\mathbb{C}$ has compact support $J=\textnormal{supp}\,\phi\subset\mathbb{C}$. Then
\begin{equation}\label{r22}
	E_{n,\tau}[\phi]=1+\sum_{\ell=1}^n\frac{(-1)^{\ell}}{\ell!}\int_{J^{\ell}}\phi(z_1)\cdots\phi(z_{\ell})\det\big[K_{n,\tau}(z_j,z_k)\big]_{j,k=1}^{\ell}\d^2 z_1\cdots\d^2 z_{\ell},
\end{equation}
valid for all $\tau\in(0,1)$ in terms of the kernel $K_{n,\tau}(z,w)$ in \eqref{r20}. In particular, 
\begin{equation*}
	E_{n,\tau}(\Omega)=E_{n,\tau}[\chi_{\Omega}]
\end{equation*}
valid for any bounded rectangle $\Omega\subset\mathbb{C}$ in terms of the indicator function $\chi_{\Omega}$ on $\Omega$.
\end{prop}
Having flushed out \eqref{r19} and \eqref{r22} we proceed with our analysis of the bulk distribution of eigenvalue real parts in the eGinUE. First, we remind the reader about the celebrated elliptic law \cite{Gir}: if $\mathcal{N}_{n,\tau}(\Omega)$ denotes the number of eigenvalues of $M_n\in\textnormal{eGinUE}$ in a bounded rectangle $\Omega\subset\mathbb{C}$, then under the law \eqref{r1},
\begin{equation}\label{r23}
	\mathbb{E}_{n,\tau}\{\mathcal{N}_{n,\tau}(\Omega)\}=n\int_{\Omega}\rho_{\textnormal{e}}(z)\,\d^2 z+o(n)\ \ \ \ \textnormal{as}\ \ n\rightarrow\infty,
\end{equation}
uniformly in $\Omega$ and uniformly in $\tau\in(0,1)$ on compact subsets, with $\rho_{\textnormal{e}}$ equal to the uniform density 
\begin{equation*}
	\rho_{\textnormal{e}}(z):=\frac{1}{\pi(1-\tau^2)}\chi_{\mathbb{D}_{\tau}}(z),\ \ \ \ \ \mathbb{D}_{\tau}:=\Big\{z\in\mathbb{C}:\ \bigg(\frac{\Re z}{1+\tau}\bigg)^2+\bigg(\frac{\Im z}{1-\tau}\bigg)^2\leq 1\Big\}.
\end{equation*}
Moving now to the statistics of bulk eigenvalue real parts we can, in principle, continue in two qualitatively different directions: One $\textnormal{(i)}$, we keep $\tau\in(0,1)$ \textit{fixed} and find in Section \ref{sec3} that for the number $\mathcal{R}_{n,\tau}(\Delta)$ of eigenvalue real parts of $M_n\in\textnormal{eGinUE}$ in a bounded interval $\Delta\subset\mathbb{R}$,
\begin{equation}\label{r24}
	\mathbb{E}_{n,\tau}\{\mathcal{R}_{n,\tau}(\Delta)\}=n\int_{\Delta}\varrho_{\tau}(x)\,\d x+o(n)\ \ \ \ \textnormal{as}\ \ n\rightarrow\infty,
\end{equation}
uniformly in $\Delta$, with $\varrho_{\tau}(x)$ equal to the density
\begin{equation}\label{r25}
	\varrho_{\tau}(x):=\frac{2}{\pi(1+\tau)}\sqrt{\bigg(1-\Big(\frac{x}{1+\tau}\Big)^2\bigg)_+},\ \ \ x\in\mathbb{R}.
\end{equation}
The asymptotic \eqref{r24} constitutes the limiting mean law for eigenvalue real parts at \textit{strong non-Hermiticity}, a term coined in \cite{FKS1} when $1-\tau>0$ uniformly in $n$. In this limit we are studying eigenvalue correlations of a genuinely non-Hermitian model and it is therefore\footnote{Also because of the findings in \cite[page $558$]{FKS1} and in \cite[Theorem $1$]{ACV}.} less surprising that the bulk spacings of eigenvalue real parts are as in the GinUE, i.e. as in \eqref{r17}. In detail, we record the below result.
\begin{lem}\label{lem:2} Assume $\lambda_0\in\mathbb{C}$ belongs to the interior of the support $\mathbb{D}_{\tau}$ of the elliptic law with $\tau\in(0,1)$ fixed and suppose $\Delta\subset\mathbb{R}$ is a bounded interval of length $|\Delta|\in(0,\infty)$. Then
\begin{equation}\label{r26}
	\lim_{n\rightarrow\infty}E_{n,\tau}\left(\lambda_0+\frac{\Omega_{\Delta}}{n\varrho_{\tau}(\Re\lambda_0)}\right)=\e^{-|\Delta|},\ \ \ \ \ \ \Omega_{\Delta}=\Delta\times\mathbb{R}\subset\mathbb{R}^2\simeq\mathbb{C},
\end{equation}
where $\varrho_{\tau}(x)$ appeared in \eqref{r24}.
\end{lem}
As can be inferred from \eqref{r26} the $\tau$-dependence in the limiting gap statistics of bulk eigenvalue real parts at strong non-Hermiticity is trivial, and thus the same limit has no chance of interpolating between \eqref{r18} and \eqref{r11}. For this reason we turn our attention to the second direction: Two $\textnormal{(ii)}$, the limit of \textit{weak non-Hermiticity} when $\tau\uparrow 1$ at an explicit $n$-dependent rate. In this case our discussion starts from the below result that can be distilled from \cite{ACV}.
\begin{prop}[{\cite[Theorem $3$]{ACV}}]\label{prop:2} Let $\rho_{n,\tau}$ denoted the density of the averaged normalized counting measure $\mathbb{E}_{n,\tau}\{\frac{1}{n}\mathcal{N}_{n,\tau}(\cdot)\}$ in the \textnormal{eGinUE} with $\tau\in(0,1)$ and set $\tau_n:=1-\frac{\sigma}{n}\in(0,1)$ with $\sigma>0$. Then
\begin{equation}\label{r28}
	\lim_{n\rightarrow\infty}\rho_{n,\tau_n}(z)=0,
\end{equation}
uniformly in any compact subset of $(\mathbb{C}\setminus\mathbb{R})\times(0,\infty)\ni (z,\sigma)$ and
\begin{equation}\label{r29}
	\lim_{n\rightarrow\infty}\int_{-\infty}^{\infty}\rho_{n,\tau_n}(x+\im y)\,\d y=\varrho_1(x)\stackrel{\eqref{r25}}{=}\frac{1}{2\pi}\sqrt{(2-x^2)_+},
\end{equation}
uniformly in any compact subset of $(\mathbb{R}\setminus\{\pm 2\})\times(0,\infty)\ni (x,\sigma)$. In particular, by \cite[$(52)$]{ACV} and the dominated convergence theorem, the number $\mathcal{R}_{n,\tau_n}(\cdot)$ of eigenvalue real parts in the \textnormal{eGinUE} at weak non-Hermiticity obeys the mean law
\begin{equation*}
	\mathbb{E}_{n,\tau_n}\{\mathcal{R}_{n,\tau_n}(\Delta)\}=n\int_{\Delta}\varrho_1(x)\,\d x+o(n)\ \ \ \ \textnormal{as}\ \ n\rightarrow\infty,
\end{equation*}
uniformly in the bounded interval $\Delta\subset\mathbb{R}$, for any fixed $\sigma>0$. See \eqref{r25} for $\varrho_1=\varrho_1(x)$.
\end{prop}
By \eqref{r28} and \eqref{r29}, at weak non-Hermiticity, we observe a limiting scenario where the global eigenvalue distribution of a non-Hermitian random matrix is supported on the real line, yet local correlations still extend to the complex plane, see the upcoming Proposition \ref{theo:1}. This feature was first discovered in the eGinUE bulk by Fyodorov, Khoruzhenko and Sommers in theoretical physics, see \cite{FKS1,FKS2,FSK}. It was made rigorous in mathematics at the eGinUE edge in \cite{Ben} and in the eGinUE bulk, even for more general complex elliptic models, in \cite{ACV}. Our upcoming spacing analysis of eigenvalue real parts in the eGinUE bulk will rest on some of the results in \cite{ACV}, but we need more, in particular we need Fredholm determinant formul\ae\, for the generating functional, for gap probabilities and lastly for the limiting bulk real part spacing distribution. 
\begin{prop}\label{theo:1} Assume $\lambda_0\in\mathbb{R}$ belongs to the interior of the interval $[-2,2]$ and define
\begin{equation}\label{r30}
	\tau_n:=1-\frac{1}{n}\bigg(\frac{\sigma}{\varrho_1(\lambda_0)}\bigg)^2\in (0,1),
\end{equation}
with $\sigma>0$ and $\varrho_1=\varrho_1(x)$ as in \eqref{r25}. Then for any $\ell\in\mathbb{Z}_{\geq 1}$, with $p_{\ell}^{(n,\tau_n)}(z_1,\ldots,z_{\ell})$ as in \eqref{r19},
\begin{equation}\label{r31}
	\lim_{n\rightarrow\infty}\big(\varrho_1(\lambda_0)\sqrt{n}\,\big)^{-2\ell}p_{\ell}^{(n,\tau_n)}\left(\lambda_0+\frac{w_1}{n\varrho_1(\lambda_0)},\ldots,\lambda_0+\frac{w_{\ell}}{n\varrho_1(\lambda_0)}\right)=\det\big[K_{\sin}^{\sigma}(w_j,w_k)\big]_{j,k=1}^{\ell},
\end{equation}
uniformly in $(w_1,\ldots,w_{\ell})$ chosen from any compact subset in $\mathbb{C}^{\ell}$ and uniformly in $\sigma\in(0,\infty)$ also chosen from any compact subset. Here, using $z_k:=x_k+\im y_k\in\mathbb{C}$,
\begin{equation}\label{r32}
	K_{\sin}^{\sigma}(z_1,z_2):=\frac{1}{\sigma\sqrt{\pi}}\exp\left[-\frac{y_1^2+y_2^2}{2\sigma^2}\right]\frac{1}{2\pi}\int_{-\pi}^{\pi}\e^{-(\sigma u)^2}\cos\big(u(z_1-\overline{z_2})\big)\,\d u.
\end{equation}
Next, the limiting generating functional of \eqref{r22} at weak non-Hermiticity, i.e.
\begin{equation*}
	E_{\sigma}[\phi]:=\lim_{n\rightarrow\infty}E_{n,\tau_n}[\phi_n],
\end{equation*}
where $\phi_n(z):=\phi((z-\lambda_0)n\varrho_1(\lambda_0)\big)$ with $\phi:\mathbb{C}\rightarrow\mathbb{C}$ bounded and of compact support $J=\textnormal{supp}\,\phi\subset\mathbb{C}$, equals
\begin{equation}\label{r33}
	E_{\sigma}[\phi]=1+\sum_{\ell=1}^{\infty}\frac{(-1)^{\ell}}{\ell!}\int_{J^{\ell}}\phi(z_1)\cdots\phi(z_{\ell})\det\big[K_{\sin}^{\sigma}(z_j,z_k)\big]_{j,k=1}^{\ell}\d^2 z_1\cdots \d^2 z_{\ell},
\end{equation}
for any fixed $\sigma>0$, and so the limiting gap probability
\begin{equation}\label{r34}
	E_{\sigma}(\Omega):=\lim_{n\rightarrow\infty}E_{n,\tau_n}\left(\lambda_0+\frac{\Omega}{n\varrho_1(\lambda_0)}\right),\ \ \ \ \Omega\subset\mathbb{C}\ \ \textnormal{measurable and bounded}
\end{equation}
is given by $E_{\sigma}(\Omega)=E_{\sigma}[\chi_{\Omega}]$ in terms of the indicator function $\chi_{\Omega}$ on $\Omega$.
\end{prop}
The limit \eqref{r31} for the marginal density is a special case of \cite[Theorem $3$(b)]{ACV}, after redefining our $\sigma>0$ to align with $\alpha>0$ in \cite{ACV}. Note that \eqref{r31} was first derived in \cite[$(51),(52)$]{FSK} using saddle point methods without details. To obtain the Fredholm determinant identity \eqref{r33}, and in turn the result for \eqref{r34}, one uses \eqref{r22} and Hadamard's inequality. The proof is easy since $\Omega\subset\mathbb{C}$ is bounded in \eqref{r34}. However, our workings also show that the result for \eqref{r34} remains valid for certain unbounded $\Omega\subset\mathbb{C}$, precisely as needed in our study of eigenvalue real part spacings.
\begin{prop}\label{prop:3} Assume $\lambda_0\in\mathbb{R}$ belongs to the interior of the interval $[-2,2]$ and define $\tau_n$ as in \eqref{r30}. Then the limiting \textnormal{eGinUE} gap probability
\begin{equation}\label{cool}
	E_{\sigma}(s):=\lim_{n\rightarrow\infty}E_{n,\tau_n}\left(\lambda_0+\frac{\Omega_s}{n\varrho_1(\lambda_0)}\right),\ \ \ \Omega_s:=(0,s)\times\mathbb{R}\subset\mathbb{R}^2\simeq\mathbb{C}
\end{equation}
exists for all fixed $s,\sigma>0$ and equals $E_{\sigma}(s)=E_{\sigma}[\chi_{\Omega_s}]$, compare \eqref{r33}.
\end{prop}
Access to the gap probability in the vertical strip $\Omega_s$, as established in Proposition \ref{prop:3}, yields access to the bulk spacing distribution of eigenvalue real parts at weak non-Hermiticity. Indeed, we have the following generalization of \eqref{r10},\eqref{r11} and Corollary \ref{cor:1} in the eGinUE at weak non-Hermiticity. The same constitutes our first main result.
\begin{theo}\label{theo:2} Assume $\lambda_0\in(-2,2)$ and define $\tau_n$ as in \eqref{r30}. Let $q_n^{\sigma}(s)$ denote the conditional probability that the distance between an eigenvalue real part at $\lambda_0$ and its nearest real part neighbor to the right is bigger than $s/(n\varrho_1(\lambda_0))$ for large $n$ and fixed $s,\sigma>0$, under the law \eqref{r1} with $\tau=\tau_n\in(0,1)$ as in \eqref{r30}, i.e.
\begin{equation*}
	q_n^{\sigma}(s):=\lim_{h\downarrow 0}\mathbb{P}_{n,\tau_n}\left\{\textnormal{no eigenvalue r.ps. in}\ \left(\lambda_0,\lambda_0+\frac{s}{n\varrho_1(\lambda_0)}\right]\,\bigg|\,\textnormal{an eigenvalue r.p. in}\ \left(\lambda_0-\frac{h}{n\varrho_1(\lambda_0)},\lambda_0\right]\right\}.
\end{equation*}
Then the corresponding limiting spacing distribution
\begin{equation*}
	\wp(s,\sigma):=1-\lim_{n\rightarrow\infty}q_n^{\sigma}(s),\ \ \ \ s,\sigma>0\ \ \textnormal{fixed},
\end{equation*}
is given by the following generalization of \eqref{r11} and \eqref{r18},
\begin{equation}\label{r35}
	\wp(s,\sigma)=\int_0^s\frac{\d^2}{\d u^2}\prod_{j=1}^{\infty}\Big(1-\omega_j\big(K_{\sin}^{\sigma},\Omega_u\big)\Big)\,\d u,\ \ \ \ \ \ \Omega_u=(0,u)\times\mathbb{R}\subset\mathbb{R}^2\simeq\mathbb{C}.
\end{equation}
Here, $\omega_j(K_{\sin}^{\sigma},\Omega)$ are the non-zero eigenvalues, counted according to their algebraic multiplicities, of the trace class integral operator $K_{\sin}^{\sigma}:L^2(\Omega)\rightarrow L^2(\Omega)$ with
\begin{equation*}
	(K_{\sin}^{\sigma}f)(z):=\int_{\Omega}K_{\sin}^{\sigma}(z,w)f(w)\,\d^2 w,\ \ \ \ \Omega\subset\mathbb{C},
\end{equation*}
see \eqref{r32}. As $s\downarrow 0$ and $\sigma>0$ is fixed,
\begin{equation*}
	\frac{\partial}{\partial s}\wp(s,\sigma)=\frac{1}{8}(\pi s)^2\int_{0}^2\int_{0}^2(x-y)^2\e^{-\frac{1}{2}(\pi\sigma)^2(x-y)^2}\,\d x\,\d y+\mathcal{O}\big(s^4\big).
\end{equation*}
\end{theo}
The limiting spacing distribution function \eqref{r35} is our sought after generalization of the Gaudin-Mehta distribution function \eqref{r11} in the eGinUE at weak non-Hermiticity. An easy argument, see Section \ref{sec5} below in Corollary \ref{deg1} and \ref{deg2}, shows that $\wp(s,\sigma)$ indeed degenerates to $F_1(s)$ in \eqref{r11}, once $\sigma\downarrow 0$ as we move to the GUE, and to $F_0(s)$ in \eqref{r18}, once $\sigma\rightarrow\infty$ as we move to the GinUE, for fixed $s>0$. The same interpolation phenomenon is visualized in Figure \ref{figure3} for selected choices of $\sigma\in[0,\infty)$. 
\begin{center}
\begin{figure}[tbh]
\resizebox{0.456\textwidth}{!}{\includegraphics{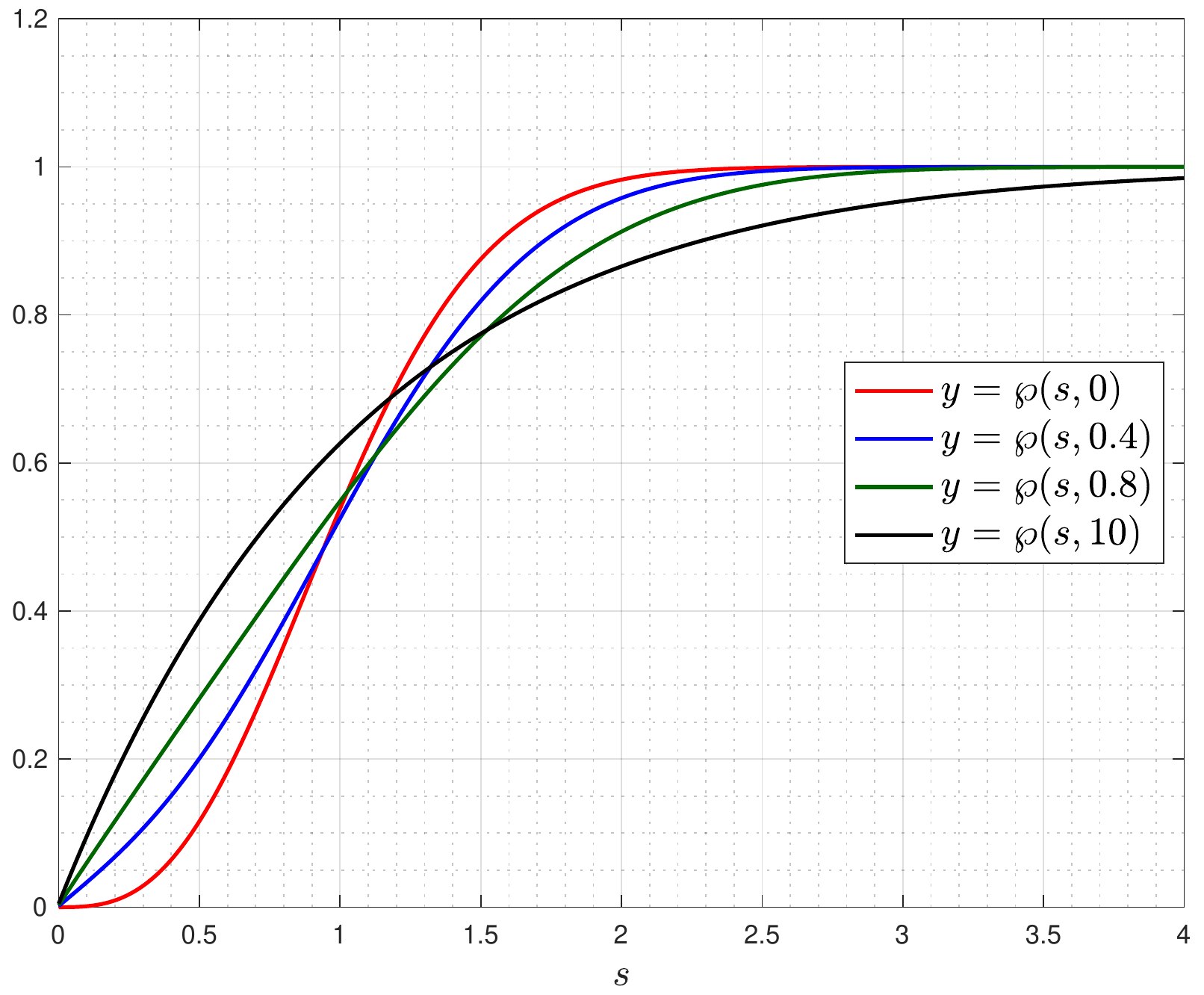}}\ \ \ \ \ \ \ \ \resizebox{0.455\textwidth}{!}{\includegraphics{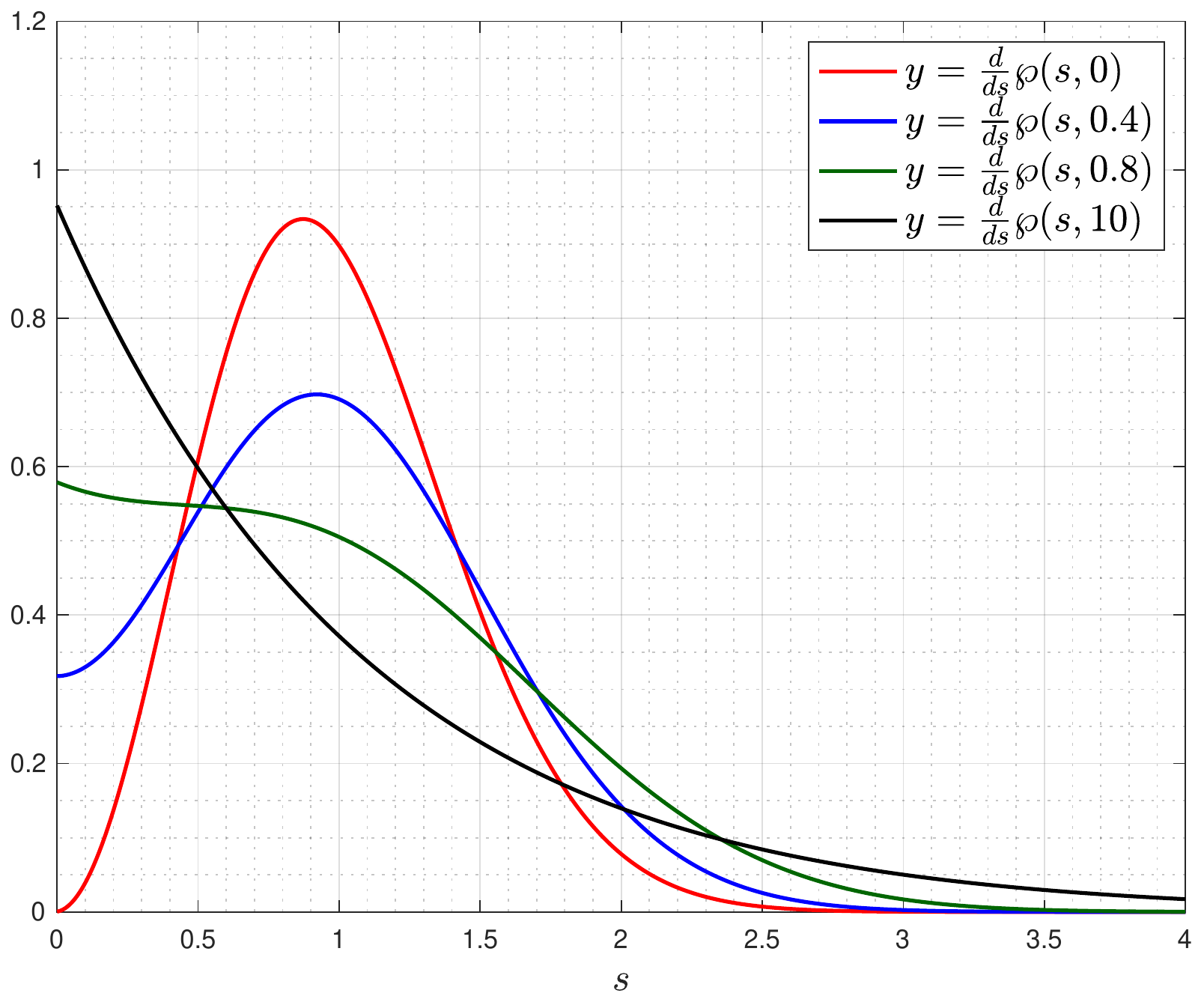}}
\caption{The generalized Gaudin-Mehta bulk spacing distribution \eqref{r35} in the eGinUE at weak non-Hermiticity. The plots were generated in MATLAB with $m=80$ quadrature points using the Nystr\"om method with Gauss-Legendre quadrature. On the left the distribution function, on the right the density function, both for varying $\sigma\in[0,\infty)$.}
\label{figure3}
\end{figure}
\end{center} 

Moving forward, we are now prepared to formulate the aforementioned integrable system that achieves the interpolation between $D(\frac{t}{2})$ in \eqref{r11}, and the exponential in \eqref{r18}. This integrable system will underwrite the Fredholm determinant appearing in the integrand of  \eqref{r35}. For brevity, we define for any $t,\sigma>0$,
\begin{equation}\label{r36}
	D_{\sigma}(t):=\prod_{j=1}^{\infty}\Big(1-\omega_j\big(K_{\sin}^{\sigma},(-t,t)\times\mathbb{R}\big)\Big),\ \ \ \ \ \ \ (-t,t)\times\mathbb{R}\subset\mathbb{R}^2\simeq\mathbb{C},
\end{equation}
and state our second main result.
\begin{theo}\label{theo:3} Set
\begin{equation*}
	w_{\alpha}(y):=\Phi\big(\alpha(y+1)\big)-\Phi\big(\alpha(y-1)\big),\ \ y\in\mathbb{R},\ \alpha>0;\ \ \ \ \ \ \ \ \ \Phi(z):=1-\frac{1}{2}\textnormal{erfc}(z),\ \ z\in\mathbb{C},
\end{equation*}
in terms of the complementary error function $w=\textnormal{erfc}(z)$, see \cite[$7.2.2$]{NIST}. Then for all $t,\sigma>0$, 
\begin{equation}\label{r37}
	D_{\sigma}(t)=\exp\left[-2t\int_0^{\infty}w_{\alpha}(y)\,\d y-\int_0^t(t-s)\left\{\int_0^{\infty}u_{\alpha}(s,y)\,\d w_{\alpha}(y)\right\}^2\d s\right]\Bigg|_{\alpha=t\sigma^{-1}}
\end{equation}
with $\d w_{\alpha}(y)\equiv w_{\alpha}'(y)\d y$ and where $u_{\alpha}=u_{\alpha}(t,y)$ solves the integro-differential equation
\begin{align}\label{r38}
	\left[\frac{\partial^2}{\partial t^2}\big(tu_{\alpha}(t,y)\big)+4\pi^2y^2tu_{\alpha}(t,y)\right]^2=4&\left[\int_0^{\infty}u_{\alpha}(t,y)\,\d w_{\alpha}(y)\right]^2\nonumber\\
	&\ \ \ \times\left[\left(\frac{\partial}{\partial t}\big(tu_{\alpha}(t,y)\big)\right)^2+4\pi^2y^2\big(tu_{\alpha}(t,y)\big)^2\right],
\end{align}
with boundary condition $u_{\alpha}(t,y)\sim 2y$ as $t\downarrow 0$ for any fixed $y,\alpha>0$.
\end{theo}
Identity \eqref{r37},\eqref{r38} is a special case of a broader result for \textit{finite-temperature} sine kernel Fredholm determinants of the type
\begin{equation*}
	F(t):=\prod_{j=1}^{\infty}\Big(1-\omega_j\big(K,(-t,t)\big)\Big),\ \ \ \ t>0,
\end{equation*}
as established in Section \ref{sec6} below, see Theorem \ref{main}. Here, $K:L^2(\Delta)\rightarrow L^2(\Delta)$ is the integral operator with kernel
\begin{equation}\label{r39}
	K(x,y):=\int_0^{\infty}\cos\big(\pi(x-y)z\big)w(z)\,\d z,\ \ \ \ x,y\in\Delta\subset\mathbb{R},
\end{equation}
depending on a smooth, even weight function $w:\mathbb{R}\rightarrow[0,1)$ such that $\lim_{|z|\rightarrow\infty}w(z)=0$ exponentially fast, and where $\omega_j(K,\Delta)$ are its non-zero eigenvalues. Integral operators of the type \eqref{r39} first appeared in the analysis of the emptiness formation probability in the one-dimensional Bose gas by Its, Izergin, Korepin and Slavnov \cite{IIKS} for a Fermi weight, see \cite[$(1.4),(1.5)$]{IIKS}. The same authors investigated the corresponding $F(t)$ in the framework of operator-valued Riemann-Hilbert problems, however a single scalar integro-differential equation as in \eqref{r38} was not found in \cite{IIKS}. The next appearance of a kernel of the type \eqref{r39} occurred in Johansson's work \cite[$(38)$]{Joh} on the Moshe-Neuberger-Shapiro model, in a suitable grand canonical scaling limit. But no focus was placed on integrable systems in \cite{Joh} and the same goes for the works \cite{DDMS,LW} where examples of $F(t)$ reappeared. Consequently, to the best of our knowledge, \eqref{r37} and \eqref{r38} mark the first time that the Fredholm determinant of a finite-temperature sine kernel \eqref{r39} has been related to Painlev\'e special function theory.
\begin{rem} The kernel \eqref{r32} is translation invariant in the horizontal direction and so 
\begin{equation*}
	D_{\sigma}(t)=\prod_{j=1}^{\infty}\Big(1-\omega_j\big(K_{\sin}^{\sigma},\Omega_{2t}\big)\Big),\ \ \ \ \ \Omega_{2t}=(0,2t)\times\mathbb{R}\subset\mathbb{R}^2\simeq\mathbb{C},
\end{equation*}
which shows that \eqref{r37} and \eqref{r38} underwrite the limiting distribution function $\wp$ in \eqref{r35}. 
\end{rem}
\begin{rem} Observing the pointwise limit
\begin{equation*}
	\lim_{\alpha\rightarrow\infty}w_{\alpha}(y)=\begin{cases}1,&y\in(-1,1)\\ 0,&y\in(-\infty,-1)\cup(1,\infty)
	\end{cases},
\end{equation*} 
of the weight in Theorem \ref{theo:3}, expression \eqref{r37} becomes \eqref{r12} as $\sigma\downarrow 0$ and $t>0$ is fixed, with $u_{\infty}(s,1)$ in \eqref{r38} formally equal to $u(s)$ in \eqref{r13}. It is in this way that \eqref{r37} and \eqref{r38} generalize \eqref{r12} and \eqref{r13}. 
\end{rem}


\subsection{Methodology and outline of paper}
The study of eigenvalue real parts in non-Hermitian random matrix ensembles was initiated by May in his landmark paper \cite{Ma} on the stability of large ecological systems. Indeed, any quantitatively accurate description of the Gardner-Ashby-May transition \cite{GA} from sharp to unstable behavior, as the system's complexity increases, requires precise control over the statistical behavior of eigenvalue real parts in the underlying matrix model. On the level of the rightmost eigenvalue, that is the eigenvalue with largest real part, such control was achieved in \cite{Ben} for the eGinUE and in \cite{CESX} for the GinUE\footnote{The companion paper \cite{CESX2} of \cite{CESX} showed that the same control is universal within a suitable class of non-Hermitian matrix ensembles with iid entries.}. Moreover, the work of Bender \cite{Ben} was recently connected to the theory of integrable systems in \cite{BL} where it was shown that the limiting distribution function of the rightmost eigenvalue in the eGinUE at weak non-Hermiticity is expressible in terms of a distinguished integro-differential Painlev\'e-II transcendent. The same transcendent interpolates, with varying degree of non-Hermiticity, between the Hastings-McLeod Painlev\'e-II transcendent that underwrites the limiting statistical behavior of the rightmost eigenvalue in the GUE, cf. \cite{TW}, and the elementary exponential that upholds the limiting Gumbel law found in \cite{CESX} at the rightmost edge in the GinUE. In summary, the precise statistical description of the largest eigenvalue real part in the GUE, eGinUE and GinUE has been obtained in the papers \cite{TW,Ben,BL,CESX}. On the level of innermost eigenvalues, i.e. bulk eigenvalue real parts, much less was known in RMT literature up to this point. Save for the ubiquitous Gaudin-Mehta distribution \eqref{r11}, the Poisson gap distribution \eqref{r18} in the GinUE and the generalized Gaudin-Mehta distribution \eqref{r35} in the eGinUE at weak non-Hermiticity are novel results, and, as before at the rightmost edge, these results have now been connected to Painlev\'e special function theory through Theorem \ref{theo:3} and Corollaries \ref{deg1}, \ref{deg2}; nearly $25$ years after the first appearance of \eqref{r32} in \cite{FKS2,FSK}.\smallskip

In proving the limit laws \eqref{r15},\eqref{r17},\eqref{r18},\eqref{r24},\eqref{r26},\eqref{cool} and \eqref{r25} one faces a common technical hurdle in that all vertical eigenvalue coordinates need to be integrated out over $\mathbb{R}$, i.e. unlike for the GUE in \eqref{r2} or \eqref{r9}, one necessarily has to deal with unbounded integration domains when studying eigenvalue real parts in non-Hermitian ensembles. Our favorite tools for this obstacle will be the dominated convergence theorem and Fatou's lemma, so we will oftentimes first establish dominance bounds and pointwise limits before integrating. Once the limits have been established, we then proceed by relating them to integrable systems theory. This is trivial in case of \eqref{r17},\eqref{r18} and \eqref{r26}, but far from trivial in \eqref{r35}. Indeed, the underlying integral operator $K_{\sin}^{\sigma}$ in Theorem \ref{theo:2} acts on $L^2(\Omega)$ where $\Omega\subset\mathbb{C}$ is two-dimensional and so most classical tools in integrable systems theory are not applicable to study the Fredholm determinant in \eqref{r35}. Faced with this obstacle one might be tempted to follow the approach in \cite{BL} where a similar higher-dimensional problem occurred at the rightmost eGinUE edge at weak non-Hermiticity. In \cite{BL} the problem was bypassed through the use of trace identities and quite surprisingly, similar trace identities also hold in the eGinUE bulk at weak non-Hermiticity for the operator $K_{\sin}^{\sigma}$ with kernel \eqref{r32}. Namely, we have for any $t,\sigma>0$,
\begin{equation*}
	\tr_{L^2(J_t)}(K_{\sin}^{\sigma})^k=\tr_{L^2(J_t)}(M_{\sin}^{\sigma})^k\ \ \ \ \ \ \forall\,k\in\mathbb{Z}_{\geq 1}\ \ \textnormal{with}\ \ J_t:=(-t,t)\times\mathbb{R}\subset\mathbb{R}^2\simeq\mathbb{C},
\end{equation*}
where $M_{\sin}^{\sigma}:L^2(J_t)\rightarrow L^2(J_t)$ is trace class with
\begin{equation}\label{r40}
	(M_{\sin}^{\sigma}f)(z):=\int_{J_t}M_{\sin}^{\sigma}(z,w)f(w)\,\d^2 w,\ \ \ M_{\sin}^{\sigma}(z_1,z_2):=\frac{1}{\sqrt{\pi}}\e^{-\frac{1}{2}y_1^2}K_{\sin}(x_1+\sigma y_1,x_2+\sigma y_2)\e^{-\frac{1}{2}y_2^2},
\end{equation}
in terms of the ordinary sine kernel \eqref{r9},\eqref{h5a} and with $z_k=x_k+\im y_k\in\mathbb{C}$. The kernel structure \eqref{r40} also appeared at the rightmost eGinUE edge at weak non-Hermiticity with $K_{\sin}$ replaced by the GUE Airy kernel, see \cite[Proposition $2.3$]{BL}. In turn, the edge equivalent of \eqref{r40} served as starting point for the analysis in \cite{BL} as it identified the underlying integral operator as Hankel composition operator acting on $L^2(\Omega)$ with a suitable domain $\Omega\subset\mathbb{C}$, and so the techniques of \cite{Kra,Bo} became available in the analysis of the underlying Fredholm determinant. Unfortunately, \eqref{r40} does \textit{not} identify $M_{\sin}^{\sigma}$ as Hankel composition operator but instead as Wiener-Hopf type operator and it is therefore not clear how \eqref{r40} can be useful in the derivation of an integrable system for the Fredholm determinant of $K_{\sin}^{\sigma}$ on $L^2(J_t)$, say. For this reason we follow a different route in the analysis of $K_{\sin}^{\sigma}$, a route which relies on a suitable operator factorization and subsequent application of Sylvester's identity. Afterwards we employ a Fourier integral identity and what results is the identification of the Fredholm determinant of $K_{\sin}^{\sigma}$ on $L^2(J_t\subset\mathbb{R}^2)$ as Fredholm determinant of a suitable \textit{finite-temperature} sine kernel \eqref{r39} on $L^2((-t,t)\subset\mathbb{R})$. Equipped with the same dimensional reduction we then adapt parts of the original algebraic Tracy-Widom method \cite{TW0} to our needs and obtain the integro-differential connection \eqref{r37},\eqref{r38}.\smallskip

In more detail, the remaining sections of this paper are organized as follows: in Section \ref{sec2} we derive our horizontal spacing results in the GinUE as stated in \eqref{r14}, in Lemma \ref{lem:1} and Corollary \ref{cor:1}. Our workings rely on the integrable structure of the GinUE and first establish a Fredholm determinant representation for the gap probability in suitable unbounded vertical sectors at finite $n$. Once achieved, we then evaluate the same Fredholm determinant in the large $n$ limit by exploiting the Plemelj-Smithies formula and by estimating the underlying operator traces. With \eqref{r17} in place, formula \eqref{r18} follows by l'Hospital's rule and another operator trace estimation. Afterwards, in Section \ref{sec3}, we establish \eqref{r24} and Lemma \ref{lem:2}. Working in the eGinUE at strong non-Hermiticity requires us to rewrite the correlation kernel \eqref{r20} as double contour integral as in \cite{ACV} and to establish another Fredholm determinant formula for the gap probability \eqref{r26} in a vertical sector, at finite $n$. Subsequently, the limit in \eqref{r26} follows by proving that a suitable integral operator is asymptotically small in Hilbert-Schmidt norm. The same yields the desired result \eqref{r26} via \eqref{p37} and we note that this approach is inspired by the recent GinUE edge workings in \cite[$(20)$]{CESX}. We move to the eGinUE at weak non-Hermiticity in Section \ref{sec4} and first establish the Fredholm determinant formul\ae\,\eqref{r33} and \eqref{r34} for the gap probability on bounded domains $\Omega\subset\mathbb{C}$. The same formul\ae\,are subsequently shown to be valid for unbounded vertical strips, too, and the limiting gap probability \eqref{cool} is obtained by Hadamard's inequality and a suitable dominance estimate. In turn, \eqref{r35} follows by Taylor expansion while exploiting our previous dominance bounds, Hadamard's inequality and the translation invariance of \eqref{r32} in horizontal direction. Moving ahead, in Section \ref{sec5} we establish trace class and invertibility properties of $K_{\sin}^{\sigma}$ and derive two equivalent representations for its Fredholm determinant: one in \eqref{h3} which is particular amenable to numerical simulations, compare Figure \ref{figure3}, and which relies on an operator factorization and Sylvester's identity. The same representation is also useful in establishing the degenerations of $\wp(s,\sigma)$ as $\sigma\downarrow 0$ and $\sigma\rightarrow\infty$ in Corollary \ref{deg1} and \ref{deg2}. The second representation in \eqref{h8} for the Fredholm determinant of $K_{\sin}^{\sigma}$ identifies the same as a finite-temperature sine kernel determinant \eqref{r39} and using this representation we compute the integrable system \eqref{r37},\eqref{r38} in Section \ref{sec6}. In detail, we bring our finite-temperature sine kernel in generalized integrable form, see \eqref{i1}, and afterwards adapt the Tracy-Widom algebraic method \cite[Section VI]{TW0} to the general kernel \eqref{r39}. En route we obtain integro-differential Gaudin-Mehta identities in \eqref{i17},\eqref{i18},\eqref{i19} and subsequently Theorem \ref{main} for the general kernel family \eqref{r39}. A special case of this general result yields \eqref{r37},\eqref{r38} and completes our paper.
\begin{rem} We choose to work with the algebraic Tracy-Widom method in the derivation of Theorem \ref{main} for ease of presentation. The same approach was used in \cite{ACQ} in the derivation of the aforementioned integro-differential Painlev\'e-II connection and only later on it was shown in \cite{Bo0,BCT} that operator-valued Riemann-Hilbert techniques yield the same result. We expect that the operator-valued Riemann-Hilbert methodology can also be applied to \eqref{h9},\eqref{i1} and it would yield \eqref{r37},\eqref{r38} through the compatibility of a suitable operator-valued Lax pair.
\end{rem}
\begin{rem} A third approach to the Fredholm determinant of \eqref{r39} emerges from the trace identities
\begin{equation*}
	\tr_{L^2(-t,t)}K^m=\tr_{L^2(\mathbb{R})}K_t^m,\ \ \ m\in\mathbb{Z}_{\geq 1},\ \ t>0;
\end{equation*}
where $K_t:L^2(\mathbb{R})\rightarrow L^2(\mathbb{R})$ acts as
\begin{equation}\label{r41}	
	f\mapsto (K_tf)(x):=\int_{-\infty}^{\infty}\left[\sqrt{w(x)}\,\frac{\sin(\pi(x-y)t)}{\pi(x-y)}\sqrt{w(y)}\right]f(y)\,\d y\equiv\int_{-\infty}^{\infty}K_t(x,y)f(y)\,\d y.
\end{equation}
The kernel $K_t(x,y)$ in \eqref{r41} is of standard integrable type and it was precisely an operator of this form that appeared in the work of Its, Izergin, Korepin and Slavnov on the one-dimensional Bose gas, see \cite[$(1.4)$]{IIKS}. Being of standard integrable type, one can analyze the Fredholm determinant of $K_t$ on $L^2(\mathbb{R})$ via matrix-valued Riemann-Hilbert techniques as done in \cite[Section $5$]{IIKS}, where it was shown that the associated Fredholm determinant, for a particular weight, solves a PDE \cite[$(5.15)$]{IIKS}\footnote{the weight in \cite[$(1.5)$]{IIKS} depends on an auxiliary parameter $\beta$ and the PDE is formulated in $(t,\beta)$.}. It is reasonable to expect that a matrix-valued Riemann-Hilbert analysis of $K_t$ will bring our integro-differential result \eqref{i24} and the PDE result of \cite{IIKS} under one umbrella, just as it happened for the finite-temperature Airy kernel in \cite{CCR}. Lastly, the matrix-valued Riemann-Hilbert approach is at present well amenable to asymptotics, see \cite[Section III]{IIKS} for some elements thereof, although the weight \eqref{h14} in the \textnormal{eGinUE} will pose additional difficulties because of the non-trivial zero locus of $1-w_{\alpha}(z)$.
\end{rem}

\section{Horizontal spacings in the GinUE: proof of \eqref{r14}, Lemma \ref{lem:1} and Corollary \ref{cor:1}}\label{sec2}
The following basic fact can be found in \cite{Gi}, in \cite[Chapter $15$]{M,F1} or in \cite[Theorem $4.3.10$]{HKPV}, compare also Remark \ref{rem:1}: the eigenvalues $\{\lambda_j(M_n)\}_{j=1}^n\subset\mathbb{C}$ of $M_n\in\textnormal{GinUE}$, see \eqref{r1} with $\tau=0$, form a determinantal point process in the complex plane with kernel
\begin{equation}\label{p1}
	K_{n,0}(z,w)=\frac{n}{\pi}\e^{-\frac{n}{2}|z|^2-\frac{n}{2}|w|^2}\sum_{\ell=0}^{n-1}\frac{n^{\ell}}{\ell!}(z\overline{w})^{\ell},\ \ \ \ z,w\in\mathbb{C},
\end{equation}
with respect to the flat measure $\d^2 z$. As such, using Andr\'eief's and the Cauchy-Binet identity, the generating functional
\begin{equation*}
	E_{n,0}[\phi]:=\mathbb{E}_{n,0}\left\{\prod_{\ell=1}^n\Big(1-\phi\big(\lambda_{\ell}(M_n)\big)\Big)\right\},\ \ \ \ M_n\in\textnormal{GinUE},
\end{equation*}
equals
\begin{equation}\label{p2}
	E_{n,0}[\phi]=1+\sum_{\ell=1}^n\frac{(-1)^{\ell}}{\ell!}\int_{J^{\ell}}\phi(z_1)\cdots\phi(z_{\ell})\det\big[K_{n,0}(z_j,z_k)\big]_{j,k=1}^{\ell}\d^2z_1\cdots\d^2z_{\ell},
\end{equation}
where the support $J=\textnormal{supp}\,\phi\subset\mathbb{C}$ of the bounded test function $\phi:\mathbb{C}\rightarrow\mathbb{C}$ is assumed to be compact. Moving ahead, we let
\begin{equation*}
	\rho_{n,0}(z):=\frac{1}{n}K_{n,0}(z,z),\ \ \ z\in\mathbb{C},
\end{equation*}
denote the density of the averaged normalized counting measure and now prove \eqref{r14} and Lemma \ref{lem:1}.
\begin{proof}[Proof of \eqref{r14} and Lemma \ref{lem:1}] Let $E_{n,0}(\Omega)$ denote the probability, under \eqref{r1} with $\tau=0$, that $M_n\in\textnormal{GinUE}$ has no eigenvalues in $\Omega\subset\mathbb{C}$. Consider the rectangle $\Omega_{ab}^{cd}:=(a,b)\times(c,d)\subset\mathbb{R}^2\simeq\mathbb{C}$ with $-\infty\leq a<b\leq\infty$ and $-\infty\leq c<d\leq\infty$. We begin by establishing the Fredholm determinant identity
\begin{equation}\label{p3}
	E_{n,0}\left(\lambda_0+\frac{\Omega_{ab}^{cd}}{n}\right)=1+\sum_{\ell=1}^n\frac{(-1)^{\ell}}{\ell!}\int_{(\Omega_{ab}^{cd})^{\ell}}\det\left[n^{-2}K_{n,0}\left(\lambda_0+\frac{w_j}{n},\lambda_0+\frac{w_k}{n}\right)\right]_{j,k=1}^{\ell}\d^2 w_1\cdots\d^2 w_{\ell}
\end{equation}
for any finite $n\in\mathbb{Z}_{\geq 1}$ and any $\lambda_0\in\mathbb{C}$. Indeed, if $a,b,c,d$ are all finite, then \eqref{p3} is a consequence of \eqref{p2} and the inclusion-exclusion principle, so we first establish convergence of the right hand side of \eqref{p3}, for any finite $n\in\mathbb{Z}_{\geq 1}$, if at least one of $a,b,c,d$ is infinite. We will achieve this by using a version of Hadamard's inequality. Namely, $B=[B_{jk}]_{j,k=1}^{\ell}\in\mathbb{C}^{\ell\times\ell}$ with
\begin{equation*}
	B_{jk}:=n^{-2}K_{n,0}\left(\lambda_0+\frac{w_j}{n},\lambda_0+\frac{w_k}{n}\right)
\end{equation*}
is a positive-definite matrix, for it is Hermitian, see \eqref{p1}, and we have for all $(t_1,\ldots,t_{\ell})\in\mathbb{C}^{\ell}$,
\begin{equation*}
	\sum_{j,k=1}^{\ell}B_{jk}t_j\overline{t_k}=\frac{1}{n\pi}\sum_{m=0}^{n-1}\frac{n^m}{m!}\left|\sum_{j=1}^{\ell}\exp\left[-\frac{n}{2}\Big|\lambda_0+\frac{w_j}{n}\Big|^2\right]\left(\lambda_0+\frac{w_j}{n}\right)^mt_j\right|^2\geq 0.
\end{equation*}
Hence by \cite[$(5.2.72)$]{PS},
\begin{equation}\label{p4}
	\det\big[B_{jk}\big]_{j,k=1}^{\ell}\leq\prod_{j=1}^{\ell}B_{jj}=\prod_{j=1}^{\ell}n^{-1}\rho_{n,0}\left(\lambda_0+\frac{w_j}{n}\right),
\end{equation}
with 
\begin{equation*}
	\rho_{n,0}\left(\lambda_0+\frac{w}{n}\right)\stackrel{\eqref{p1}}{=}\frac{1}{\pi}\exp\left[-n\bigg|\lambda_0+\frac{w}{n}\bigg|^2\right]\sum_{\ell=0}^{n-1}\frac{n^{\ell}}{\ell!}\left|\lambda_0+\frac{w}{n}\right|^{2\ell},\ \ \ w\in\Omega_{ab}^{cd}.
\end{equation*}
However, for any $z,w\in\mathbb{C}$, by the binomial theorem and non-negativity,
\begin{equation*}
	\sum_{\ell=0}^{n-1}\frac{1}{\ell!}|z+w|^{2\ell}\leq\sum_{\ell=0}^{n-1}\frac{1}{\ell!}\big(\Re(z+w)\big)^{2\ell}\sum_{\ell=0}^{n-1}\frac{1}{\ell!}\big(\Im(z+w)\big)^{2\ell},
\end{equation*}
and so
\begin{equation*}
	\rho_{n,0}\left(\lambda_0+\frac{w}{n}\right)\leq f_n\bigg(\Re\Big(\lambda_0+\frac{w}{n}\Big)\bigg)f_n\bigg(\Im\Big(\lambda_0+\frac{w}{n}\Big)\bigg);\ \ \ \ \ \ f_n(x):=\frac{1}{\sqrt{\pi}}\e^{-nx^2}\sum_{\ell=0}^{n-1}\frac{n^{\ell}}{\ell!}x^{2\ell},\ \ x\in\mathbb{R},
\end{equation*}
with $\mathbb{R}\ni x\mapsto f_n(x)\in L^1(\mathbb{R})$ for all $n\in\mathbb{Z}_{\geq 1}$. Consequently, for any $n,\ell\in\mathbb{Z}_{\geq 1}$,
\begin{align*}
	\int_{(\Omega_{ab}^{cd})^{\ell}}\det\big[B_{jk}\big]_{j,k=1}^{\ell}&\,\d^2 w_1\cdots\d^2 w_{\ell}\stackrel{\eqref{p4}}{\leq}\left[\frac{1}{n}\int_{\Omega_{ab}^{cd}}\rho_{n,0}\left(\lambda_0+\frac{w}{n}\right)\,\d^2 w\right]^{\ell}\\
	&\,\leq\left[\frac{1}{n}\int_a^bf_n\Big(\Re\lambda_0+\frac{x}{n}\Big)\,\d x\,\int_c^df_n\Big(\Im\lambda_0+\frac{y}{n}\Big)\,\d y\right]^{\ell}<\infty.
\end{align*}
This proves that the right hand side of \eqref{p3} is well-defined for all $-\infty\leq a<b\leq\infty,-\infty\leq c<d\leq\infty$ and so by continuity of $\mathbb{P}_{n,0}$, equality in \eqref{p3} follows. In particular, we have the exact identity
\begin{equation}\label{p5}
	\mathbb{E}_{n,0}\{\mathcal{R}_n(\Delta)\}=n\int_{\Omega_{\Delta}}\rho_{n,0}(z)\,\d^2z,\ \ \ \ \ \Omega_{\Delta}=\Delta\times\mathbb{R}\subset\mathbb{R}^2\simeq\mathbb{C},
\end{equation}
for the average number of eigenvalue real parts of $M_n\in\textnormal{GinUE}$ in a bounded interval $\Delta\subset\mathbb{R}$, which will be used to establish \eqref{r14} below. Next, we note that
\begin{equation*}
	\Omega_{ab}^{cd}\times\Omega_{ab}^{cd}\ni (z,w)\mapsto n^{-2}K_{n,0}\left(\lambda_0+\frac{z}{n},\lambda_0+\frac{w}{n}\right)
\end{equation*}
is continuous for any $n\in\mathbb{Z}_{\geq 1}$ and $\lambda_0\in\mathbb{C}$. Moreover, the integral operator $K_n:L^2(\Omega_{ab}^{cd})\rightarrow L^2(\Omega_{ab}^{cd})$ with
\begin{equation*}
	(K_nf)(z):=\int_{\Omega_{ab}^{cd}}n^{-2}K_{n,0}\left(\lambda_0+\frac{z}{n},\lambda_0+\frac{w}{n}\right)f(w)\,\d^2 w,\ \ \ \ \ \ z\in\Omega_{ab}^{cd},
\end{equation*}
is well-defined, since $|K_{n,0}(z,w)|\leq\sqrt{K_{n,0}(z,z)}\sqrt{K_{n,0}(w,w)},z,w\in\mathbb{C}$ for the reproducing kernel, and has rank at most $n$, so by \cite[Chapter I, $(1.8),(6.3)$]{GGK} and \eqref{p3},
\begin{equation}\label{p6}
	E_{n,0}\left(\lambda_0+\frac{\Omega_{ab}^{cd}}{n}\right)=\prod_j\big(1-\omega_j(K_n,\Omega_{ab}^{cd})\big),
\end{equation}
where $\omega_j(K_n,\Omega_{ab}^{cd})$ are the finitely many eigenvalues, counted according to their algebraic multiplicities, of $K_n$ when restricted to the subspace $M_{K_n}$ in the decomposition \cite[Chapter I, $(1.1)$]{GGK}. Noting that $\mathbb{C}\ni\gamma\mapsto\prod_j(1-\gamma\omega_j(K_n,\Omega_{ab}^{cd}))$ is a polynomial and that for $|\gamma|$ small
\begin{equation}\label{p7}
	\prod_{j}\big(1-\gamma\omega_j(K_n,\Omega_{ab}^{cd})\big)=\exp\left[-\sum_{k=1}^{\infty}\frac{\gamma^k}{k}\tr_{L^2(\Omega_{ab}^{cd})}K_n^k\right],
\end{equation}
by \cite[Chapter I, Theorem $3.3$]{GGK}, we proceed by estimating the operator traces that occur in the right hand side of \eqref{p7} for large $n$, when $(a,b)=\Delta/\varrho_0(\Re\lambda_0)=:\Delta_0$ with fixed length $|\Delta|\in(0,\infty)$ and $c=-\infty,d=\infty$, i.e. when $\Omega_{ab}^{cd}=\Omega_{\Delta}/\varrho_0(\Re\lambda_0)=:\Omega_0$. First, by \eqref{p1},
\begin{equation*}
	\tr_{L^2(\Omega_0)}K_n=\frac{1}{\pi}\int_{\Delta_0}\int_{-\infty}^{\infty}Q\bigg(n,n\left[\bigg(\Re\lambda_0+\frac{x}{n}\bigg)^2+y^2\right]\bigg)\d y\,\d x
\end{equation*}
in terms of the normalized incomplete gamma function, cf. \cite[$(8.2.4)$]{NIST},
\begin{equation}\label{p8}
	Q(a,z):=\frac{1}{\Gamma(a)}\int_z^{\infty}t^{a-1}\e^{-t}\,\d t,\ \ \ \Re a>0,\ z\in\mathbb{C}.
\end{equation}
A residue calculation, see also \cite[Proposition $1,2$]{BG}, readily verifies that for $x\in\mathbb{R}\setminus\{\pm 1\}$ and $n\in\mathbb{Z}_{\geq 1}$,
\begin{equation*}
	Q(n,nx)=\e^{-nx}\sum_{\ell=0}^{n-1}\frac{n^{\ell}}{\ell!}x^{\ell}=\frac{\im}{2\pi}\e^{-nx}x^n\ointctrclockwise_{|z|=1}\frac{\e^{nz}}{z-x}\frac{\d z}{z^n}+\begin{cases}1,&|x|<1\\ 0,&|x|>1\end{cases},
\end{equation*}
and thus
\begin{equation}\label{p9}
	\big|Q(n,nx)-\chi_{[0,1)}(x)\big|\leq|1-x|^{-1}\e^{-n(x-\ln x-1)},\ \ \ \ \ x\geq 0,\ x\neq 1.
\end{equation}
Consequently,
\begin{equation}\label{p10}
	Q(n,nx)\leq 1\ \ \ \textnormal{and}\ \ \ \frac{\d}{\d x}Q(n,nx)\leq 0\ \ \ \ \forall\,x\geq 0;\ \ \ \ \ \ \ \lim_{n\rightarrow\infty}Q(n,nx)=\begin{cases}1,&x\in[0,1)\\ 0,&x\in(1,\infty)\end{cases},
\end{equation}
which motivates the decomposition
\begin{align*}
	\tr_{L^2(\Omega_0)}K_n=\frac{1}{\pi}\int_{\Delta_0}\left[\int_{-2}^2+\int_{\mathbb{R}\setminus[-2,2]}\right]Q\bigg(n,n\left[\bigg(\Re\lambda_0+\frac{x}{n}\bigg)^2+y^2\right]\bigg)\d y\,\d x.
\end{align*}
For in the first integral we can use dominated convergence: as $n\rightarrow\infty$, uniformly in $\Delta$, noting $|\lambda_0|<1$,
\begin{equation*}
	\frac{1}{\pi}\int_{\Delta_0}\int_{-2}^2Q\bigg(n,n\left[\bigg(\Re\lambda_0+\frac{x}{n}\bigg)^2+y^2\right]\bigg)\d y\,\d x\stackrel{\eqref{p10}}{=}\frac{1}{\pi}\int_{\Delta_0}\int_{(\Re\lambda_0)^2+y^2<1}\d y\,\d x+o(1)=|\Delta|+o(1),
\end{equation*}
and in the second integral we can employ \eqref{p9},
\begin{align*}
	\frac{1}{\pi}\int_{\Delta_0}&\int_{\mathbb{R}\setminus[-2,2]}Q\bigg(n,n\left[\bigg(\Re\lambda_0+\frac{x}{n}\bigg)^2+y^2\right]\bigg)\d y\,\d x\stackrel{\eqref{p10}}{\leq}\frac{2}{\pi}|\Delta_0|\int_2^{\infty}Q\big(n,ny^2\big)\,\d y\\
	&\leq\frac{2}{3\pi}|\Delta_0|\int_2^{\infty}\e^{-n(y^2-2\ln y-1)}\,\d y\leq\frac{2}{9\pi n}|\Delta_0|\,\e^{-n(3-2\ln 2)},\ \ n\in\mathbb{Z}_{\geq 1}.
\end{align*}
Combined together, for any bounded interval $\Delta\subset\mathbb{R}$ of length $|\Delta|\in(0,\infty)$,
\begin{equation}\label{p11}
	\tr_{L^2(\Omega_0)}K_n=|\Delta|+o(1)\ \ \ \textnormal{as}\ \ n\rightarrow\infty.
\end{equation}
Before moving ahead with the higher order traces in \eqref{p7}, we point out that the argument leading to \eqref{p11} yields \eqref{r14}: indeed, by \eqref{p5},\eqref{p1} and \eqref{p8},
\begin{equation*}
	\mathbb{E}_{n,0}\{\mathcal{R}_n(\Delta)\}=n\int_{\Omega_{\Delta}}\rho_{n,0}(z)\,\d^2 z=\frac{n}{\pi}\int_{\Delta}\int_{-\infty}^{\infty}Q\Big(n,n\big[x^2+y^2\big]\Big)\,\d y\,\d x=n\int_{\Delta}\varrho_0(x)\,\d x+o(n),
\end{equation*}
as $n\rightarrow\infty$ uniformly in the bounded interval $\Delta\subset\mathbb{R}$. Back to \eqref{p7}, we proceed with the elementary estimate
\begin{align}
	&\bigg|\tr_{L^2(\Omega_0)}K_n^k\,\bigg|\leq\int_{\Omega_0^k}\prod_{j=1}^k\bigg|n^{-2}K_{n,0}\left(\lambda_0+\frac{z_j}{n},\lambda_0+\frac{z_{j+1}}{n}\right)\bigg|\,\d^2 z_1\cdots\d^2 z_k\nonumber\\
	&\leq\frac{1}{(\pi n)^k}\int_{\Omega_{\ast}^k}\left\{\prod_{j=1}^{k-1}\exp\left[-\frac{n}{2}\left(\bigg|\lambda_0+\frac{z_j}{n}\bigg|-\bigg|\lambda_0+\frac{z_{j+1}}{n}\bigg|\right)^2\right]\right\}\sqrt{Q\Big(n,n\bigg|\lambda_0+\frac{z_k}{n}\bigg|^2\Big)}\,\d^2 z_1\cdots\d^2z_k,\label{p12}
\end{align}
where $z_{k+1}\equiv z_1$ by convention, where $\Omega_{\ast}:=(-\tau,\tau)\times\mathbb{R}\supset \Omega_0$ with $\tau:=\sup\{|x|:\,x\in\Delta_0\}>0$ and where we have used the Cauchy-Schwarz inequality,
\begin{equation*}
	Q(n,n|zw|)\leq\e^{\frac{n}{2}(|z|-|w|)^2}\sqrt{Q(n,n|z|^2)}\sqrt{Q(n,n|w|^2)},\ \ \ z,w\in\mathbb{C},
\end{equation*}
as well as \eqref{p10}. Note that the right hand side of the inequality \eqref{p12} only depends on the modulus of $\lambda_0+z_j/n$, so we will evaluate the same via polar coordinates. First, for $\lambda_0=0$, the value of the angular integral for each $z_j=r_j\e^{\im\theta_j}$ is given by
\begin{equation*}
	\int_{\Pi_j}\d\theta_j=f\Big(\frac{r_j}{\tau}\Big),\ \ \ \ \ \ \ f(x):=\begin{cases}2\pi,&0\leq x\leq 1\\ 4\arcsin(\frac{1}{x}),&x>1\end{cases},
\end{equation*}
where $\Pi_j$ is the range of angles of the intersection of a circle of radius $r_j>0$ with $\Omega_{\ast}$. Hence,
\begin{align*}
\bigg|\tr_{L^2(\Omega_0)}K_n^k\,\bigg|\leq\frac{1}{\pi^k}\int_{\mathbb{R}_+^k}&\,\left\{\prod_{j=1}^{k-1}r_jf\bigg(\frac{r_j}{\tau}\sqrt{n}\bigg)\exp\left[-\frac{1}{2}(r_j-r_{j+1})^2\right]\right\}r_kf\bigg(\frac{r_k}{\tau}\sqrt{n}\bigg)\\
&\,\,\,\times\sqrt{Q(n,r_k^2)}\,\d r_1\cdots\d r_k,
\end{align*}
and since $xf(x)\leq 2\pi$ for all $x\geq 0$,
\begin{align}
	\bigg|\tr_{L^2(\Omega_0)}&\,K_n^k\,\bigg|\leq\bigg(\frac{2\tau}{\sqrt{n}}\bigg)^k\int_{\mathbb{R}_+^k}\left\{\prod_{j=1}^{k-1}\exp\left[-\frac{1}{2}(r_j-r_{j+1})^2\right]\right\}\sqrt{Q(n,r_k^2)}\,\d r_1\cdots\d r_k\nonumber\\
	\leq&\bigg(\frac{2\tau}{\sqrt{n}}\bigg)^k(2\pi)^{\frac{1}{2}(k-1)}\int_0^{\infty}\sqrt{Q(n,r_k^2)}\,\d r_k\stackrel{\eqref{p9}}{\leq}\bigg(\frac{2\tau}{\sqrt{n}}\bigg)^k(2\pi)^{\frac{1}{2}(k-1)}2\sqrt{n}\left[1+\frac{1}{n}\e^{-n(1-\ln 2)}\right],\label{p13}
\end{align}
valid for all $n,k\in\mathbb{Z}_{\geq 1}$. Thereafter, for $\Re\lambda_0>0$, we first change variables according to $w_j=\sqrt{n}\,\Re\lambda_0+z_j/\sqrt{n}$,
\begin{equation}\label{p14}
	\bigg|\tr_{L^2(\Omega_0)}K_n^k\,\bigg|\leq\frac{1}{\pi^k}\int_{(\sqrt{n}\,\Re\lambda_0+\Omega_{\ast}/\sqrt{n}\,)^k}\left\{\prod_{j=1}^{k-1}\exp\left[-\frac{1}{2}(|w_j|-|w_{j+1}|)^2\right]\right\}\sqrt{Q(n,|w_k|^2)}\,\d^2w_1\cdots\d^2 w_k,
\end{equation}
and perform the angular integral for each $w_j=r_j\e^{\im\theta_j}$ with $n\geq n_0$ sufficiently large,
\begin{equation*}
	\int_{\Pi_j}\d\theta_j=g_n(r_j).
\end{equation*}
Here, $\Pi_j$ denotes the range of angles of the intersection of a circle of radius $r_j>\sqrt{n}\,\Re\lambda_0-\tau/\sqrt{n}>0$ with the strip $\sqrt{n}\,\Re\lambda_0+\Omega_{\ast}/\sqrt{n}$, and the continuous non-negative function $x\mapsto g_n(x)$ equals
\begin{equation*}
	g_n(x):=2\begin{cases}\arcsin(\sqrt{1-x^{-2}(\sqrt{n}\,\Re\lambda_0-\tau/\sqrt{n}\,)^2}),&x\in A_n\smallskip\\
	\arcsin(x^{-1}(\sqrt{n}\,\Re\lambda_0+\tau/\sqrt{n}))-\arcsin(x^{-1}(\sqrt{n}\,\Re\lambda_0-\tau/\sqrt{n})),&x\in B_n
	\end{cases},
\end{equation*}
with the intervals $A_n:=(\sqrt{n}\,\Re\lambda_0-\tau/\sqrt{n},\sqrt{n}\,\Re\lambda_0+\tau/\sqrt{n}\,)$ and $B_n:=[\sqrt{n}\,\Re\lambda_0+\tau/\sqrt{n},\infty)$. However, by Jordan's inequality \cite[$4.18.1$]{NIST},
\begin{equation*}
	xg_n(x)\leq\pi\sqrt{x^2-\Big(\sqrt{n}\,\Re\lambda_0-\frac{\tau}{\sqrt{n}}\Big)^2}\leq c\left[1+\frac{1}{\sqrt{n}}\right],\ \ x\in A_n;\ \ \ \ \ c=c(\tau,\Re\lambda_0)>0,
\end{equation*}
and also
\begin{align*}
	xg_n(x)=&\,\frac{4\tau}{\sqrt{n}}\int_0^1\frac{x\,\d u}{\sqrt{\big(x-\frac{2\tau}{\sqrt{n}}u-\sqrt{n}\,\Re\lambda_0+\frac{\tau}{\sqrt{n}}\big)\big(x+\frac{2\tau}{\sqrt{n}}u+\sqrt{n}\,\Re\lambda_0-\frac{\tau}{\sqrt{n}}\big)}}\\
	\leq&\,\frac{4\tau}{\sqrt{n}}\frac{x}{\sqrt{x+\sqrt{n}\,\Re\lambda_0-\frac{\tau}{\sqrt{n}}}}\int_0^1\frac{\d u}{\sqrt{x-\frac{2\tau}{\sqrt{n}}u-\sqrt{n}\,\Re\lambda_0+\frac{\tau}{\sqrt{n}}}}\\
	\leq&\,\frac{8\tau}{\sqrt{n}}\frac{\sqrt{x}}{\sqrt{x-\sqrt{n}\,\Re\lambda_0+\frac{\tau}{\sqrt{n}}}+\sqrt{x-\sqrt{n}\,\Re\lambda_0-\frac{\tau}{\sqrt{n}}}}\leq c\left[\frac{1}{\sqrt{n}}+\frac{n^{-\frac{1}{4}}}{\sqrt{x-\sqrt{n}\,\Re\lambda_0+\frac{\tau}{\sqrt{n}}}}\right],\ \ x\in B_n,
\end{align*}
with another $c=c(\tau,\Re\lambda_0)>0$, provided $\sqrt{n}\,\Re\lambda_0-\tau/\sqrt{n}>0$. Put together, for $x\in A_n\cup B_n$ and $n\geq n_0$ sufficiently large,
\begin{equation}\label{p15}
	xg_n(x)\leq c\left[\frac{1}{\sqrt{n}}+\min\left\{1,\frac{n^{-\frac{1}{4}}}{\sqrt{x-\sqrt{n}\,\Re\lambda_0+\frac{\tau}{\sqrt{n}}}}\right\}\right]\leq c\left[\frac{1}{\sqrt{n}}+\frac{1}{1+n^{\frac{1}{4}}\sqrt{x-\sqrt{n}\,\Re\lambda_0+\frac{\tau}{\sqrt{n}}}}\right].
\end{equation}
Equipped with \eqref{p15} we return to \eqref{p14}, apply the aforementioned polar coordinates,
\begin{align*}
	\bigg|\tr_{L^2(\Omega_0)}K_n^k\,\bigg|\leq\frac{1}{\pi^k}\int_{(\sqrt{n}\,\Re\lambda_0-\tau/\sqrt{n},\infty)^k}\left\{\prod_{j=1}^{k-1}g_n(r_j)r_j\exp\left[-\frac{1}{2}(r_j-r_{j+1})^2\right]\right\}g_n(r_k)r_k\sqrt{Q(n,r_k^2)}\,\d r_1\cdots\d r_k,
\end{align*}
and obtain by H\"older's inequality, for any $j\in\{1,\ldots,k-1\}$,
\begin{equation*}
	\int_{(\sqrt{n}\,\Re\lambda_0-\tau/\sqrt{n},\infty)}g_n(r_j)r_j\exp\left[-\frac{1}{2}(r_j-r_{j+1})^2\right]\d r_j\stackrel{\eqref{p15}}{\leq} \frac{c}{\sqrt[8]{n}}\ \ \ \forall\,n\geq n_0,\ \ \ \ c=c(\tau,\Re\lambda_0)>0.
\end{equation*}
Consequently,
\begin{eqnarray}
	\bigg|\tr_{L^2(\Omega_0)}K_n^k\,\bigg|\!\!\!&\leq&\!\!\!cn^{-\frac{1}{8}(k-1)}\int_{\sqrt{n}\,\Re\lambda_0+\Omega_{\ast}/\sqrt{n}}\sqrt{Q(n,|w_k|^2)}\,\d^2w_k\nonumber\\
	\!\!\!&\stackrel{\eqref{p10}}{\leq}&\!\!\!cn^{-\frac{1}{8}(k-1)}\int_{\sqrt{n}\,\Re\lambda_0-\tau/\sqrt{n}}^{\sqrt{n}\,\Re\lambda_0+\tau/\sqrt{n}}\left[\int_{-\infty}^{\infty}\sqrt{Q(n,y^2)}\,\d y\right]\d x\stackrel[\eqref{p10}]{\eqref{p9}}{\leq}cn^{-\frac{1}{8}(k-1)},\label{p16}
\end{eqnarray}
valid for all $n\geq n_0$ uniformly in $k\in\mathbb{Z}_{\geq 1}$ with $c=c(\tau,\Re\lambda_0)>0$. Seeing that the case $\Re\lambda_0<0$ can be treated as in \eqref{p16}, we now combine \eqref{p13},\eqref{p16} with \eqref{p11} and obtain in the right hand side of \eqref{p7},
\begin{equation}\label{p17}
	-\sum_{k=1}^{\infty}\frac{\gamma^k}{k}\tr_{L^2(\Omega_0)}K_n^k=-|\Delta|\gamma+o(1)+\mathcal{O}\Big(n^{-\frac{1}{8}}\Big)\ \ \ \ \forall\,n\geq n_0,\ \ \ |\gamma|\leq 1,
\end{equation}
i.e. the LHS and RHS in \eqref{p7} with $\Omega_{ab}^{cd}=\Omega_0$ are analytic for $|\gamma|\leq 1$ if $n\geq n_0$ is sufficiently large. However, \eqref{p7} says that both sides agree for $|\gamma|$ small and all $n\in\mathbb{Z}_{\geq 1}$, hence by the identity theorem for analytic functions, \eqref{p6},\eqref{p7} and \eqref{p17},
\begin{equation*}
	\lim_{n\rightarrow\infty}E_{n,0}\left(\lambda_0+\frac{\Omega_{\Delta}}{n\varrho_0(\Re\lambda_0)}\right)=\lim_{n\rightarrow\infty}E_{n,0}\left(\lambda_0+\frac{\Omega_0}{n}\right)=\e^{-|\Delta|},
\end{equation*}
uniformly in the bounded interval $\Delta\subset\mathbb{R}$ of length $|\Delta|\in(0,\infty)$. This completes our proof of \eqref{r14} and of \eqref{r17}.
\end{proof}

\begin{rem}
The workings in \eqref{p13} and \eqref{p16} show that $K_n:L^2(\Omega_0)\rightarrow L^2(\Omega_0)$ is small in the Hilbert-Schmidt norm as $n\rightarrow\infty$. A similar thing will happen in our proof of Lemma \ref{lem:2} when analyzing the limiting bulk gap probability in the \textnormal{eGinUE} at strong non-Hermiticity.
\end{rem}

\begin{proof}[Proof of Corollary \ref{cor:1}] By definition of conditional probability, for any $s>0$ and $n\in\mathbb{Z}_{\geq 1}$,
\begin{equation}\label{p18}
	q_{n,0}(s)=\lim_{h\downarrow 0}\frac{\mathbb{P}_{n,0}\big\{\mathcal{R}_n(\big(\lambda_0,\lambda_0+\frac{s}{n\varrho_0(\lambda_0)}\big])=0\ \ \textnormal{and}\ \ \mathcal{R}_n(\big(\lambda_0-\frac{h}{n\varrho_0(\lambda_0)},\lambda_0\big])>0\big\}}{\mathbb{P}_{n,0}\big\{\mathcal{R}_n(\big(\lambda_0-\frac{h}{n\varrho_0(\lambda_0)},\lambda_0\big])>0\big\}},
\end{equation}
where $\mathcal{R}_n(\Delta)$ is the number of eigenvalue real parts of $M_n\in\textnormal{GinUE}$ in a bounded interval $\Delta\subset\mathbb{R}$. But
\begin{equation*}
	\mathbb{P}_{n,0}\{\mathcal{R}_n(\Delta)=0\}=\mathbb{E}_{n,0}\left\{\prod_{\ell=1}^n\Big(1-\chi_{\Delta}\big(\Re\lambda_{\ell}(M_n)\big)\Big)\right\}\stackrel{\eqref{p3}}{=}E_{n,0}(\Omega_{\Delta})
\end{equation*}
by the inclusion-exclusion principle with $\Omega_{\Delta}:=\Delta\times\mathbb{R}\subset\mathbb{R}^2\simeq\mathbb{C}$ and so \eqref{p18} becomes
\begin{equation*}
	q_{n,0}(s)=\lim_{h\downarrow 0}\frac{E_{n,0}\Big(\lambda_0+\frac{(0,s]\times\mathbb{R}}{n\varrho_0(\lambda_0)}\Big)-E_{n,0}\Big(\lambda_0+\frac{(-h,s]\times\mathbb{R}}{n\varrho_0(\lambda_0)}\Big)}{1-E_{n,0}\Big(\lambda_0+\frac{(-h,0]\times\mathbb{R}}{n\varrho_0(\lambda_0)}\Big)},\ \ \ s>0,
\end{equation*}
since $\mathbb{P}(A\cap B)=\mathbb{P}(A)-\mathbb{P}(A\cap B^c)$. Then, by l'Hospital's rule, for any $s>0$ and $n\in\mathbb{Z}_{\geq 1}$,
\begin{equation}\label{p19}
	q_{n,0}(s)=\frac{\partial_xe_n(x,s)}{\partial_xe_n(x,0)}\bigg|_{x=0}\ \ \ \ \textnormal{in terms of}\ \ \ \ \ e_n(x,a):=E_{n,0}\left(\lambda_0+\frac{(-x,a]\times\mathbb{R}}{n\varrho_0(\lambda_0)}\right),\ \ x,a\geq 0,
\end{equation}
and our proof workings for Lemma \ref{lem:1} have shown that for sufficiently large $n\geq n_0$, see \eqref{p7},\eqref{p13} and \eqref{p16},
\begin{equation*}
	e_n(x,a)=\exp\left[-\sum_{k=1}^{\infty}\frac{1}{k}\tr_{L^2(\Omega_x^a)}K_n^k\right],\ \ \ \ \ \Omega_x^a:=\left(-\frac{x}{\varrho_0(\lambda_0)},\frac{a}{\varrho_0(\lambda_0)}\right]\times\mathbb{R}\subset\mathbb{R}^2\simeq\mathbb{C}.
\end{equation*}
We thus proceed by proving that the last series is differentiable with respect to $x$ and afterwards by computing the large $n$ limit of its derivative at $x=0$ for $a\in\{0,s\}$. To that end, abbreviate
\begin{equation*}
	T_k^n(x,a):=\frac{1}{k}\tr_{L^2(\Omega_x^a)}K_n^k,\ \ \ \ x\geq 0,\ \ \ \ k,n\in\mathbb{Z}_{\geq 1},
\end{equation*}
and treat $a\geq 0$ as arbitrary but fixed in what follows. Then,
\begin{align}
	\frac{\d}{\d x}T_1^n(x,a)=&\,\frac{1}{\pi\varrho_0(\lambda_0)}\frac{\d}{\d x}\int_{-x}^a\int_{-\infty}^{\infty}Q\bigg(n,n\left[\left(\lambda_0+\frac{s}{n\varrho_0(\lambda_0)}\right)^2+t^2\right]\bigg)\d t\,\d s\nonumber\\
	=&\,\frac{1}{\pi\varrho_0(\lambda_0)}\int_{-\infty}^{\infty}Q\bigg(n,n\left[\left(\lambda_0-\frac{x}{n\varrho_0(\lambda_0)}\right)^2+t^2\right]\bigg)\d t\stackrel[\eqref{p10}]{\eqref{p9}}{=}1+o(1),\ \ n\rightarrow\infty,\label{p20}
\end{align}
uniformly in $x\geq 0$ chosen from compact subsets. Moreover, for $k\in\mathbb{Z}_{\geq 2}$,
\begin{equation}\label{p21}
	\frac{\d}{\d x}T_k^n(x,a)=\frac{1}{\varrho_0(\lambda_0)}\int_{(\Omega_x^a)^{k-1}}\int_{-\infty}^{\infty}\left\{\prod_{j=1}^kn^{-2}K_{n,0}\left(\lambda_0+\frac{z_j}{n},\lambda_0+\frac{z_{j+1}}{n}\right)\right\}\d\Im z_1\,\d^2z_2\cdots\d^2z_k,
\end{equation}
with $\Re z_1=-x/\varrho_0(\lambda_0)$ and the convention $z_{k+1}\equiv z_1$. Here, we have exploited the fact that
\begin{equation*}
	\mathbb{R}^k\ni(\Re z_1,\ldots,\Re z_k)\mapsto\int_{\mathbb{R}^k}\left\{\prod_{j=1}^kn^{-2}K_{n,0}\left(\lambda_0+\frac{z_j}{n},\lambda_0+\frac{z_{j+1}}{n}\right)\right\}\d\Im z_1\cdots\d\Im z_k,\ \ \ \ z_{k+1}\equiv z_1,
\end{equation*}
is symmetric in its variables. Consequently, first estimating the integral over $\Im z_1$ in \eqref{p21}, and the remaining terms as in \eqref{p12},
\begin{align}
	\bigg|\frac{\d}{\d x}T_k^n(x,a)\,\bigg|\leq \frac{c\sqrt{n}}{(\pi n)^k}\int_{\Omega_{\ast}^{k-1}}&\,\left\{\prod_{j=2}^{k-1}\exp\left[-\frac{n}{2}\left(\bigg|\lambda_0+\frac{z_j}{n}\bigg|-\bigg|\lambda_0+\frac{z_{j+1}}{n}\bigg|\right)^2\right]\right\}\nonumber\\
	&\hspace{1cm}\times\sqrt{Q\Big(n,n\bigg|\lambda_0+\frac{z_k}{n}\bigg|^2\Big)}\,\d^2z_2\cdots\d^2z_k,\ \ \ c=c(\tau,\lambda_0)>0\label{p22}
\end{align}
valid for all $k\in\mathbb{Z}_{\geq 2},n\in\mathbb{Z}_{\geq 1}$ with $\Omega_{\ast}:=(-\tau,\tau)\times\mathbb{R}\supset\Omega_x^a$ where $\tau:=\sup\big\{|y|:\,y\in(-x/\varrho_0(\lambda_0),a/\varrho_0(\lambda_0)]\big\}$. The RHS in \eqref{p22} can be treated as the RHS in \eqref{p12}. Namely, for $\lambda_0=0$ one has
\begin{equation*}
	\bigg|\frac{\d}{\d x}T_k^n(x,a)\,\bigg|\leq c\bigg(\frac{2\tau}{\sqrt{n}}\bigg)^{k-1}(2\pi)^{\frac{1}{2}(k-2)}\ \ \ \forall\,k\in\mathbb{Z}_{\geq 2},\ n\in\mathbb{Z}_{\geq 1}\ \ \ \textnormal{with}\ \ c=c(\tau)>0,
\end{equation*}
and for $\lambda_0\in(-1,1)\setminus\{0\}$,
\begin{equation*}
	\bigg|\frac{\d}{\d x}T_k^n(x,a)\,\bigg|\leq \frac{c}{\sqrt{n}}\,n^{-\frac{1}{8}(k-2)}\ \ \ \ \forall\,k\in\mathbb{Z}_{\geq 2},\ n\geq n_0\ \ \ \textnormal{with}\ \ c=c(\tau,\lambda_0)>0.
\end{equation*}
Hence, for fixed $a\geq 0$,
\begin{equation}\label{p23}
	\sum_{k=1}^{\infty}\frac{\d}{\d x}T_k(x,a)=1+o(1)+\mathcal{O}\big(n^{-\frac{1}{2}}\big)\ \ \ \ \ \forall\,n\geq n_0,
\end{equation}
uniformly in $x\geq 0$ chosen from compact subsets, which yields
\begin{equation*}
	q_{n,0}(s)\stackrel{\eqref{p19}}{=}\frac{e_n(0,s)}{e_n(0,0)}\frac{\sum_{k=1}^{\infty}\frac{\d}{\d x}T_k(x,s)}{\sum_{k=1}^{\infty}\frac{\d}{\d x}T_k^n(x,0)}\bigg|_{x=0}\stackrel{\eqref{p23}}{=}e_n(0,s)+o(1)=\e^{-s}+o(1)\ \ \ \forall\,n\geq n_0
\end{equation*}
uniformly in $s\in(0,\infty)$ on compact subsets where we used $e_n(0,0)=1$ in the second equality and \eqref{r17} in the third. The last estimate yields \eqref{r18} and completes our proof of Corollary \ref{cor:1}.
\end{proof}


\section{The eGinUE at strong non-Hermiticity: Proof of \eqref{r24} and Lemma \ref{lem:2}}\label{sec3}

In our first section on the eGinUE we will derive the mean law \eqref{r24} for eigenvalue real parts and the bulk spacing limit \eqref{r26}, both in the limit of strong non-Hermiticity when $\tau\in(0,1)$ remains fixed. Our proof workings start out as in \cite[page $471$]{FSK}, see also \cite[page $1128$]{ACV}, but are then adapted to the \textit{unbounded} integration domain $\Omega_{\Delta}=\Delta\times\mathbb{R}$.
\begin{proof}[Proof of \eqref{r24}] Fix $\tau\in(0,1)$. We have
\begin{equation}\label{p24}
	\mathbb{E}_{n,\tau}\{\mathcal{R}_{n,\tau}(\Delta)\}=n\int_{\Omega_{\Delta}}\rho_{n,\tau}(z)\,\d^2z,\ \ \ \ \ \ \ \rho_{n,\tau}(z):=p_1^{(n,\tau)}(z)\stackrel{\eqref{r19}}{=}\frac{1}{n}K_{n,\tau}(z,z),
\end{equation}
in terms of the kernel
\begin{equation*}
	K_{n,\tau}(z,w)\stackrel{\eqref{r20}}{=}\frac{n}{\pi\sqrt{1-\tau^2}}\e^{-\frac{n}{2}V_{\tau}(z)-\frac{n}{2}V_{\tau}(w)}\sum_{\ell=0}^{n-1}\frac{(\tau/2)^{\ell}}{\ell!}H_{\ell}\bigg(\sqrt{\frac{n}{2\tau}}\,z\bigg)H_{\ell}\bigg(\sqrt{\frac{n}{2\tau}}\,\overline{w}\bigg),\ \ \ (z,w)\in\mathbb{C}^2,
\end{equation*}
with the Hermite polynomials $\{H_{\ell}\}_{\ell=0}^{\infty}\subset\mathbb{R}[x]$ in \eqref{r5} and the potential function
\begin{equation*}
	V_{\tau}(z)=\frac{1}{1-\tau^2}\Big(|z|^2-\tau\,\Re(z^2)\Big).
\end{equation*}
The integral in \eqref{p24} over $\Omega_{\Delta}=\Delta\times\mathbb{R}$ is well-defined for all $n\in\mathbb{Z}_{\geq 1}$, for with \cite[$18.6.1,18.10.10$]{NIST},
\begin{equation*}
	H_n(z)=\frac{(\pm 2\im)^n}{\sqrt{\pi}}\,\e^{z^2}\int_{-\infty}^{\infty}\e^{-t^2}t^n\e^{\mp 2\im zt}\,\d t,\ \ \ \ z\in\mathbb{C},
\end{equation*}
leading to
\begin{equation}\label{p25}
	\rho_{n,\tau}(z)=\frac{\e^{-nV_{\tau}(z)}}{\pi^2\sqrt{1-\tau^2}}\,\e^{\frac{n}{2\tau}(z^2+\overline{z}^2)}\sum_{\ell=0}^{n-1}\frac{(2\tau)^{\ell}}{\ell!}\int_{-\infty}^{\infty}\int_{-\infty}^{\infty}\e^{-(t^2+s^2)}(ts)^{\ell}\exp\left[2\im\sqrt{\frac{n}{2\tau}}\big(zt-\overline{z}s\big)\right]\d t\,\d s.
\end{equation}
Note that
\begin{equation*}
	\int_{\mathbb{C}}\rho_{n,\tau}(z)\,\d^2 z=1,
\end{equation*}
and hence $\mathbb{C}\ni z\mapsto\rho_{n,\tau}(z)$ is integrable with respect to $\d^2 z$ for any $n\in\mathbb{Z}_{\geq 1},\tau\in(0,1)$. Moving forward, we now evaluate the large $n$-limit of \eqref{p24} by the dominated convergence theorem and Fatou's lemma. To that end consider the change of variables 
\begin{equation*}
	u=(t+s)/\sqrt{n/(2\tau)},\ \ \ \ v=(t-s)/\sqrt{n/(2\tau)},
\end{equation*}
in \eqref{p25} and obtain for any $z\in\mathbb{C},n\in\mathbb{Z}_{\geq 1}$ and $\tau\in(0,1)$,
\begin{align*}
	\rho_{n,\tau}(z)=&\,\frac{1}{\pi^2\sqrt{1-\tau^2}}\Big(\frac{n}{4\tau}\Big)\int_{-\infty}^{\infty}\int_{-\infty}^{\infty}\exp\Bigg[-\frac{n}{4}\Big(\frac{1-\tau}{\tau}\Big)\bigg(u-\frac{\im(z-\overline{z})}{1-\tau}\bigg)^2\\
	&\,-\frac{n}{4}\Big(\frac{1+\tau}{\tau}\Big)\bigg(v-\frac{\im(z+\overline{z})}{1+\tau}\bigg)^2\,\Bigg]Q\Big(n,\frac{n}{4}\big(u^2-v^2\big)\Big)\,\d v\,\d u,\ \ \ \ Q(n,nz)=\e^{-nz}\sum_{\ell=0}^{n-1}\frac{n^{\ell}}{\ell!}z^{\ell},\ \ z\in\mathbb{C}.
\end{align*}
Starting with the innermost $v$-integral we first change variables,
\begin{align}
	\int_{-\infty}^{\infty}\exp&\,\left[-\frac{n}{4}\bigg(\frac{1+\tau}{\tau}\bigg)\left(v-\frac{\im(z+\overline{z})}{1+\tau}\right)^2\right]Q\Big(n,\frac{n}{4}\big(u^2-v^2\big)\Big)\,\d v\label{p26}\\
	=&\,\int_{\mathbb{R}_z}\exp\left[-\frac{n}{4}\bigg(\frac{1+\tau}{\tau}\bigg)v^2\right]Q\big(n,n\,\lambda_{\tau}^z(u,v)\big)\,\d v,\ \ \ \lambda_{\tau}^z(u,v):=\frac{u^2}{4}-\left(\frac{v}{2}+\frac{\im\Re z}{1+\tau}\right)^2,\nonumber
\end{align}
with $\mathbb{R}_z:=\mathbb{R}-2\im\frac{\Re z}{1+\tau}$ and then recall from \cite[$\S 8.12$]{NIST} the decomposition
\begin{equation}\label{p27}
	Q\big(n,n\,\lambda_{\tau}^z(u,v)\big)=\frac{1}{2}\textnormal{erfc}\big(\sqrt{n}\,\eta_{\tau}^z(u,v)\big)+S\big(n,\sqrt{2}\,\eta_{\tau}^z(u,v)\big).
\end{equation}
Here, using principal branches for all multivalued functions,
\begin{equation*}
	\eta_{\tau}^z\equiv\eta_{\tau}^z(u,v):=\big(\lambda_{\tau}^z-\ln\lambda_{\tau}^z-1\big)^{\frac{1}{2}}\sim\frac{1}{\sqrt{2}}(\lambda_{\tau}^z-1),\ \ \ \lambda_{\tau}^z\rightarrow 1,
\end{equation*}
and where, uniformly with respect to $\lambda_z^{\tau}$ in the sector $|\textnormal{arg}\,\lambda_z^{\tau}|<2\pi$,
\begin{equation}\label{p28}
	S\big(n,\sqrt{2}\,\eta_{\tau}^z\big)\sim\frac{\e^{-n(\eta_{\tau}^z)^2}}{\sqrt{2\pi n}}\left(\frac{1}{\lambda_{\tau}^z-1}-\frac{1}{\sqrt{2}\,\eta_{\tau}^z}\right),\ \ \ \textnormal{as}\ \ n\rightarrow\infty.
\end{equation}
Consequently, by the analytic and asymptotic properties of the integrand in the RHS of \eqref{p26}, as function of $v$, cf. \cite[$(7.12.1)$]{NIST}, we can collapse $\mathbb{R}_z$ back to $\mathbb{R}$ without affecting the value of the integral, namely
\begin{equation}\label{p29}
	\textnormal{LHS in}\ \eqref{p26}=\int_{-\infty}^{\infty}\exp\left[-\frac{n}{4}\bigg(\frac{1+\tau}{\tau}\bigg)v^2\right]Q\big(n,n\,\lambda_{\tau}^z(u,v)\big)\,\d v.
\end{equation}
Then, by Cauchy-Schwarz inequality,
\begin{equation}\label{little-2}
	\big|Q(n,nzw)\big|\leq\e^{\frac{n}{2}|z|^2+\frac{n}{2}|w|^2-n\Re(zw)}\sqrt{Q(n,n|z|^2)}\sqrt{Q(n,n|w|^2)}\stackrel{\eqref{p10}}{\leq}\e^{\frac{n}{2}|z|^2+\frac{n}{2}|w|^2-n\Re(zw)},\ \ \ \ z,w\in\mathbb{C},
\end{equation}
which yields
\begin{equation*}
	\left|\exp\left[-\frac{n}{4}\left(\frac{1+\tau}{\tau}\right)v^2\right]Q\big(n,n\,\lambda_{\tau}^z(u,v)\big)\right|\leq\exp\left[-\frac{n}{4}\left(\frac{1-\tau}{\tau}\right)v^2\right]\ \ \forall\,(u,v)\in\mathbb{R}^2,\ \tau\in(0,1),\ z\in\mathbb{C},
\end{equation*}
and thus shows that the integral in the right hand side of \eqref{p29} can be truncated to an integral over the short interval 
\begin{equation*}
	I_n:=\left(-\frac{1}{n^{\epsilon}},\frac{1}{n^{\epsilon}}\right)\subset\mathbb{R},\ \ \  \epsilon\in\Big(0,\frac{1}{2}\Big)\ \ \textnormal{fixed},
\end{equation*}
at the cost of an exponentially small error, once $n$ is large. In detail,
\begin{equation}\label{p30}
	\textnormal{LHS in}\ \eqref{p26}=\int_{I_n}\exp\left[-\frac{n}{4}\left(\frac{1+\tau}{\tau}\right)v^2\right]Q\big(n,n\,\lambda_{\tau}^z(u,v)\big)\,\d v+R_{\tau}^z(u,n),
\end{equation}
with
\begin{equation*}
	\big|R_{\tau}^z(u,n)\big|\leq\sqrt{\frac{\pi\tau}{n(1-\tau)}}\exp\left[-\frac{1-\tau}{4\tau}n^{1-2\epsilon}\right]\ \ \ \ \forall\,n\in\mathbb{Z}_{\geq 1},\ \tau\in(0,1),\ \epsilon\in\Big(0,\frac{1}{2}\Big),\ u\in\mathbb{R},\ z\in\mathbb{C}.
\end{equation*}
Hence, for the same parameter values,
\begin{align*}
	\rho_{n,\tau}(z)=&\,\frac{1}{\pi^2\sqrt{1-\tau^2}}\Big(\frac{n}{4\tau}\Big)\int_{-\infty}^{\infty}\int_{I_n}\exp\left[-\frac{n}{4}\left(\frac{1-\tau}{\tau}\right)\left(u+\frac{2\Im z}{1-\tau}\right)^2-\frac{n}{4}\left(\frac{1+\tau}{\tau}\right)v^2\right]\\
	&\ \ \ \times Q\big(n,n\,\lambda_{\tau}^z(u,v)\big)\,\d v\,\d u+S_{\tau}^z(n),\ \ \ \ \big|S_{\tau}^z(n)\big|\leq\frac{1}{2\pi\sqrt{1-\tau^2}}\frac{1}{1-\tau}\exp\left[-\frac{1-\tau}{4\tau}n^{1-2\epsilon}\right].
\end{align*}
Moving ahead, we note that $\Re\big(\lambda_{\tau}^z(u,v)-\ln\lambda_{\tau}^z(u,v)-1\big)>0$ for large $n\geq n_0$ uniformly in $v\in I_n,u\in\mathbb{R},\tau\in[0,1]$ and for all bounded $|\Re z|$. So with
\begin{equation}\label{alittle0}
	\left|Q(n,nz)-\begin{cases}1,&|z|<1\\ 0,&|z|>1\end{cases}\,\,\,\right|\leq\big|1-|z|\big|^{-1}\,\e^{-n\Re(z-\ln z-1)},\ \ \ z\in\mathbb{C}:\,|z|\neq 1,
\end{equation}
and which is obtained just as \eqref{p9}, we have
\begin{align}
	\rho_{n,\tau}(z)=&\,\frac{1}{\pi^2\sqrt{1-\tau^2}}\Big(\frac{n}{4\tau}\Big)\int_{-4}^{4}\int_{I_n}\exp\left[-\frac{n}{4}\left(\frac{1-\tau}{\tau}\right)\left(u+\frac{2\Im z}{1-\tau}\right)^2-\frac{n}{4}\left(\frac{1+\tau}{\tau}\right)v^2\right]\nonumber\\
	&\ \ \ \times Q\big(n,n\,\lambda_{\tau}^z(u,v)\big)\,\d v\,\d u+T_{\tau}^z(n)+S_{\tau}^z(n),\ \ \ \ \big|T_{\tau}^z(n)\big|\leq\frac{1}{\sqrt{1-\tau^2}}\frac{\e^{-\frac{n}{4}}}{\sqrt{\tau(1-\tau)}}\label{p31}
\end{align}
and 
\begin{equation*}
	\lim_{n\rightarrow\infty}Q\big(n,n\,\lambda_{\tau}^z(u,v)\big)=\begin{cases} 1,&(\frac{u}{2})^2+\big(\frac{\Re z}{1+\tau}\big)^2<1\smallskip\\
	0,&(\frac{u}{2})^2+\big(\frac{\Re z}{1+\tau}\big)^2>1\end{cases},
\end{equation*}
uniformly in $v\in I_n,\tau\in[0,1]$ and for all bounded $|\Re z|$. Consequently, by dominated convergence theorem,
\begin{eqnarray}
	\rho_{n,\tau}(z)&=&\frac{1+o(1)}{\pi^2\sqrt{1-\tau^2}}\sqrt{\frac{n}{4\tau}}\sqrt{\frac{\pi}{1+\tau}}\int_{\Delta_z}\exp\left[-\frac{n}{4}\bigg(\frac{1-\tau}{\tau}\bigg)\left(u+\frac{2\Im z}{1-\tau}\right)^2\right]\,\d u+o(1)\nonumber\\
	&=&\rho_{\textnormal{e}}(z)+o(1)\ \ \ \textnormal{as}\ n\rightarrow\infty;\ \ \ \ \ \ \ \ \ \ \ \ \ \ \rho_{\textnormal{e}}(z)=\frac{1}{\pi(1-\tau^2)}\chi_{\mathbb{D}_{\tau}}(z),\label{p32}
\end{eqnarray}
uniformly in $z\in\mathbb{C}\setminus\partial\mathbb{D}_{\tau}$ and in $\tau\in(0,1)$, both on compact subsets, with $\Delta_z:=\{u\in\mathbb{R}:\,u^2+\big(\frac{2\Re z}{1+\tau}\big)^2<4\}$. Equipped with the pointwise elliptic law \eqref{p32}, we now establish \eqref{r24}: we have for all $n\in\mathbb{Z}_{\geq 1}$ and $\tau\in(0,1)$,
\begin{equation}\label{little1}
	\int_{\mathbb{C}}\rho_{n,\tau}(z)\,\d^2 z=1\ \ \ \ \ \ \textnormal{and}\ \ \ \ \ \ \int_{\mathbb{C}}\rho_{\textnormal{e}}(z)\,\d^2 z=1.
\end{equation}
Moreover, by Fatou's lemma, noting \eqref{p32} and the fact that $\rho_{n,\tau}$ and $\rho_{\textnormal{e}}$ are non-negative, 
\begin{align*}
	A_{\tau}:=&\,\,\liminf_{n\rightarrow\infty}\int_{\Omega_{\Delta}}\rho_{n,\tau}(z)\,\d^2 z-\int_{\Omega_{\Delta}}\rho_{\textnormal{e}}(z)\,\d^2 z\geq 0,\\
	B_{\tau}:=&\,\,\liminf_{n\rightarrow\infty}\int_{\mathbb{C}\setminus\Omega_{\Delta}}\rho_{n,\tau}(z)\,\d^2 z-\int_{\mathbb{C}\setminus\Omega_{\Delta}}\rho_{\textnormal{e}}(z)\,\d^2 z\geq 0,
\end{align*}
where $\Omega_{\Delta}=\Delta\times\mathbb{R}\subset\mathbb{R}^2\simeq\mathbb{C}$. Hence, $A_{\tau}+B_{\tau}=0$ for $\tau\in(0,1)$ by \eqref{little1}, i.e. $A_{\tau}=B_{\tau}=0$ and so
\begin{equation}\label{little2}
	\liminf_{n\rightarrow\infty}\int_{\Omega_{\Delta}}\rho_{n,\tau}(z)\,\d^2z=\int_{\Omega_{\Delta}}\rho_{\textnormal{e}}(z)\,\d^2 z.
\end{equation}
Moving forward, \eqref{little1} also guarantees existence of a subsequence $(\rho_{n_k,\tau})_{k=1}^{\infty}\subset\{\rho_{n,\tau}\}_{n=1}^{\infty}$ so that
\begin{equation*}
	\lim_{k\rightarrow\infty}\int_{\Omega_{\Delta}}\rho_{n_k,\tau}(z)\,\d^2 z=\limsup_{n\rightarrow\infty}\int_{\Omega_{\Delta}}\rho_{n,\tau}(z)\,\d^2 z,
\end{equation*}
and consequently
\begin{equation}\label{little3}
	\limsup_{n\rightarrow\infty}\int_{\Omega_{\Delta}}\rho_{n,\tau}(z)\,\d^2 z=\lim_{k\rightarrow\infty}\int_{\Omega_{\Delta}}\rho_{n_k,\tau}(z)\,\d^2 z=\liminf_{k\rightarrow\infty}\int_{\Omega_{\Delta}}\rho_{n_k,\tau}(z)\,\d^2 z\stackrel{\eqref{little2}}{=}\int_{\Omega_{\Delta}}\rho_{\textnormal{e}}(z)\,\d^2 z.
\end{equation}
So, by \eqref{little2} and \eqref{little3}, as $n\rightarrow\infty$ for any fixed $\tau\in(0,1)$,
\begin{equation*}
	\int_{\Omega_{\Delta}}\rho_{n,\tau}(z)\,\d^2 z=\int_{\Omega_{\Delta}}\rho_{\textnormal{e}}(z)\,\d^2 z+o(1)\stackrel{\eqref{p32}}{=}\int_{\Delta}\varrho_{\tau}(x)\,\d x+o(1),
\end{equation*}
which yields \eqref{r24} upon substitution into \eqref{p24}. Our proof of \eqref{r24} is now complete.
\end{proof}
Equipped with the pointwise limit \eqref{p32}, we now establish Lemma \ref{lem:2}.
\begin{proof}[Proof of Lemma \ref{lem:2}] We adapt the proof workings of Lemma \ref{lem:1} to the eGinUE. To that end we note that the probability $E_{n,\tau}(\Omega)$ that $M_n\in\textnormal{eGinUE}$ has no eigenvalues in a bounded rectangle $\Omega\subset\mathbb{C}$, under \eqref{r1} with $\tau\in(0,1)$, admits a Fredholm determinant representation by Proposition \ref{prop:1}. Indeed, for any finite $n\in\mathbb{Z}_{\geq 1}$ and any $\lambda_0\in\mathbb{C}$,
\begin{eqnarray}
	E_{n,\tau}\left(\lambda_0+\frac{\Omega_{ab}^{cd}}{n}\right)\!\!\!\!\!&=&\!\!\!\!\!E_{n,\tau}\Big[\chi_{\lambda_0+\frac{1}{n}\Omega_{ab}^{cd}}\Big]\label{p36}\\
	&\stackrel{\eqref{r22}}{=}&\!\!\!\!\!1+\sum_{\ell=1}^n\frac{(-1)^{\ell}}{\ell!}\int_{(\Omega_{ab}^{cd})^{\ell}}\det\Big[n^{-2}K_{n,\tau}\left(\lambda_0+\frac{w_j}{n},\lambda_0+\frac{w_k}{n}\right)\Big]_{j,k=1}^{\ell}\d^2 w_1\cdots\d^2 w_{\ell},\nonumber
\end{eqnarray}
with $\Omega_{ab}^{cd}:=(a,b)\times (c,d)\subset\mathbb{R}^2\simeq\mathbb{C}$ where $-\infty<a<b<\infty,-\infty<c<d<\infty$. But $C=[C_{jk}]_{j,k=1}^{\ell}\in\mathbb{C}^{\ell\times\ell}$ with
\begin{equation*}
	C_{jk}:=n^{-2}K_{n,\tau}\left(\lambda_0+\frac{w_j}{n},\lambda_0+\frac{w_k}{n}\right)
\end{equation*}
is positive-definite since Hermitian, compare \eqref{r20}, and since for all $(t_1,\ldots,t_{\ell})\in\mathbb{C}^{\ell}$,
\begin{equation*}
	\sum_{j,k=1}^{\ell}C_{jk}t_j\overline{t_k}=\frac{1}{n\pi\sqrt{1-\tau^2}}\sum_{m=0}^{n-1}\frac{(\tau/2)^{m}}{m!}\left|\sum_{j=1}^{\ell}\exp\left[-\frac{n}{2}V_{\tau}\Big(\lambda_0+\frac{w_j}{n}\Big)\right]H_m\bigg(\sqrt{\frac{n}{2\tau}}\,\Big(\lambda_0+\frac{w_j}{n}\Big)\bigg)t_j\right|^2\geq 0.
\end{equation*}
Thus, by Hadamard's inequality, cf. \cite[$(5.2.72)$]{PS},
\begin{equation*}
	\det\big[C_{jk}\big]_{j,k=1}^{\ell}\leq\prod_{j=1}^{\ell}C_{jj}=\prod_{j=1}^{\ell}n^{-1}\rho_{n,\tau}\Big(\lambda_0+\frac{w_j}{n}\Big),
\end{equation*}
and since
\begin{equation*}
	\int_{\Omega_{ab}^{cd}}n^{-1}\rho_{n,\tau}\Big(\lambda_0+\frac{w}{n}\Big)\,\d^2 w\leq\int_{\mathbb{C}}n^{-1}\rho_{n,\tau}\Big(\lambda_0+\frac{w}{n}\Big)\,\d^2 w\stackrel{\eqref{little1}}{=}n,
\end{equation*}
the right hand side of \eqref{p36} is well-defined for all $-\infty<a<b<\infty,-\infty\leq c<d\leq\infty$, i.e. by continuity of $\mathbb{P}_{n,\tau}$, equality in \eqref{p36} holds also for unbounded vertical strips $\Omega_{ab}^{cd}$ with $c=-\infty$ or $d=\infty$, for any $n\in\mathbb{Z}_{\geq 1}$. Moving forward, the induced integral operator $K_{n,\tau}:L^2(\Omega_{ab}^{cd})\rightarrow L^2(\Omega_{ab}^{cd})$ with
\begin{equation*}
	(K_{n,\tau}f)(z):=\int_{\Omega_{ab}^{cd}}n^{-2}K_{n,\tau}\Big(\lambda_0+\frac{z}{n},\lambda_0+\frac{w}{n}\Big)f(w)\,\d^2 w,\ \ \ \ \ \ z\in\Omega_{ab}^{cd},
\end{equation*}
is well-defined by the properties of the reproducing kernel and has rank at most $n$. Hence, the operator $K_{n,\tau}$ is in particular trace-class and Hilbert-Schmidt on $L^2(\Omega_{ab}^{cd})$, with continuous kernel. In turn, see \cite[$(20)$]{CESX},
\begin{equation}\label{p37}
	\left|E_{n,\tau}\left(\lambda_0+\frac{\Omega_{ab}^{cd}}{n}\right)-\exp\bigg[-\tr_{L^2(\Omega_{ab}^{cd})}K_{n,\tau}\bigg]\right|\leq\|K_{n,\tau}\|_2\exp\left[\frac{1}{2}(\|K_{n,\tau}\|_2+1)^2-\tr_{L^2(\Omega_{ab}^{cd})}K_{n,\tau}\right],
\end{equation}
with the Hilbert-Schmidt norm $\|\cdot\|_2$. Choosing $(a,b)=\Delta/\varrho_{\tau}(\Re\lambda_0)=:\Delta_0^{\tau}$ and $(c,d)=\mathbb{R}$, assuming $\lambda_0\in\mathbb{C}$ belongs to the interior of $\mathbb{D}_{\tau}$, we now first evaluate
\begin{equation*}
	\tr_{L^2(\Omega_0^{\tau})}K_{n,\tau}=\int_{\Delta_0^{\tau}}\int_{-\infty}^{\infty}\rho_{n,\tau}\Big(\Re\lambda_0+\frac{x}{n}+\im y\Big)\,\d y\,\d x,\ \ \ \Omega_0^{\tau}:=\Delta_0^{\tau}\times\mathbb{R}\subset\mathbb{R}^2\simeq\mathbb{C},
\end{equation*}
asymptotically as $n\rightarrow\infty$ for any fixed $\tau\in(0,1)$ and afterwards estimate the Hilbert-Schmidt norm of $K_{n,\tau}$. To that end, by the proof workings of \eqref{r24},
\begin{align}
	\rho_{n,\tau}(x+\im y)=\frac{1}{\pi^2\sqrt{1-\tau^2}}&\,\Big(\frac{n}{4\tau}\Big)\int_{-\infty}^{\infty}\int_{-\infty}^{\infty}\exp\Bigg[-\frac{n}{4}\Big(\frac{1-\tau}{\tau}\Big)\bigg(u+\frac{2y}{1-\tau}\bigg)^2-\frac{n}{4}\Big(\frac{1+\tau}{\tau}\Big)v^2\Bigg]\nonumber\\
	&\,\times Q\big(n,n\,\lambda_{\tau}^x(u,v)\big)\,\d v\,\d u,\ \ \ \ \ \lambda_{\tau}^x(u,v):=\frac{u^2}{4}-\bigg(\frac{v}{2}+\frac{\im x}{1+\tau}\bigg)^2,\label{little-1}
\end{align}
which yields by Fubini's theorem,
\begin{equation}\label{little4}
	\int_{-\infty}^{\infty}\rho_{n,\tau}(x+\im y)\,\d y=\frac{1}{4\pi}\sqrt{\frac{n}{\pi\tau(1+\tau)}}\int_{-\infty}^{\infty}\int_{-\infty}^{\infty}\exp\left[-\frac{n}{4}\bigg(\frac{1+\tau}{\tau}\bigg)v^2\right]Q\big(n,n\,\lambda_{\tau}^x(u,v)\big)\,\d v\,\d u.
\end{equation}
Here, see again our proof workings of \eqref{r24}, especially \eqref{little-2},
\begin{equation*}
	\left|\exp\left[-\frac{n}{4}\bigg(\frac{1+\tau}{\tau}\bigg)v^2\right]Q\big(n,n\,\lambda_{\tau}^x(u,v)\big)\right|\leq\exp\left[-\frac{n}{4}\bigg(\frac{1-\tau}{\tau}\bigg)v^2\right]\sqrt{Q\bigg(n,n\left[\bigg(\frac{u-v}{2}\bigg)^2+\bigg(\frac{x}{1+\tau}\bigg)^2\right]\bigg)},
\end{equation*}
and so
\begin{align*}
	\big|\,\textnormal{LHS}\ \textnormal{in}\ \eqref{little4}\,\big|\leq&\,\,\frac{1}{4\pi}\sqrt{\frac{n}{\pi\tau(1+\tau)}}\int_{-\infty}^{\infty}\exp\left[-\frac{n}{4}\bigg(\frac{1-\tau}{\tau}\bigg)v^2\right]\,\d v\,\int_{-\infty}^{\infty}\sqrt{Q\bigg(n,n\left[\frac{u^2}{4}+\bigg(\frac{x}{1+\tau}\bigg)^2\right]\bigg)}\,\d u\\
	=&\,\,\frac{1}{2\pi}\frac{1}{\sqrt{1-\tau^2}}\int_{-\infty}^{\infty}\sqrt{Q\bigg(n,n\left[\frac{u^2}{4}+\bigg(\frac{x}{1+\tau}\bigg)^2\right]\bigg)}\,\d u
\end{align*}
which by \eqref{p9} and \eqref{p10} leads to the following bound,
\begin{equation}\label{little5}
	\left|\int_{-\infty}^{\infty}\rho_{n,\tau}\Big(\Re\lambda_0+\frac{x}{n}+\im y\Big)\,\d y\,\right|\leq\frac{5}{\pi}\frac{1}{\sqrt{1-\tau^2}}\ \ \ \ \forall\,n\in\mathbb{Z}_{\geq 1},\ \ x\in\mathbb{R},\ \ \tau\in(0,1),
\end{equation}
assuming $\lambda_0\in\mathbb{C}$ belongs to the interior of $\mathbb{D}_{\tau}$. Hence, by the dominated convergence theorem,
\begin{align*}
	\lim_{n\rightarrow\infty}\tr_{L^2(\Omega_0^{\tau})}&\,K_{n,\tau}=\int_{\Delta_0^{\tau}}\left[\lim_{n\rightarrow\infty}\int_{-\infty}^{\infty}\rho_{n,\tau}\Big(\Re\lambda_0+\frac{x}{n}+\im y\Big)\,\d y\right]\d x\\
	\stackrel{\eqref{little4}}{=}&\,\,\frac{1}{4\pi}\frac{1}{\sqrt{\pi\tau(1+\tau)}}\int_{\Delta_0^{\tau}}\left[\lim_{n\rightarrow\infty}\int_{-\infty}^{\infty}\int_{-\infty}^{\infty}\exp\left[-\frac{1}{4}\bigg(\frac{1+\tau}{\tau}\bigg)v^2\right]Q\big(n,n\,\lambda_{\tau}^{w(x)}\Big(u,\frac{v}{\sqrt{n}}\Big)\big)\,\d v\,\d u\right]\d x,
\end{align*}
with $w(x):=\Re\lambda_0+\frac{x}{n}$. However, there exists $c>0$ so that
\begin{align*}
	\bigg|\exp&\,\left[-\frac{1}{4}\bigg(\frac{1+\tau}{\tau}\bigg)v^2\right]Q\big(n,n\,\lambda_{\tau}^{w(x)}\Big(u,\frac{v}{\sqrt{n}}\Big)\big)\bigg|\leq\exp\left[-\frac{1}{4}\bigg(\frac{1-\tau}{\tau}\bigg)v^2\right]\sqrt{Q\bigg(n,\frac{n}{4}\bigg[u-\frac{v}{\sqrt{n}}\bigg]^2\bigg)}\\
	\leq&\,\,c\,\exp\left[-\frac{1}{4}\bigg(\frac{1-\tau}{\tau}\bigg)v^2\right]\begin{cases}1,&\frac{1}{4}(u-v)\in[-2,2]\smallskip\\ \displaystyle\exp\left[-\frac{1}{2}\Big(\frac{1}{4}(u-v)^2-\ln\Big(\frac{1}{4}(u-v)^2\Big)-1\Big)\right],&\frac{1}{4}(u-v)\in\mathbb{R}\setminus[-2,2]\end{cases}
\end{align*}
holds true for all $n\in\mathbb{Z}_{\geq 1},\tau\in(0,1)$ and $u,v\in\mathbb{R}$, compare \eqref{p9} and \eqref{p10}. Hence, again by the dominated convergence theorem, as $n\rightarrow\infty$,
\begin{equation*}
	\tr_{L^2(\Omega_0^{\tau})}K_{n,\tau}\sim\frac{1}{4\pi}\frac{1}{\sqrt{\pi\tau(1+\tau)}}\int_{\Delta_0^{\tau}}\left[\int_{-\infty}^{\infty}\int_{-\infty}^{\infty}\exp\left[-\frac{1}{4}\bigg(\frac{1+\tau}{\tau}\bigg)v^2\right]\lim_{n\rightarrow\infty}Q\big(n,n\,\lambda_{\tau}^{w(x)}\Big(u,\frac{v}{\sqrt{n}}\Big)\big)\,\d v\,\d u\right]\d x,
\end{equation*}
and so with the $(u,v,x,\tau,\lambda_0)$-pointwise limit
\begin{equation*}
	\lim_{n\rightarrow\infty}Q\big(n,n\,\lambda_{\tau}^{w(x)}\Big(u,\frac{v}{\sqrt{n}}\Big)\big)\stackrel{\eqref{alittle0}}{=}\begin{cases}1,&\big(\frac{u}{2}\big)^2+\big(\frac{\Re\lambda_0}{1+\tau}\big)^2<1\smallskip\\ 0,&\big(\frac{u}{2}\big)^2+\big(\frac{\Re\lambda_0}{1+\tau}\big)^2>1\end{cases},
\end{equation*}
in turn
\begin{equation}\label{little6}
	\tr_{L^2(\Omega_0^{\tau})}K_{n,\tau}=|\Delta|+o(1)\ \ \ \textnormal{as}\ \ n\rightarrow\infty\ \ \textnormal{for any fixed}\ \ \tau\in(0,1).
\end{equation}
Next, moving to the Hilbert-Schmidt norm while keeping \eqref{p37} in mind, we first have
\begin{align}
	\|K_{n,\tau}\|_2^2=&\,\int_{(\Delta_0^{\tau}\times\mathbb{R})^2}\bigg|\,n^{-2}K_{n,\tau}\Big(\lambda_0+\frac{z}{n},\lambda_0+\frac{w}{n}\Big)\bigg|^2\d^2 w\,\d^2 z\nonumber\\
	=&\,\int_{\Delta_0^{\tau}}\int_{-\infty}^{\infty}\int_{\Delta_0^{\tau}}\int_{-\infty}^{\infty}\bigg|n^{-1}K_{n,\tau}\Big(\Re\lambda_0+\frac{x_1}{n}+\im y_1,\Re\lambda_0+\frac{x_2}{n}+\im y_2\Big)\bigg|^2\d y_2\d x_2\d y_1\d x_1,\label{little7}
\end{align}
which we now evaluate, as $n\rightarrow\infty$, through an application of the dominated convergence theorem. To that end recall the estimate $|K_{n,\tau}(z,w)|\leq n\sqrt{\rho_{n,\tau}(z)\rho_{n,\tau}(w)},(z,w)\in\mathbb{C}^2$ where, by \eqref{little-1}, after changing variables,
\begin{align*}
	\rho_{n,\tau}(x+\im y)=\frac{1}{\pi^2\sqrt{1-\tau^2}}\Big(\frac{n}{4\tau}\Big)\int_{-\infty}^{\infty}\int_{-\infty}^{\infty}\exp\left[-\frac{n}{4}\left(\frac{1-\tau}{\tau}\right)u^2-\frac{n}{4}\left(\frac{1+\tau}{\tau}\right)v^2\right]Q\big(n,n\,\hat{\lambda}_{\tau}^{x,y}(u,v)\big)\,\d v\,\d u,
\end{align*}
in terms of
\begin{equation*}
	\hat{\lambda}_{\tau}^{x,y}(u,v):=\left(\frac{u}{2}-\frac{y}{1-\tau}\right)^2-\left(\frac{v}{2}+\frac{\im x}{1+\tau}\right)^2.
\end{equation*}
Hence, identifying $z=(\frac{u}{2}-\frac{y}{1-\tau})-(\frac{v}{2}+\frac{\im x}{1+\tau})$ and $w=(\frac{u}{2}-\frac{y}{1-\tau})+(\frac{v}{2}+\frac{\im x}{1+\tau})$ in \eqref{little-2}, we find
\begin{equation}\label{little8}
	\big|\rho_{n,\tau}(x+\im y)\big|\leq\frac{1}{\pi^2\sqrt{1-\tau^2}}\Big(\frac{n}{4\tau}\Big)\int_{-\infty}^{\infty}\int_{-\infty}^{\infty}\exp\left[-\frac{n}{4}\left(\frac{1-\tau}{\tau}\right)(u^2+v^2)\right]\sqrt{Q(n,n|z|^2)}\,\d v\,\d u,
\end{equation}
since $Q(n,nx)\leq 1$ for all $x\geq 0$. Then, after changing variables one more time in \eqref{little8}, for any $\tau\in(0,1)$,
\begin{align*}
	\big|\rho_{n,\tau}(x+\im y)\big|\leq c\sqrt{n}\int_{-\infty}^{\infty}\exp\left[-\frac{n}{8}\left(\frac{1-\tau}{\tau}\right)s^2\right]\sqrt{Q\Big(n,n\bigg[\left(\frac{s}{2}-\frac{y}{1-\tau}\right)^2+\left(\frac{x}{1+\tau}\right)^2\bigg]\Big)}\,\d s,
\end{align*}
with $c=c(\tau)>0$. Hence, by \eqref{p9},\eqref{p10}, for any $\tau\in(0,1)$,
\begin{align*}
	\big|\rho_{n,\tau}(x+\im y)\big|\leq c\sqrt{n}\left\{\int_{I_y^{\tau}}\exp\left[-\frac{n}{8}\left(\frac{1-\tau}{\tau}\right)s^2\right]\d s+\int_{\mathbb{R}\setminus I_y^{\tau}}\exp\left[-\frac{n}{8}\left(\frac{1-\tau}{\tau}\right)s^2-\frac{n}{2}\left(\frac{s}{2}-\frac{y}{1-\tau}\right)^2\right]\d s\right\}
\end{align*}
with $I_y^{\tau}:=\big\{s\in\mathbb{R}:\ |\frac{s}{2}-\frac{y}{1-\tau}|\leq\sqrt{6}\big\}$, and therefore
\begin{equation}\label{little9}
	\big|\rho_{n,\tau}(x+\im y)\big|\leq c_1\left\{\textnormal{erfc}\left[c_2\sqrt{n}\left(\frac{2|y|}{1-\tau}-\sqrt{6}\right)\right]+\e^{-nc_3y^2}\right\},\ \ \forall\,(x,y,n)\in\mathbb{R}^2\times\mathbb{N},
\end{equation}
where $c_k=c_k(\tau)>0$.	 Seeing that the right hand side in \eqref{little9} is decreasing in $n$ once $2|y|>\sqrt{6}(1-\tau)$, we then obtain back in the integrand of \eqref{little7}, for any fixed $\tau\in(0,1)$,
\begin{align*}
	\bigg|n^{-1}K_{n,\tau}\Big(\Re\lambda_0+\frac{x_1}{n}+&\,\im y_1,\Re\lambda_0+\frac{x_2}{n}+\im y_2\Big)\bigg|^2\leq \rho_{n,\tau}\Big(\Re\lambda_0+\frac{x_1}{n}+\im y_1\Big)\rho_{n,\tau}\Big(\Re\lambda_0+\frac{x_2}{n}+\im y_2\Big)\\
	&\,\stackrel{\eqref{little9}}{\leq}\phi(y_1)\phi(y_2)\ \ \ \textnormal{where}\ \ \phi\in L^1(\mathbb{R}),\ \ \ \textnormal{valid for all}\ \ \ (x_1,x_2,y_1,y_2,n)\in\mathbb{R}^4\times\mathbb{N}.
\end{align*}
So, by dominated convergence,
\begin{equation}\label{little10}
	\lim_{n\rightarrow\infty}\|K_{n,\tau}\|_2^2=\int_{\Delta_0^{\tau}}\int_{-\infty}^{\infty}\int_{\Delta_0^{\tau}}\int_{-\infty}^{\infty}\lim_{n\rightarrow\infty}\bigg|n^{-1}K_{n,\tau}\Big(\Re\lambda_0+\frac{x_1}{n}+\im y_1,\Re\lambda_0+\frac{x_2}{n}+\im y_2\Big)\bigg|^2\d y_2\d x_2\d y_1\d x_1,
\end{equation}
provided the limit in the integrand exists pointwise almost everywhere in $(x_1,x_2,y_1,y_2)\in\mathbb{R}^4$. But this is indeed the case as we now demonstrate: First, working along the same lines as back in \eqref{p25}, for any $z,w\in\mathbb{C}$,
\begin{align}
	K_{n,\tau}(z,w)=&\,\frac{n}{\pi^2\sqrt{1-\tau^2}}\left(\frac{n}{4\tau}\right)\exp\left[-\frac{1}{2}\frac{n}{1-\tau^2}\big(|z|^2-2z\overline{w}+|w|^2\big)-\frac{\im}{2}\frac{n\tau}{1-\tau^2}\big(\Im(z^2)-\Im(w^2)\big)\right]\nonumber\\
	&\times\int_{-\infty}^{\infty}\int_{-\infty}^{\infty}\exp\left[-\frac{n}{4}\left(\frac{1-\tau}{\tau}\right)u^2-\frac{n}{4}\left(\frac{1+\tau}{\tau}\right)v^2\right]Q\big(n,n\,\tilde{\lambda}_{\tau}^{z,w}(u,v)\big)\,\d v\,\d u,\label{p38}
%
%
\end{align}
with the abbreviation
\begin{equation*}
	\tilde{\lambda}_{\tau}^{z,w}(u,v):=\left(\frac{u}{2}+\frac{\im(z-\overline{w})}{2(1-\tau)}\right)^2-\left(\frac{v}{2}+\frac{\im(z+\overline{w})}{2(1+\tau)}\right)^2.
\end{equation*}
Note that, once $z\mapsto z':=\Re\lambda_0+x_1/n+\im y_1$ and $w\mapsto w':=\Re\lambda_0+x_2/n+\im y_2$ as needed in \eqref{little10},
\begin{equation*}
	-\frac{1}{2}\frac{n}{1-\tau^2}\big(|z|^2-2z\overline{w}+|w|^2\big)-\frac{\im}{2}\frac{n\tau}{1-\tau^2}\big(\Im(z^2)-\Im(w^2)\big)+\frac{1}{2}\frac{n}{1-\tau^2}(y_1-y_2)^2+\frac{1}{2n}\frac{1}{1-\tau^2}(x_1-x_2)^2
\end{equation*}
is purely imaginary, i.e. by triangle inequality,
\begin{align*}
	\big|n^{-1}K_{n,\tau}(z',w')\big|\leq&\,\frac{1}{\pi^2\sqrt{1-\tau^2}}\left(\frac{n}{4\tau}\right)\exp\left[-\frac{1}{2}\frac{n}{1-\tau^2}(y_1-y_2)^2-\frac{1}{2n}\frac{1}{1-\tau^2}(x_1-x_2)^2\right]\\
	&\,\times\int_{-\infty}^{\infty}\int_{-\infty}^{\infty}\exp\left[-\frac{n}{4}\left(\frac{1-\tau}{\tau}\right)u^2-\frac{n}{4}\left(\frac{1+\tau}{\tau}\right)v^2\right]\Big|Q\big(n,n\,\tilde{\lambda}_{\tau}^{z',w'}(u,v)\big)\Big|\,\d v\,\d u.
\end{align*}
Using \eqref{p8}, we then record the coarse estimate
\begin{equation*}
	\big|Q(n,nz)\big|\leq 1+\sqrt{n}\,\big|\e^{-n(z-1)}z^n\big|\ \ \ \ \ \textnormal{valid for any} \ z\in\mathbb{C}\ \textnormal{and}\ n\in\mathbb{N},
\end{equation*}
and find in turn, after some simplifications,
\begin{equation}\label{p39a}
	\big|n^{-1}K_{n,\tau}(z',w')\big|\leq I_{n,\tau}\big(\{x_j,y_j\}_{j=1}^2\big)+J_{n,\tau}\big(\{x_j,y_j\}_{j=1}^2\big),\ \ n\in\mathbb{N},\ \tau\in(0,1),\ x_j,y_j\in\mathbb{R},
\end{equation}
where
\begin{equation}\label{hp1}
	I_{n,\tau}\big(\{x_j,y_j\}_{j=1}^2\big):=\frac{1}{\pi(1-\tau^2)}\exp\left[-\frac{1}{2}\frac{n}{1-\tau^2}(y_1-y_2)^2-\frac{1}{2n}\frac{1}{1-\tau^2}(x_1-x_2)^2\right],
\end{equation}
and
\begin{align}
	J_{n,\tau}\big(\{x_j,y_j\}_{j=1}^2\big):=\frac{\e^{\frac{n}{4}a_nb_n}}{\pi^2\sqrt{1-\tau^2}}&\,\left(\frac{n^{3/2}}{8\tau}\right)\left(\frac{\e}{4}\right)^n\exp\left[-\frac{1}{2\tau}\frac{n}{1-\tau}(y_1^2+y_2^2)-\frac{1}{2n}\frac{1}{1-\tau^2}(x_1-x_2)^2\right]\nonumber\\
	&\,\times f_{n,\tau}(y_1,a_n)f_{n,\tau}(y_2,b_n),\label{p39aa}
\end{align}
which involves the integral
\begin{equation}\label{p39aaa}
	f_{n,\tau}(y,\omega):=\int_{-\infty}^{\infty}\exp\left[-\frac{n}{8\tau}(t^2+4ty)\right]\big(t^2+\omega^2\big)^{\frac{n}{2}}\d t, \ \ \ \ \ \ y,\omega\in\mathbb{R},
\end{equation}
and the variables
\begin{equation*}
	a_n:=\frac{x_1-x_2}{n(1-\tau)}+\frac{2\Re\lambda_0}{1+\tau}+\frac{x_1+x_2}{n(1+\tau)},\ \ \ \ \ \ b_n:=\frac{x_1-x_2}{n(1-\tau)}-\frac{2\Re\lambda_0}{1+\tau}-\frac{x_1+x_2}{n(1+\tau)}.
\end{equation*}
In order to estimate \eqref{p39aa} further we temporarily fix $1<p,q<\infty$ such that $\frac{1}{p}+\frac{1}{q}=1$ and apply H\"older's inequality in \eqref{p39aaa},
\begin{equation*}
	f_{n,\tau}(y,\omega)\leq\left(\frac{8\pi\tau}{n}\right)^{\frac{1}{2p}}\exp\left[\frac{np}{2\tau}y^2+\frac{n\omega^2}{8\tau q}\right]\left\{\int_{-\infty}^{\infty}\exp\left[-\frac{n}{8\tau}(t^2+\omega^2)\right]\big(t^2+\omega^2\big)^{\frac{nq}{2}}\d t\right\}^{\frac{1}{q}}.
\end{equation*}
Note that $[0,\infty)\ni x\mapsto\e^{-x/8\tau}x^{q/2}=:g_{\tau,q}(x)$ is maximized for $x=4q\tau>0$, and therefore
\begin{equation*}
	f_{n,\tau}(y,\omega)\leq\left(\frac{8\pi\tau}{n}\right)^{\frac{1}{2p}}\exp\left[\frac{np}{2\tau}y^2+\frac{n\omega^2}{8\tau q}\right]\big(g_{\tau,q}(4q\tau)\big)^{\frac{n}{q}}\underbrace{\left\{\int_{-\infty}^{\infty}\frac{g_{\tau,q}(t^2+\omega^2)}{g_{\tau,q}(4q\tau)}\,\d t\right\}^{\frac{1}{q}}}_{=:\,C_{\tau,q}(\omega)>0},
\end{equation*}
as well as
\begin{equation}\label{p39aaaa}
	f_{n,\tau}(y_1,a_n)f_{n,\tau}(y_2,b_n)\leq C_{\tau,q}(a_n)C_{\tau,q}(b_n)\left(\frac{8\pi\tau}{n}\right)^{\frac{1}{p}}\exp\left[\frac{np}{2\tau}(y_1^2+y_2^2)+\frac{n}{8\tau q}(a_n^2+b_n^2)\right]\e^{-n}(4q\tau)^n.
\end{equation}
Note that $(0,1)\ni u\mapsto\frac{1}{2\tau}\frac{1}{1-u}(y_1^2+y_2^2)+\frac{u}{8\tau}(a_n^2+b_n^2)-\ln u=:h(u)$ has a unique global minimum at $u=u_{\ast}=u_{\ast}(\tau,y_1,y_2,a_n,b_n)\in(0,1)$ which is determined as the real-valued solution of
\begin{equation*}
	h'(u)=\frac{1}{2\tau}\frac{1}{(1-u)^2}(y_1^2+y_2^2)+\frac{1}{8\tau}(a_n^2+b_n^2)-\frac{1}{u}=0.
\end{equation*}
Writing $q_{\ast}=\frac{1}{u_{\ast}}\in(1,\infty)$, we then obtain by choosing $q=q_{\ast},p=1/(1-u_{\ast})$ in \eqref{p39aaaa},
\begin{equation*}
	f_{n,\tau}(y_1,a_n)f_{n,\tau}(y_2,b_n)\leq C_{\tau,q_{\ast}}(a_n)C_{\tau,q_{\ast}}(b_n)\left(\frac{8\pi\tau}{n}\right)^{1-u_{\ast}}\e^{nh(u_{\ast})}\Big(\frac{4\tau}{\e}\Big)^n,
\end{equation*}
and thus for any $n>8\pi\tau$ back in \eqref{p39aa},
\begin{equation}\label{hp2}
	J_{n,\tau}\big(\{x_j,y_j\}_{j=1}^2\big)\leq\frac{\e^{n(h(u_{\ast})-h(\tau))}}{\pi^2\sqrt{1-\tau^2}}\left(\frac{n^{3/2}}{8\tau}\right)\exp\left[\frac{\tau}{n(1-\tau)^2(1+\tau)}(x_1-x_2)^2\right]C_{\tau,q_{\ast}}(a_n)C_{\tau,q_{\ast}}(b_n).
\end{equation}
Notice that $h(u_{\ast})<h(\tau)$ almost everywhere in $(x_1,x_2,y_1,y_2)\in\mathbb{R}^4$ and uniformly in $n\in\mathbb{N}$. Moreover, $C_{\tau,q_{\ast}}(a_n)$ and $C_{\tau,q_{\ast}}(b_n)$ are bounded from above uniformly in $n\in\mathbb{N}$, hence both terms in \eqref{p39a}, compare \eqref{hp1} and \eqref{hp2}, tend to zero as $n\rightarrow\infty$, pointwise almost everywhere in $(x_1,x_2,y_1,y_2)\in\mathbb{R}^4$, for any fixed $\tau\in(0,1)$. So, back in \eqref{little10},
\begin{equation}\label{p38aaa}
	\|K_{n,\tau}\|_2=o(1)\ \ \ \textnormal{as}\ \ n\rightarrow\infty\ \ \ \textnormal{for any fixed}\ \ \tau\in(0,1),
\end{equation}
by the dominated convergence theorem, and consequently, \eqref{p38aaa},\eqref{little6} and \eqref{p37} establish
\begin{equation*}
	\lim_{n\rightarrow\infty}E_{n,\tau}\left(\lambda_0+\frac{\Omega_{\Delta}}{n\varrho_{\tau}(\Re\lambda_0)}\right)=\e^{-|\Delta|},\ \ \ \ \Omega_{\Delta}=\Delta\times\mathbb{R}\subset\mathbb{R}^2\simeq\mathbb{C},
\end{equation*}
as claimed in \eqref{r26}, for any fixed $\tau\in(0,1)$, any bounded interval $\Delta\subset\mathbb{R}$ of finite length $|\Delta|\in(0,\infty)$ and any point $\lambda_0$ in the interior of $\mathbb{D}_{\tau}$. Our proof of Lemma \ref{lem:2} is complete.
\end{proof}


\section{The eGinUE at weak non-Hermiticity: proof of \eqref{r33},\eqref{r34}, Proposition \ref{prop:3} and Theorem \ref{theo:2}}\label{sec4}
In this section we establish the Fredholm determinant representations \eqref{r33} and \eqref{r34}. Afterwards Proposition \ref{prop:3} will follow from the bound \cite[$(52)$]{ACV} and the pointwise limit \eqref{r31} established in \cite[Theorem $3$(b)]{ACV}. Lastly we obtain \eqref{r35} by using Proposition \ref{prop:3} and Taylor expansion.

\begin{proof}[Proof of \eqref{r33} and \eqref{r34}] By \eqref{r22}, after changing variables $w_j=(z_j-\lambda_0)n\varrho_1(\lambda_0)$,
\begin{align}
	E_{n,\tau_n}[\phi_n]=1+\sum_{\ell=1}^n&\,\frac{(-1)^{\ell}}{\ell!}\int_{J^{\ell}}\phi(w_1)\cdots\phi(w_{\ell})\nonumber\\
	&\times\det\bigg[\big(n\varrho_1(\lambda_0)\big)^{-2}K_{n,\tau_n}\left(\lambda_0+\frac{w_j}{n\varrho_1(\lambda_0)},\lambda_0+\frac{w_k}{n\varrho_1(\lambda_0)}\right)\bigg]_{j,k=1}^{\ell}\d^2w_1\cdots\d^2w_{\ell},\label{p39}
\end{align}
so that by \eqref{r31}, noticing en route that
\begin{align*}
	\big(\varrho_1(\lambda_0)\sqrt{n}\,\big)^{-2\ell}p_{\ell}^{(n,\tau_n)}&\,\left(\lambda_0+\frac{w_1}{n\varrho(\lambda_0)},\ldots,\lambda_0+\frac{w_{\ell}}{n\varrho_1(\lambda_0)}\right)=\prod_{j=1}^{\ell-1}\left(1-\frac{j}{n}\right)^{-1}\\
	&\times\det\bigg[\big(n\varrho_1(\lambda_0)\big)^{-2}K_{n,\tau_n}\left(\lambda_0+\frac{w_j}{n\varrho_1(\lambda_0)},\lambda_0+\frac{w_k}{n\varrho_1(\lambda_0)}\right)\bigg]_{j,k=1}^{\ell},
\end{align*}
we can pass to the limit $n\rightarrow\infty$ in every term of the sum in \eqref{p39}, given that $J\subset\mathbb{C}$ is compact and $\phi:\mathbb{C}\rightarrow\mathbb{C}$ bounded. We thus need to prove that the same limit can be carried out termwise, i.e. that the $\ell^{\textnormal{th}}$ term of the sum in \eqref{p39} is bounded by the $\ell^{\textnormal{th}}$ term of a convergent series, for large $n$. Once more, we shall achieve this by using a version of Hadamard's inequality. Namely, $D=[D_{jk}]_{j,k=1}^{\ell}\in\mathbb{C}^{\ell\times\ell}$ with
\begin{equation*}
	D_{jk}:=\big(n\varrho_1(\lambda_0)\big)^{-2}K_{n,\tau_n}\left(\lambda_0+\frac{w_j}{n\varrho_1(\lambda_0)},\lambda_0+\frac{w_k}{n\varrho_1(\lambda_0)}\right),
\end{equation*}
is a positive-definite matrix for it is Hermitian and we have for all $(t_1,\ldots,t_{\ell})\in\mathbb{C}^{\ell}$, compare the proof workings of Lemma \ref{lem:2},
\begin{align*}
	\sum_{j,k=1}^{\ell}D_{jk}t_j\overline{t_k}=&\,\frac{(\varrho_1(\lambda_0))^{-2}}{n\pi\sqrt{1-\tau_n^2}}\sum_{m=0}^{n-1}\frac{(\tau_n/2)^m}{m!}\\
	&\ \ \ \times\left|\sum_{j=1}^{\ell}\exp\left[-\frac{n}{2}V_{\tau_n}\left(\lambda_0+\frac{w_j}{n\varrho_1(\lambda_0)}\right)\right]H_m\left(\sqrt{\frac{n}{2\tau_n}}\,\left(\lambda_0+\frac{w_j}{n\varrho_1(\lambda_0)}\right)\right)t_j\right|^2\geq 0.
\end{align*}
Consequently, by \cite[$(5.2.72)$]{PS},
\begin{equation*}
	\det\big[D_{jk}\big]_{j,k=1}^{\ell}\leq\prod_{j=1}^{\ell}D_{jj}=\prod_{j=1}^{\ell}\big(\varrho_1(\lambda_0)\sqrt{n}\,\big)^{-2}\rho_{n,\tau_n}\left(\lambda_0+\frac{w_j}{n\varrho_1(\lambda_0)}\right),
\end{equation*}
and our workings in \eqref{little-1} show
\begin{align*}
	\frac{1}{n}\rho_{n,\tau_n}\left(\lambda_0+\frac{z}{n}\right)=\frac{1}{\pi^2\sqrt{1-\tau_n^2}}\Big(\frac{1}{4\tau_n}\Big)&\,\int_{-\infty}^{\infty}\int_{-\infty}^{\infty}\exp\left[-\frac{n}{4}\left(\frac{1-\tau_n}{\tau_n}\right)\left(u+\frac{2y}{1-\tau_n}\right)^2-\frac{n}{4}\left(\frac{1+\tau_n}{\tau_n}\right)v^2\right]\\
	&\times Q\big(n,n\,\lambda_{\tau_n}^x(u,v)\big)\,\d v\,\d u,\ \ \ \ \ \ x\equiv\lambda_0+\frac{\Re z}{n},\ \ y\equiv\frac{\Im z}{n}.
\end{align*}
But for $\lambda_0\in(-2,2)$,
\begin{equation*}
	\lim_{n\rightarrow\infty}\frac{1}{n}\rho_{n,\tau_n}\left(\lambda_0+\frac{z}{n}\right)=\sqrt{\pi}\left(\frac{\varrho_1(\lambda_0)}{2\pi\sigma}\right)^2\int_{-2\pi\sigma}^{2\pi\sigma}\exp\left[-\frac{1}{4}\left(v+\frac{2}{\sigma}\varrho_1(\lambda_0)\Im z\right)^2\right]\d v,
\end{equation*}
uniformly in a neighborhood of $\lambda_0\in(-2,2)$ for any fixed $\sigma>0$. Thus, there exists $a\in(0,\infty)$ such that the inequality
\begin{equation*}
	\frac{1}{n}\rho_{n,\tau_n}\left(\lambda_0+\frac{z}{n\varrho_1(\lambda_0)}\right)\leq a
\end{equation*}
is valid uniformly in $z$ varying in a compact subset of $\mathbb{C}$ if $n$ is sufficiently large, with $\sigma>0$ fixed. Hence, the $\ell^{\textnormal{th}}$ term of the sum in \eqref{p39} is bounded by $(a\|\phi\|_{\infty}|J|\varrho_1^{-2}(\lambda_0))^{\ell}/\ell!$ and thus \eqref{r33} is established by \eqref{p39} and \eqref{r31}. Lastly, the limiting gap probability is a special case of \eqref{r33} since
\begin{equation*}
	E_{n,\tau_n}\left(\lambda_0+\frac{\Omega}{n\varrho_1(\lambda_0)}\right)=E_{n,\tau_n}\big[\chi_{\lambda_0+\Omega/(n\varrho_1(\lambda_0))}\big]=E_{n,\tau_n}[\phi_n]\ \ \ \ \textnormal{with}\ \ \ \ \phi_n(z):=\chi_{\Omega}\big((z-\lambda_0)n\varrho_1(\lambda_0)\big),
\end{equation*}
so the equality for $E_{\sigma}(\Omega)$ follows indeed from \eqref{r33}. Our proof of \eqref{r33} and \eqref{r34} is complete.
\end{proof}
\begin{proof}[Proof of Proposition \ref{prop:3}] Given our workings in the proof of \eqref{r24}, the step from \eqref{r34} to the unbounded $\Omega_s=(0,s)\times\mathbb{R}$ in Proposition \ref{prop:3} is not hard. Indeed, Theorem \ref{theo:1} establishes the Fredholm determinant identity
\begin{align}
	E_{n,\tau_n}&\,\left(\lambda_0+\frac{\Omega_{ab}^{cd}}{n\varrho_1(\lambda_0)}\right)\label{p40}\\
	=&\,1+\sum_{\ell=1}^n\frac{(-1)^{\ell}}{\ell!}\int_{(\Omega_{ab}^{cd})^{\ell}}\det\left[\big(n\varrho_1(\lambda_0)\big)^{-2}K_{n,\tau_n}\left(\lambda_0+\frac{w_j}{n\varrho_1(\lambda_0)},\lambda_0+\frac{w_k}{n\varrho_1(\lambda_0)}\right)\right]_{j,k=1}^{\ell}\d^2 w_1\cdots\d^2 w_{\ell}\nonumber
\end{align}
for any finite $n\in\mathbb{Z}_{\geq 1}$, any $\lambda_0\in(-2,2)$ and any $-\infty <a<b<\infty,-\infty<c<d<\infty$. And the same still holds true once $c=-\infty$ or $d=\infty$, for by Hadamard's inequality, see our proof of \eqref{r33} for $D=[D_{jk}]_{j,k=1}^{\ell}$,
\begin{equation*}
	\det\big[D_{jk}\big]_{j,k=1}^{\ell}\leq\prod_{j=1}^{\ell}\big(\varrho_1(\lambda_0)\sqrt{n}\,\big)^{-2}\rho_{n,\tau_n}\left(\lambda_0+\frac{w_j}{n\varrho_1(\lambda_0)}\right),
\end{equation*}
and by \cite[(52)]{ACV},
\begin{equation}\label{p41}
	\frac{1}{n}\rho_{n,\tau_n}\left(x+\frac{\im y}{n}\right)\leq \frac{C_1}{n}\exp\left[-C_2\Big(\frac{y}{n}\Big)^2\right],
\end{equation}
for all $x\in\mathbb{R}$ on compact subsets such that either $|x|\geq 2+\delta$ or $0\leq|x|\leq 2-\delta$, for all $y\in\mathbb{R}$, all $\sigma>0$ fixed and all $n\in\mathbb{Z}_{\geq 1}$ with $C_j=C_j(\delta,\sigma)>0$. Hence, noting that $\lambda_0+\Re w_j/(n\varrho_1(\lambda_0))$ is contained in a compact subset of $\mathbb{R}$ for all $n\in\mathbb{Z}_{\geq 1}$ since $a,b$ are finite in \eqref{p40}, we have for all $-\infty\leq c<d\leq\infty$,
\begin{equation}\label{p42}
	\int_{(\Omega_{ab}^{cd})^{\ell}}\det\big[D_{jk}\big]_{j,k=1}^{\ell}\d^2w_1\cdots\d^2 w_{\ell}
	\leq\left(\frac{C_1(b-a)}{\varrho_1^2(\lambda_0)}\right)^{\ell}\left[\int_{-\infty}^{\infty}\e^{-C_2y^2} \d y\right]^{\ell}<\infty.
\end{equation}
Estimate \eqref{p42} proves that the right hand side of \eqref{p40} is well-defined for all $-\infty<a<b<\infty$ and all $-\infty\leq c<d\leq\infty$. Moreover the $\ell^{\textnormal{th}}$ term in the same sum is bounded by the $\ell^{\textnormal{th}}$ term of a convergent series, so, choosing $a=0,b=s>0$ and $c=-\infty,d=\infty$, we obtain
\begin{equation*}
	E_{n,\tau_n}\left(\lambda_0+\frac{\Omega_s}{n\varrho_1(\lambda_0)}\right)\stackrel{\eqref{p40}}{=}E_{n,\tau_n}[\phi_n],\ \ \ \ \ \ \phi_n(z):=\chi_{\Omega_s}\big((z-\lambda_0)n\varrho_1(\lambda_0)\big).
\end{equation*}
Here, the right hand side converges to $E_{\sigma}[\chi_{\Omega_s}]$ as $n\rightarrow\infty$, pointwise in $s,\sigma>0$ by the dominated convergence theorem using \eqref{p41} and the estimate \cite[$(52)$]{ACV} which asserts that the limit \eqref{r31} is also uniform in $(w_1,\ldots,w_{\ell})$ chosen from any vertical strip in $\mathbb{C}^{\ell}$ of bounded width. In turn,
\begin{equation*}
	E_{\sigma}(s)=\lim_{n\rightarrow\infty}E_{n,\tau_n}\left(\lambda_0+\frac{\Omega_s}{n\varrho_1(\lambda_0)}\right)=E_{\sigma}[\chi_{\Omega_s}],
\end{equation*}
for any $s,\sigma>0$, as claimed in Proposition \ref{prop:3}. Our proof is complete.
\end{proof}
We are now prepared to derive Theorem \ref{theo:2}

\begin{proof}[Proof of Theorem \ref{theo:2}] As in the proof of Corollary \ref{cor:1} we begin with the definition of conditional probability, compare \eqref{p18}, and then use
\begin{equation*}
	\mathbb{P}_{n,\tau_n}\big\{\mathcal{R}_{n,\tau_n}(\Delta)=0\big\}=E_{n,\tau_n}(\Omega_{\Delta}),\ \ \ \ \Omega_{\Delta}=\Delta\times\mathbb{R}.
\end{equation*}
What results is
\begin{equation}\label{p43}
	q_n^{\sigma}(s)=\lim_{h\downarrow 0}\frac{E_{n,\tau_n}\big(\lambda_0+\frac{(0,s]\times\mathbb{R}}{n\varrho_1(\lambda_0)}\big)-E_{n,\tau_n}\big(\lambda_0+\frac{(-h,s]\times\mathbb{R}}{n\varrho_1(\lambda_0)}\big)}{1-E_{n,\tau_n}\big(\lambda_0+\frac{(-h,0]\times\mathbb{R}}{n\varrho_1(\lambda_0)}\big)},\ \ \ s>0;\ \ \ \ \tau_n=1-\frac{1}{n}\Big(\frac{\sigma}{\varrho_1(\lambda_0)}\Big)^2\in(0,1),
\end{equation}
and we now use \eqref{p40} in the evaluation of each $E_{n,\tau_n}$. First, as $h\downarrow 0$, uniformly in $n\in\mathbb{Z}_{\geq 1}$, and for any fixed $\sigma>0$,
\begin{equation*}
	E_{n,\tau_n}\left(\lambda_0+\frac{(-h,0]\times\mathbb{R}}{n\varrho_1(\lambda_0)}\right)=1-\frac{h}{\varrho_1(\lambda_0)}\int_{-\infty}^{\infty}\frac{1}{n}\rho_{n,\tau_n}\left(\lambda_0+\frac{\im y}{n}\right)\,\d y+\mathcal{O}\big(h^2\big).
\end{equation*}
Next, by symmetry of $p_{\ell}^{(n,\tau_n)}(z_1,\ldots,z_{\ell})$, as $h\downarrow 0$,
\begin{align*}
	\frac{1}{h}&\,\left[E_{n,\tau_n}\left(\lambda_0+\frac{(0,s]\times\mathbb{R}}{n\varrho_1(\lambda_0)}\right)-E_{n,\tau_n}\left(\lambda_0+\frac{(-h,s]\times\mathbb{R}}{n\varrho_1(\lambda_0)}\right)\right]=\sum_{\ell=1}^n\frac{(-1)^{\ell-1}}{(\ell-1)!}\\
	&\,\times\int_{(0,s)^{\ell-1}}\int_{\mathbb{R}^{\ell}}\det\bigg[\big(n\varrho_1(\lambda_0)\big)^{-2}K_{n,\tau_n}\left(\lambda_0+\frac{w_j}{n\varrho_1(\lambda_0)},\lambda_0+\frac{w_k}{n\varrho_1(\lambda_0)}\right)\bigg]_{\substack{j,k=1\smallskip\\ \Re w_1=0}}^{\ell}\d\Im w_1\d^2w_2\cdots\d^2w_{\ell}+o(1),
\end{align*}
again uniformly in $n\in\mathbb{Z}_{\geq 1}$ and for any fixed $\sigma>0$. Consequently, using also our  proof workings that led to Proposition \ref{prop:3} and by the limit \eqref{r29},
\begin{equation}\label{p44}
	\lim_{n\rightarrow\infty}q_n^{\sigma}(s)=\sum_{\ell=1}^{\infty}\frac{(-1)^{\ell-1}}{(\ell-1)!}\int_{(0,s)^{\ell-1}}\int_{\mathbb{R}^{\ell}}\det\Big[K_{\sin}^{\sigma}(z_j,z_k)\Big]_{\substack{j,k=1\smallskip\\ \Re z_1=0}}^{\ell}\d\Im z_1\,\d^2z_2\cdots\d^2z_{\ell}\ \ \ \ \ \ \ \forall\,s,\sigma>0.
\end{equation}
It remains to notice that
\begin{equation*}
	-\frac{\d}{\d\epsilon}\left[1+\sum_{\ell=1}^{\infty}\frac{(-1)^{\ell}}{\ell!}\int_{(-\epsilon,s)^{\ell}}\int_{\mathbb{R}^{\ell}}\det\big[K_{\sin}^{\sigma}(z_j,z_k)\big]_{j,k=1}^{\ell}\d^2z_1\cdots\d^2 z_{\ell}\right]\Bigg|_{\epsilon=0}=\textnormal{RHS of}\ \eqref{p44},
\end{equation*}
and so by translation invariance of $K_{\sin}^{\sigma}(z,w)$ in the horizontal direction, see \eqref{r32}, and by Proposition \ref{prop:3},
\begin{equation*}
	\textnormal{RHS of}\ \eqref{p44}=-\frac{\d}{\d s}E_{\sigma}[\chi_{\Omega_s}].
\end{equation*}
Thus, combined together,
\begin{equation}\label{p45}
	\wp(s,\sigma)=1-\lim_{n\rightarrow\infty}q_n^{\sigma}(s)=1+\frac{\d}{\d s}E_{\sigma}[\chi_{\Omega_s}]=\int_0^s\frac{\d^2}{\d u^2}E_{\sigma}[\chi_{\Omega_u}]\,\d u;\ \ \ \ \ \ \ \frac{\d}{\d s}E_{\sigma}[\chi_{\Omega_s}]\bigg|_{s=0}=-1.
\end{equation}
Identity \eqref{p45} yields \eqref{r35} provided the induced integral operator $K_{\sin}^{\sigma}:L^2(\Omega_s)\rightarrow L^2(\Omega_s)$ with kernel \eqref{r32} is trace class. We will establish the same in Proposition \ref{alex} below so are now left with the small $s$-behavior of $\wp(s,\sigma)$ for fixed $\sigma>0$. To that end we use the upcoming \eqref{h4},
\begin{equation*}
	E_{\sigma}[\chi_{\Omega_s}]=1-s+\frac{\pi^2s^4}{96}\int_0^2\int_0^2(x-y)^2\e^{-\frac{1}{2}(\pi\sigma)^2(x-y)^2}\,\d x\,\d y+\mathcal{O}\big(s^6\big),\ \ \ \textnormal{as}\ s\downarrow 0,
\end{equation*}
where the error term is twice $s$-differentiable. In turn, for any fixed $\sigma>0$,
\begin{equation*}
	\frac{\d}{\d s}\wp(s,\sigma)=\frac{\d^2}{\d s^2}E_{\sigma}[\chi_{\Omega_s}]=\frac{(\pi s)^2}{8}\int_0^2\int_0^2(x-y)^2\e^{-\frac{1}{2}(\sigma\pi)^2(x-y)^2}\,\d x\,\d y+\mathcal{O}\big(s^4\big),\ \ \ s\downarrow 0.
\end{equation*}
This last item concludes our proof of Theorem \ref{theo:2}.
\end{proof}


\section{The search for an integrable system}\label{sec5}
Our goal is to identify an integrable system for the limiting gap probability $E_{\sigma}(s)$ in \eqref{cool} which, by Proposition \ref{prop:3}, admits a Fredholm determinant representation. Indeed, aligning our notations from now on with \cite[Chapter $21$]{M}, we have $E_{\sigma}(s)=G_{\sigma}(t)$ with $t=\frac{s}{2}>0$ and 
\begin{equation}\label{h1}
	G_{\sigma}(t):=1+\sum_{\ell=1}^{\infty}\frac{(-1)^{\ell}}{\ell!}\int_{(J_t)^{\ell}}\det\big[K_{\sin}^{\sigma}(z_j,z_k)\big]_{j,k=1}^{\ell}\d^2 z_1\cdots\d^2z_{\ell},
\end{equation}
in terms of the kernel \eqref{r32} and the integration domain $J_t:=(-t,t)\times\mathbb{R}\subset\mathbb{R}^2\simeq\mathbb{C}$. The kernel function $\mathbb{C}^2\ni(z,w)\mapsto K_{\sin}^{\sigma}(z,w)$ is continuous for any $\sigma>0$ and it determines a trace class operator on $L^2(J_t)$.
\begin{prop}\label{alex} Let $K_{\sin}^{\sigma}:L^2(J_t)\rightarrow L^2(J_t)$ denote the integral operator with kernel \eqref{r32}, i.e.
\begin{equation}\label{h2}
	(K_{\sin}^{\sigma}f)(z):=\int_{J_t}K_{\sin}^{\sigma}(z,w)f(w)\,\d^2w.
\end{equation}
Then $K_{\sin}^{\sigma}$ is trace class for any $t,\sigma>0$.
\end{prop}
\begin{proof} Writing $z=x+\im y\in\mathbb{C}$, we define the integral operator $A_{\sigma}:L^2(J_t)\rightarrow L^2(-\pi,\pi)$ with kernel
\begin{equation*}
	A_{\sigma}(a,z):=\frac{1}{\sqrt{2\pi^{\frac{3}{2}}\sigma}}\exp\left[-\frac{y^2}{2\sigma^2}-\frac{1}{2}(\sigma a)^2-\im a \overline{z}\right],\ \ \ \ a\in(-\pi,\pi),\ z\in J_t.
\end{equation*}
Since
\begin{equation*}
	\int_{-\pi}^{\pi}\int_{J_t}\big|A_{\sigma}(a,z)\big|^2\d^2 z\,\d a=2t<\infty,
\end{equation*}
the same constitutes a Hilbert-Schmidt transformation from $L^2(J_t)$ to $L^2(-\pi,\pi)$, and thus in particular a bounded linear transformation between the same Hilbert spaces. Moreover, its Hilbert space adjoint is given by the Hilbert-Schmidt integral operator $A_{\sigma}^{\ast}:L^2(-\pi,\pi)\rightarrow L^2(J_t)$ with kernel
\begin{equation*}
	A_{\sigma}^{\ast}(z,a):=\frac{1}{\sqrt{2\pi^{\frac{3}{2}}\sigma}}\exp\left[-\frac{y^2}{2\sigma^2}-\frac{1}{2}(\sigma a)^2+\im az\right],\ \ \ \ z\in J_t,\ a\in(-\pi,\pi).
\end{equation*}
It now remains to check that the action of $A_{\sigma}^{\ast}A_{\sigma}:L^2(J_t)\rightarrow L^2(J_t)$ coincides with the action of $K_{\sin}^{\sigma}:L^2(J_t)\rightarrow L^2(J_t)$ and thus the latter is trace class by \cite[Chapter IV, Lemma $7.2$]{GGK}.
\end{proof}
Equipped with Proposition \ref{alex}, \cite[Theorem $3.7,3.10$]{Sim} yields the following equivalent representation of $G_{\sigma}(t)$, namely
\begin{equation*}
	G_{\sigma}(t)=\prod_{j=1}^{\infty}\big(1-\omega_j(K_{\sin}^{\sigma},J_t)\big),\ \ \ t,\sigma>0,
\end{equation*}
in terms of the non-zero eigenvalues $\omega_j(K_{\sin}^{\sigma},J_t)$ of the integral operator $K_{\sin}^{\sigma}:L^2(J_t)\rightarrow L^2(J_t)$ defined in \eqref{h2}, taking multiplicities into account. More importantly, by Sylvester's identity \cite[Chapter IV, $(5.9)$]{GGK} and the proof workings in Proposition \ref{alex} we record another representation for $G_{\sigma}(t)$ in \eqref{h4} below.
\begin{cor} Introduce the kernel function
\begin{equation}\label{h3}
	S_t^{\sigma}(a,b):=\frac{\sin \pi(a-b)}{\pi(a-b)}\exp\left[-\left(\frac{\pi\sigma}{2t}(a-b)\right)^2\right]
\end{equation}
and let $S_t^{\sigma}:L^2(-t,t)\rightarrow L^2(-t,t)$ denote the corresponding integral operator with $t,\sigma>0$. Then
\begin{equation}\label{h4}
	G_{\sigma}(t)=\prod_{k=1}^{\infty}\Big(1-\omega_k\big(S_t^{\sigma},(-t,t)\big)\Big),\ \ \ t,\sigma>0,
\end{equation}
where $\omega_k(S_t^{\sigma},(-t,t))$ are the non-zero eigenvalues of $S_t^{\sigma}:L^2(-t,t)\rightarrow L^2(-t,t)$, counted according to multiplicities.
\end{cor}
\begin{proof} The trace class operator $A_{\sigma}^{\ast}A_{\sigma}:L^2(-\pi,\pi)\rightarrow L^2(-\pi,\pi)$ has kernel
\begin{equation*}
	(A_{\sigma}A_{\sigma}^{\ast})(a,b)=\int_{J_t}A_{\sigma}(a,z)A_{\sigma}^{\ast}(z,b)\,\d^2 z=\frac{\sin t(a-b)}{\pi(a-b)}\exp\left[-\left(\frac{\sigma}{2}(a-b)\right)^2\right]
\end{equation*}
and rescaling yields the trace equalities
\begin{equation*}
	\tr_{L^2(-\pi,\pi)}(A_{\sigma}A_{\sigma}^{\ast})^k=\tr_{L^2(-t,t)}(S_t^{\sigma})^k,\ \ \ \ \ t,\sigma>0,
\end{equation*}
for all $k\in\mathbb{Z}_{\geq 1}$. Thus \eqref{h4} follows from Sylvester's identity, the Plemelj-Smithies formul\ae\,\cite[Theorem $5.4$]{Sim} and \cite[Theorem $3.7$]{Sim}.
\end{proof}
The simple dependence of \eqref{h3} on $\sigma>0$ makes it possible to analyze the $\sigma\downarrow 0$ and $\sigma\rightarrow\infty$ degenerations of $G_{\sigma}(t)$ at this point.
\begin{cor}\label{deg1} For any fixed $t>0$,
\begin{equation}\label{h5}
	\lim_{\sigma\downarrow 0}G_{\sigma}(t)=1+\sum_{\ell=1}^{\infty}\frac{(-1)^{\ell}}{\ell!}\int_{(-t,t)^{\ell}}\det\left[\frac{\sin\pi(x_j-x_k)}{\pi(x_j-x_k)}\right]_{j,k=1}^{\ell}\d x_1\cdots\d x_{\ell},
\end{equation}
and thus in particular $\lim_{\sigma\downarrow 0}\wp(s,\sigma)=F_1(s),s>0$ in terms of \eqref{r11}.
\end{cor}
\begin{proof} Clearly, pointwise in $a,b\in(-t,t)$ with $t>0$,
\begin{equation}\label{h5a}
	\lim_{\sigma\downarrow 0}S_t^{\sigma}(a,b)=K_{\sin}(a,b):=\frac{\sin\pi(a-b)}{\pi(a-b)}.
\end{equation}
Also, $S_t^{\sigma}$ and $K_{\sin}$ (the operator on $L^2(-t,t)$ with kernel $K_{\sin}(a,b)$) are non-negative and self-adjoint, so we need to show $S_t^{\sigma}\rightarrow K_{\sin}$ weakly and $\|S_t^{\sigma}\|_1\rightarrow\|K_{\sin}\|_1$ in trace norm as $\sigma\downarrow 0$ to conclude the validity of \eqref{h5} by \cite[Theorem $2.20,3.4$]{Sim}. However, for all $\sigma>0$,
\begin{equation*}
	\|S_t^{\sigma}\|_1=\int_{-t}^tS_t^{\sigma}(a,a)\,\d a=2t=\int_{-t}^tK_{\sin}(a,a)\,\d a=\|K_{\sin}\|_1,
\end{equation*}
and, for any $f,g\in L^2(-t,t)$, by the dominated convergence theorem seeing that $|S_t^{\sigma}(a,b)|\leq 1$ for all $a,b\in\mathbb{R}$ and all $t,\sigma>0$,
\begin{equation*}
	\lim_{\sigma\downarrow 0}\int_{-t}^t\int_{-t}^tS_t^{\sigma}(a,b)f(b)\overline{g(a)}\,\d a\,\d b=\int_{-t}^t\int_{-t}^tK_{\sin}(a,b)f(b)\overline{g(a)}\,\d a\,\d b.
\end{equation*}
The proof of \eqref{h5} is complete.
\end{proof}
\begin{cor}\label{deg2} For any fixed $t>0$,
\begin{equation}\label{h6}
	\lim_{\sigma\rightarrow\infty}G_{\sigma}(t)=1+\sum_{\ell=1}^{\infty}\frac{(-1)^{\ell}}{\ell!}\int_{(-t,t)^{\ell}}\det\big[K_{\textnormal{Poi}}(x_j,x_k)\big]_{j,k=1}^{\ell}\d x_1\cdots\d x_{\ell}=\e^{-2t},
\end{equation}
with the Poisson correlation kernel corresponding to unit density,
\begin{equation*}
	K_{\textnormal{Poi}}(x,y)=\begin{cases}1,&x=y\\ 0,&x\neq y\end{cases},\ \ \ \ \ \ x,y\in(-t,t),
\end{equation*}
and thus in particular $\lim_{\sigma\rightarrow\infty}\wp(s,\sigma)=F_0(s),s>0$ in terms of \eqref{r18}.
\end{cor}
\begin{proof} Identity \eqref{h3} readily shows that, pointwise in $a,b\in(-t,t)$ with $t>0$,
\begin{equation}\label{h7}
	\lim_{\sigma\rightarrow\infty}S_t^{\sigma}(a,b)=K_{\textnormal{Poi}}(a,b).
\end{equation}
Moreover, by Hadamard's inequality, for any $\ell\in\mathbb{Z}_{\geq 1}$,
\begin{equation*}
	\det\big[S_t^{\sigma}(x_j,x_k)\big]_{j,k=1}^{\ell}\leq\prod_{j=1}^{\ell}S_t^{\sigma}(x_j,x_j)=1,\ \ \ \ \ x_j\in(-t,t),
\end{equation*}
and thus by the dominated convergence theorem, for any fixed $t>0$,
\begin{align*}
	\lim_{\sigma\rightarrow\infty}G_{\sigma}(t)=&\,1+\lim_{\sigma\rightarrow\infty}\sum_{\ell=1}^{\infty}\frac{(-1)^{\ell}}{\ell!}\int_{(-t,t)^{\ell}}\det\big[S_t^{\sigma}(x_j,x_k)\big]_{j,k=1}^{\ell}\d x_1\cdots\d x_{\ell}\\
	=&\,1+\sum_{\ell=1}^{\infty}\frac{(-1)^{\ell}}{\ell!}\int_{(-t,t)^{\ell}}\lim_{\sigma\rightarrow\infty}\det\big[S_t^{\sigma}(x_j,x_k)\big]_{j,k=1}^{\ell}\d x_1\cdots\d x_{\ell}\stackrel{\eqref{h7}}{=}\e^{-2t}.
\end{align*}
This completes our proof of \eqref{h6}.
\end{proof}
While the representation \eqref{h4}, with \eqref{h3} in place, is very convenient in the analysis of the degenerations of $G_{\sigma}(t)$ as $\sigma\downarrow 0$ and $\sigma\rightarrow\infty$, our derivation of an integrable system for $G_{\sigma}(t)$ will make use of yet another, that is finite-temperature, representation of $G_{\sigma}(t)$. The details are as follows.
\begin{prop} Abbreviate
\begin{equation*}
	\Phi(x):=\frac{1}{\sqrt{\pi}}\int_{-\infty}^x\e^{-y^2}\,\d y,\ \ \ x\in\mathbb{R}.
\end{equation*}
Then for any $a,b\in(-t,t)$ with $t>0$ and any $\sigma>0$,
\begin{equation}\label{h8}
	S_t^{\sigma}(a,b)=\int_0^{\infty}\left[\Phi\Big(\frac{t}{\sigma}(z+1)\Big)-\Phi\Big(\frac{t}{\sigma}(z-1)\Big)\right]\cos\big(\pi(a-b)z\big)\,\d z,
\end{equation}
i.e. $S_t^{\sigma}(a,b)$ is an example of a finite-temperature sine kernel.
\end{prop}
\begin{proof} Using the Gaussian integral
\begin{equation*}
	\e^{-\frac{1}{4}x^2}=\frac{1}{\sqrt{\pi}}\int_{-\infty}^{\infty}\e^{-u^2\pm\im xu}\,\d u,\ \ \ \ x\in\mathbb{R},
\end{equation*}
we rewrite \eqref{h3} as
\begin{align*}
	S_t^{\sigma}(a,b)=&\,\frac{t}{\sigma\sqrt{\pi}}\int_{-\infty}^{\infty}\exp\left[-\Big(\frac{t}{\sigma}\Big)^2(v-1)^2\right]\frac{\sin\big(\pi(a-b)v\big)}{\pi(a-b)}\,\d v\\
	=&\,\frac{t}{\sigma\sqrt{\pi}}\int_0^{\infty}\left\{\exp\left[-\Big(\frac{t}{\sigma}\Big)^2(v-1)^2\right]-\exp\left[-\Big(\frac{t}{\sigma}\Big)^2(v+1)^2\right]\right\}\frac{\sin\big(\pi(a-b)v\big)}{\pi(a-b)}\,\d v,
\end{align*}
and then integrate by parts,
\begin{equation*}
	S_t^{\sigma}(a,b)=\int_0^{\infty}\left(\frac{t}{\sigma\sqrt{\pi}}\int_v^{\infty}\left\{\exp\left[-\Big(\frac{t}{\sigma}\Big)^2(u-1)^2\right]-\exp\left[-\Big(\frac{t}{\sigma}\Big)^2(u+1)^2\right]\right\}\d u\right)\cos\big(\pi(a-b)v\big)\,\d v.
\end{equation*}
All that remains is to express the $v$-antiderivative in terms of $\Phi$ and what results is precisely \eqref{h8}. Our proof is complete.
\end{proof}
Note that \eqref{h8} is an example of the type
\begin{equation}\label{h9}
	K_{\alpha}(x,y):=\int_0^{\infty}\cos\big(\pi(x-y)z\big)w_{\alpha}(z)\,\d z,\ \ \ \ \ x,y\in(-t,t),\ \ \alpha>0
\end{equation}
with smooth, even weight $w_{\alpha}:\mathbb{R}\rightarrow[0,1)$ such that $\lim_{z\rightarrow\infty}w_{\alpha}(z)=0$ exponentially fast\footnote{The $\alpha$-dependence of the weight is not essential in the derivation of an integrable system. We keep it in place only because of \eqref{h8} which has $\alpha=t/\sigma>0$ built in.}. Staying within the same class of kernels, we now gather the following facts.
\begin{lem}\label{traceprop} The integral operator $K_{\alpha}:L^2(-t,t)\rightarrow L^2(-t,t)$ given by
\begin{equation}\label{h10}
	(K_{\alpha}f)(x):=\int_{-t}^tK_{\alpha}(x,y)f(y)\,\d y,
\end{equation}
on the open integral $(-t,t)\subset\mathbb{R}$ with kernel \eqref{h9} is trace class for any $t,\alpha>0$.
\end{lem}
\begin{proof} The kernel function \eqref{h9} is a Hermitian positive function, i.e. we have $K_{\alpha}(x,y)=\overline{K_{\alpha}(y,x)}$ on $\mathbb{R}^2$ and for any $x_1,\ldots,x_n\in(-1,1),z\in\mathbb{C}^n$,
\begin{equation*}
	\sum_{j,k=1}^nz_jK_{\alpha}(x_j,x_k)\overline{z_k}=\frac{1}{2}\int_{-\infty}^{\infty}\left|\sum_{j=1}^nz_j\e^{\im\pi x_jz}\right|^2w_{\alpha}(z)\,\d z\geq 0.
\end{equation*}
Since $\mathbb{R}^2\ni(x,y)\mapsto K_{\alpha}(x,y)$ is also continuous and since
\begin{equation*}
	\int_{-t}^tK_{\alpha}(x,x)\,\d x=2t\int_0^{\infty}w_{\alpha}(z)\,\d z<\infty,
\end{equation*}
it follows by \cite[Theorem $2.12$]{Sim} that the operator $K_{\alpha}$ in \eqref{h10} is trace class on $L^2(-t,t)$ for all $t,\alpha>0$.
\end{proof}
\begin{lem}\label{posprop} For every $t,\alpha>0$, the operator $K_{\alpha}$ in \eqref{h10} satisfies $0\leq K_{\alpha}\leq 1$ and thus in operator norm, $\|K_{\alpha}\|\leq 1$. Additionally, $I-K_{\alpha}$ is invertible on $L^2(-t,t)$ for all $t,\alpha>0$.
\end{lem}
\begin{proof} The argument is standard for a Wiener-Hopf type operator such as \eqref{h10}: one has in the inner product $\langle\cdot,\cdot\rangle$ on $L^2(-t,t)$,
\begin{equation*}
	\langle f,K_{\alpha} f\rangle=\frac{1}{2}\int_{-\infty}^{\infty}\left|\int_{-t}^tf(x)\e^{\im\pi xz}\,\d x\right|^2w_{\alpha}(z)\,\d z,\ \ \ \ \ t,\alpha>0,
\end{equation*}
and thus in terms of the Fourier transform $\check{g}(y):=\frac{1}{\sqrt{2\pi}}\int_{-\infty}^{\infty}g(x)\e^{-\im xy}\,\d x$,
\begin{equation}\label{h11}
	\langle f,K_{\alpha}f\rangle=\pi\int_{-\infty}^{\infty}\big|\check{f_t}(-\pi z)\big|^2w_{\alpha}(z)\,\d z,
\end{equation}
where $f_t(x):=f(x)\chi_{(-t,t)}(x)$ with the indicator function $\chi_{(-t,t)}$ on $(-t,t)$. Hence, by Plancherel's theorem,
\begin{equation}\label{h12}
	\langle f,K_{\alpha}f\rangle\stackrel{\eqref{h11}}{\leq}\pi\int_{-\infty}^{\infty}\big|\check{f}_t(-\pi z)\big|^2\d z=\int_{-\infty}^{\infty}\big|f_t(x)\big|^2\d x=\langle f,f\rangle,
\end{equation}
which says $0\leq K_{\alpha}\leq 1$ and where we used that $w_{\alpha}:\mathbb{R}\rightarrow[0,1)$. Moreover, by self-adjointness of $K_{\alpha}$,
\begin{equation*}
	\|K_{\alpha}\|=\sup_{\langle f,f\rangle=1}\big|\langle f,K_{\alpha}f\rangle\big|\stackrel{\eqref{h12}}{\leq} 1,
\end{equation*}
so we are left to verify invertibility of $I-K_{\alpha}$. To that end, note that $K_{\alpha}$ is compact by Lemma \ref{traceprop}, so we assume that there exists $f\in L^2(-t,t)\setminus\{0\}$ such that $K_{\alpha}f=f$. For the same $f$, equality in \eqref{h12} takes place, in particular we must have
\begin{equation*}
	\langle f,K_{\alpha}f\rangle=\pi\int_{-\infty}^{\infty}\big|\check{f}_t(-\pi z)\big|^2\d z\ \ \ \ \stackrel{\eqref{h11}}{\Rightarrow}\ \ \ \ \int_{-\infty}^{\infty}\big(1-w_{\alpha}(z)\big)\big|\check{f}_t(-\pi z)\big|^2\d z=0.
\end{equation*}
But $w_{\alpha}:\mathbb{R}\rightarrow[0,1)$ is smooth, i.e. the last equality necessarily implies
\begin{equation*}
	\check{f}_t(-\pi z)=\frac{1}{\sqrt{2\pi}}\int_{-t}^tf(x)\e^{\im\pi zx}\,\d x=0\ \ \ \forall\,z\in\mathbb{R}.
\end{equation*}
This yields $f=0\in L^2(-t,t)\subset L^1(-t,t)$ by the Fourier uniqueness theorem, a contradiction. Thus the operator $I-K_{\alpha}$ is injective on $L^2(-t,t)$ for all $t,\alpha>0$ and hence invertible by the Fredholm Alternative. Our proof is complete.
\end{proof}
By Lemma \ref{traceprop}, the Fredholm determinant $D(t,\alpha)$ of the operator $K_{\alpha}$ on $L^2(-t,t)$, i.e.
\begin{equation}\label{h13}
	D(t,\alpha):=\prod_{j=1}^{\infty}\Big(1-\omega_j\big(K_{\alpha},(-t,t)\big)\Big),
\end{equation}
is well-defined and non-zero for all $t,\alpha>0$. In \eqref{h13} we let $\omega_j(K_{\alpha},(-t,t))$ denote the non-zero eigenvalues of $K_{\alpha}$ counted according to their multiplicities. Moving ahead, we will now express \eqref{h13} in terms of an integro-differential dynamical system in the main variable $t>0$, and the same representation will be valid for all values $\alpha>0$. To get back to \eqref{h4} we then use \eqref{h8}, i.e. we choose
\begin{equation}\label{h14}
	w_{\alpha}(z)=\Phi\big(\alpha(z+1)\big)-\Phi\big(\alpha(z-1)\big),\ \ \ \ \ z\in\mathbb{R},\ \alpha>0,
\end{equation}
and equate, with \eqref{h14} in place, $G_{\sigma}(t)=D(t,\alpha)\big|_{\alpha=t/\sigma}$. The details are worked out in the next section.

\section{Deriving the integrable system: Proof of Theorem \ref{theo:3}}\label{sec6}
In order to derive an integrable system for $D(t,\alpha)$ we follow, at least for a while, the algebraic steps of Tracy and Widom \cite[Section VI]{TW0}, see also \cite[Section $3.6$]{AGZ}, in their derivation of the Jimbo-Miwa-Mori-Sato equations for the sine kernel determinant \eqref{h9} when $w_{\alpha}$ is the indicator function on $(-1,1)$. We begin by rewriting \eqref{h9} in generalized \textit{integrable form}, cf. \cite{IIKS}. Throughout we treat $t$ and $\alpha$ as independent variables.
\begin{lem} Set $\phi(x):=\frac{1}{\pi}\sin(\pi x)$ and $\psi(x):=\phi'(x)$, then for any $x,y\in(-t,t)$ with $t>0$ and any $\alpha>0$,
\begin{equation}\label{i1}
	K_{\alpha}(x,y)=-\int_0^{\infty}\left[\frac{\phi(xz)\psi(zy)-\psi(xz)\phi(zy)}{x-y}\right]\d w_{\alpha}(z),\ \ \ \ \ \ \d w_{\alpha}(z)\equiv w_{\alpha}'(z)\,\d z.
\end{equation}
\end{lem}
\begin{proof} Simply integrate \eqref{h9} by parts, using that $w_{\alpha}$ is smooth and vanishes at $\infty$,
\begin{equation*}
	K_{\alpha}(x,y)=-\int_0^{\infty}\left[\frac{\sin(\pi(x-y)z)}{\pi(x-y)}\right]\d w_{\alpha}(z).
\end{equation*}
The last identity yields \eqref{i1} after an application of the sine addition formula.
\end{proof}
Moving ahead, we will think of the operator $K_{\alpha}$ as acting, not on $(-t,t)$ but on $\mathbb{R}$ and to have kernel
\begin{equation*}
	K_{\alpha}^t(x,y):=K_{\alpha}(x,y)\chi_J(y),\ \ \ \ \ \ \ J:=[-t,t]\subset\mathbb{R},\ \ \ \ (x,y)\in\mathbb{R}^2.
\end{equation*}
Observe that $K_{\alpha}^t$ is no longer symmetric, unlike $K_{\alpha}$. Since $K_{\alpha}^t$ is trace class, we begin with the identity, 
\begin{equation*}
	\frac{\partial}{\partial t}\ln D(t,\alpha)=-\tr_{L^2(\mathbb{R})}\left((I-K_{\alpha}^t)^{-1}\frac{\partial K_{\alpha}^t}{\partial t}\right),\ \ \ \ \ \frac{\partial K_{\alpha}^t}{\partial t}\tw K_{\alpha}^t(x,t)\delta(y-t)+K_{\alpha}^t(x,-t)\delta(y+t),
\end{equation*}
where ``$\tw$'' means ``has kernel'' and conclude from it
\begin{equation}\label{i2}
	\frac{\partial}{\partial t}\ln D(t,\alpha)=-R_{\alpha}^t(t,t)-R_{\alpha}^t(-t,-t)=-2R_{\alpha}^t(t,t),
\end{equation}
by using the resolvent kernel $R_{\alpha}^t(x,y)$ of $K_{\alpha}^t$, i.e. $R_{\alpha}^t=(I-K_{\alpha}^t)^{-1}K_{\alpha}^t\tw R_{\alpha}^t(x,y)$ and where the second equality in \eqref{i2} follows from the fact that the kernel of $K_{\alpha}$ is symmetric and even for $x,y\in[-t,t]$, see \eqref{h9}. Note that $R_{\alpha}^t(x,y)$ is smooth in $x$, but discontinuous at $y=t$, so $R_{\alpha}^t(t,t)$ really means
\begin{equation*}
	\lim_{y\uparrow t}R_{\alpha}^t(t,y).
\end{equation*}
Following the standard moves in the game of \cite{TW0} we now assemble formul\ae\,for $R_{\alpha}^t(x,y)$ and closely related quantities.
\subsection{Formula for $R_{\alpha}^t(x,y)$} Let $M$ denote multiplication by the independent variable. Then for the commutator
\begin{equation}\label{i3}
	\big[M,K_{\alpha}^t\big]\tw x K_{\alpha}^t(x,y)-K_{\alpha}^t(x,y)y\stackrel{\eqref{i1}}{=}-\int_0^{\infty}\big(\phi(xz)\psi(zy)-\psi(xz)\phi(zy)\big)\chi_J(y)\,\d w_{\alpha}(z),
\end{equation}
and so, using the dilation $\tau_z$ with $(\tau_zf)(x):=f(xz)$, we obtain the commutator kernel identity
\begin{align}
	\big[M,&\,(I-K_{\alpha}^t)^{-1}\big]=(I-K_{\alpha}^t)^{-1}\big[M,K_{\alpha}^t\big](I-K_{\alpha}^t)^{-1}\nonumber\\
	&\,\tw -\int_0^{\infty}\Big[Q_{\alpha}(x;z,t)\big((I-K_{\alpha}^{t\ast})^{-1}(\tau_z\psi)\chi_J\big)(y)-P_{\alpha}(x;z,t)\big((I-K_{\alpha}^{t\ast})^{-1}(\tau_z\phi)\chi_J\big)(y)\Big]\d w_{\alpha}(z),\label{i4}
\end{align}
where $K_{\alpha}^{t\ast}$ is the real adjoint of $K_{\alpha}^t$ and we use the below generalizations of \cite[$(53),(54)$]{TW0},
\begin{equation*}
	Q_{\alpha}(x;z,t):=\big((I-K_{\alpha}^t)^{-1}\tau_z\phi\big)(x),\ \ \ \ \ \ \ P_{\alpha}(x;z,t):=\big((I-K_{\alpha}^t)^{-1}\tau_z\psi\big)(x),
\end{equation*}
defined for any $x\in\mathbb{R}$ and $z,t,\alpha>0$. Next, with $(I-K_{\alpha}^t)^{-1}\tw\delta(x-y)+R_{\alpha}^t(x,y)$, we also have 
\begin{equation*}
	[M,(I-K_{\alpha}^t)^{-1}]\tw (x-y)R_{\alpha}^t(x,y)\ \ \textnormal{on}\  \mathbb{R}^2,
\end{equation*}
and so after comparison with \eqref{i4},
\begin{equation}\label{i5}
	R_{\alpha}^t(x,y)=-\int_0^{\infty}\left[\frac{Q_{\alpha}(x;z,t)P_{\alpha}(y;z,t)-P_{\alpha}(x;z,t)Q_{\alpha}(y;z,t)}{x-y}\right]\d w_{\alpha}(z),\ \ \ x\in[-t,t],\ y\in(-t,t),
\end{equation}
where we used the identities, valid for any $z,t>0$,
\begin{equation*}
	\begin{cases}(I-K_{\alpha}^{t\ast})^{-1}(\tau_z\psi)\chi_J=(I-K_{\alpha}^t)^{-1}\tau_z\psi=P_{\alpha}\\
	(I-K_{\alpha}^{t\ast})^{-1}(\tau_z\phi)\chi_J=(I-K_{\alpha}^t)^{-1}\tau_z\phi=Q_{\alpha}
	\end{cases}\ \ \ \ \ \textnormal{on}\ \ (-t,t),
\end{equation*}
and which are based on the symmetry of $K_{\alpha}$, i.e. the identity $K_{\alpha}=K_{\alpha}^{\ast}$. The special case $y\rightarrow x\in(-t,t)$ in \eqref{i5} reads
\begin{equation*}
	R_{\alpha}^t(x,x)=-\int_0^{\infty}\Big[Q_{\alpha}'(x;z,t)P_{\alpha}(x;z,t)-P_{\alpha}'(x;z,t)Q_{\alpha}(x;z,t)\Big]\d w_{\alpha}(z),\ \ \ x\in(-t,t),
\end{equation*}
with $(')$ as $x$-derivative (the first variable in $Q_{\alpha},P_{\alpha}$) and it yields the below generalization of \cite[Lemma $4$]{TW0}.
\begin{lem} For any $t,\alpha>0$,
\begin{equation}\label{i6}
	R_{\alpha}^t(t,t)=\lim_{y\uparrow t} R_{\alpha}^t(t,y)=-\int_0^{\infty}\Big[Q_{\alpha}'(t;z,t)p_{\alpha}(t,z)-P_{\alpha}'(t;z,t)q_{\alpha}(t,z)\Big]\d w_{\alpha}(z)
\end{equation}
with the quantities, generalizing \cite[$(6.10)$]{TW0},
\begin{equation*}
	p_{\alpha}(t,z):=P_{\alpha}(t;z,t),\ \ \ \ \ \ \ \ q_{\alpha}(t,z):=Q_{\alpha}(t;z,t),\ \ \ \ \ z,t,\alpha>0.
\end{equation*}
\end{lem}
Because of \eqref{i2} and \eqref{i6} we now proceed with the computation of the derivatives $Q_{\alpha}'(x;z,t)$ and $P_{\alpha}'(x;z,t)$, adapting \cite[$(6.12)-(6.18)$]{TW0} to our needs in the course of it. First, with $D$ as differentiation with respect to the first variable, for any $x\in\mathbb{R}$ and $z,t,\alpha>0$,
\begin{align}
	Q_{\alpha}'(x;z,t)=&\,\,\big(D(I-K_{\alpha}^t)^{-1}\tau_z\phi\big)(x)=\big((I-K_{\alpha}^t)^{-1}D\tau_z\phi\big)(x)+\big([D,(I-K_{\alpha}^t)^{-1}]\tau_z\phi\big)(x)\nonumber\\
	=&\,\,zP_{\alpha}(x;z,t)+\big([D,(I-K_{\alpha}^t)^{-1}]\tau_z\phi\big)(x).\label{i7}
\end{align}
having used $D\tau_z\phi=z\tau_z\psi$ in the third equality. Next we compute the commutator $[D,(I-K_{\alpha}^t)^{-1}]$.
\begin{prop} For any $x,y\in\mathbb{R}$ and $t,\alpha>0$,
\begin{equation}\label{i8}
	[D,(I-K_{\alpha}^t)^{-1}]\tw-\big(R_{\alpha}^t(x,t)\rho(t,y)-R_{\alpha}^t(x,-t)\rho(-t,y)\big),
\end{equation}
where $\rho(x,y)=\delta(x-y)+R_{\alpha}^t(x,y)$ is the distributional kernel of $(I-K_{\alpha}^t)^{-1}$.
\end{prop}
\begin{proof} We use $[D,(I-K_{\alpha}^t)^{-1}]=(I-K_{\alpha}^t)^{-1}[D,K_{\alpha}^t](I-K_{\alpha}^t)^{-1}$ and integration by parts,
\begin{eqnarray*}
	\big[D,K_{\alpha}^t\big]\!\!\!\!&\tw&\!\!\!\!\left(\frac{\partial}{\partial x}K_{\alpha}(x,y)+\frac{\partial}{\partial y}K_{\alpha}(x,y)\right)\chi_J(y)-\big(K_{\alpha}^t(x,t)\delta(y-t)-K_{\alpha}^t(x,-t)\delta(y+t)\big)\\
	\!\!\!\!&\stackrel{\eqref{h9}}{=}&\!\!\!\!-\big(K_{\alpha}^t(x,t)\delta(y-t)-K_{\alpha}^t(x,-t)\delta(y+t)\big),
\end{eqnarray*}
given that $\frac{\partial}{\partial x}K_{\alpha}(x,y)+\frac{\partial}{\partial y}K_{\alpha}(x,y)=0$. The above yields
\begin{equation*}
	[D,(I-K_{\alpha}^t)^{-1}]\tw-\big(R_{\alpha}^t(x,t)\rho(t,y)-R_{\alpha}^t(x,-t)\rho(-t,y)\big),
\end{equation*}
and thus verifies the claim \eqref{i8}.
\end{proof}
Inserting \eqref{i8} into \eqref{i7} we find for all $x\in\mathbb{R}$ and $z,t,\alpha>0$,
\begin{eqnarray}
	Q_{\alpha}'(x;z,t)\!\!&=&\!\!zP_{\alpha}(x;z,t)-R_{\alpha}^t(x,t)Q_{\alpha}(t;z,t)+R_{\alpha}^t(x,-t)Q_{\alpha}(-t;z,t)\nonumber\\
	&=&\!\!zP_{\alpha}(x;z,t)-\big(R_{\alpha}^t(x,t)+R_{\alpha}^t(x,-t)\big)q_{\alpha}(t,z),\label{i9}
\end{eqnarray}
since $Q_{\alpha}(-t;z,t)=-Q_{\alpha}(t;z,t)=-q_{\alpha}(t,z)$ by symmetry and evenness of $K_{\alpha}(x,y)$, see \eqref{h9}. Next,
\begin{eqnarray}
	P_{\alpha}'(x;z,t)\!\!\!\!&=&\!\!\!\!\big(D(I-K_{\alpha}^t)^{-1}\tau_z\psi\big)(x)=\big((I-K_{\alpha}^t)^{-1}D\tau_z\psi\big)(x)+\big([D,(I-K_{\alpha}^t)^{-1}]\tau_z\psi\big)(x)\nonumber\\
	&\stackrel{\eqref{i8}}{=}&\!\!\!\!-z\pi^2Q_{\alpha}(x;z,t)-\big(R_{\alpha}^t(x,t)-R_{\alpha}^t(x,-t)\big)p_{\alpha}(t,z),\label{i10}
\end{eqnarray}
valid for any $x\in\mathbb{R}$ and $z,t,\alpha>0$, where we used $D\tau_z\psi=-\pi^2z\tau_z\phi$ and $P_{\alpha}(-t;z,t)=P_{\alpha}(t;z,t)=p_{\alpha}(t,z)$, again by symmetry and evenness of $K_{\alpha}(x,y)$. Summarizing our immediate results, from \eqref{i9} and \eqref{i10} for $x=t$, which are needed in \eqref{i2} and \eqref{i6}:
\begin{prop} For any $z,t,\alpha>0$,
\begin{eqnarray*}
	Q_{\alpha}'(t;z,t)\!\!\!&=&\ \ \ \,zp_{\alpha}(t,z)-\big(R_{\alpha}^t(t,t)+R_{\alpha}^t(t,-t)\big)q_{\alpha}(t,z),\\
	P_{\alpha}'(t;z,t)\!\!\!&=&\!\!\!-z\pi^2q_{\alpha}(t,z)-\big(R_{\alpha}^t(t,t)-R_{\alpha}^t(t,-t)\big)p_{\alpha}(t,z),
\end{eqnarray*}
so we have in particular back in \eqref{i6}, for any $t,\alpha>0$,
\begin{equation}\label{i11}
	R_{\alpha}^t(t,t)=-\int_0^{\infty}\Big[zp_{\alpha}^2(t,z)+z\pi^2q_{\alpha}^2(t,z)-2R_{\alpha}^t(t,-t)q_{\alpha}(t,z)p_{\alpha}(t,z)\Big]\d w_{\alpha}(z),
\end{equation}
which generalizes \cite[$(6.25)$]{TW0} to our setup \eqref{h9}. Moreover, from \eqref{i5} and with $Q_{\alpha}(-t;z,t)=-q_{\alpha}(t,z)$ as well as $P_{\alpha}(-t;z,t)=p_{\alpha}(t,z)$,
\begin{equation}\label{i12}
	tR_{\alpha}^t(t,-t)=-\int_0^{\infty}q_{\alpha}(t,z)p_{\alpha}(t,z)\,\d w_{\alpha}(z),\ \ \ \ t,\alpha>0.
\end{equation}
\end{prop}
At this point of our calculation we derive first order equations for $q_{\alpha}$ and $p_{\alpha}$.
\subsection{Equations for $q_{\alpha}$ and $p_{\alpha}$} By chain rule
\begin{equation*}
	\frac{\partial}{\partial t}q_{\alpha}(t,z)=\left(\frac{\partial}{\partial x}+\frac{\partial}{\partial t}\right)Q_{\alpha}(x;z,t)\bigg|_{x=t},\ \ \ \ \ \ \frac{\partial}{\partial t}p_{\alpha}(t,z)=\left(\frac{\partial}{\partial x}+\frac{\partial}{\partial t}\right)P_{\alpha}(x;z,t)\bigg|_{x=t},
\end{equation*}
so we can use the $x$-derivatives in \eqref{i9} and \eqref{i10} combined with
\begin{align*}
	\frac{\partial}{\partial t}Q_{\alpha}(x;z,t)=&\,\frac{\partial}{\partial t}\big((I-K_{\alpha}^t)^{-1}\tau_z\phi\big)(x)=\Big((I-K_{\alpha}^t)^{-1}\frac{\partial K_{\alpha}^t}{\partial t}(I-K_{\alpha}^t)^{-1}\tau_z\phi\Big)(x),\\
	\frac{\partial}{\partial t}P_{\alpha}(x;z,t)=&\,\frac{\partial}{\partial t}\big((I-K_{\alpha}^t)^{-1}\tau_z\psi\big)(x)=\Big((I-K_{\alpha}^t)^{-1}\frac{\partial K_{\alpha}^t}{\partial t}(I-K_{\alpha}^t)^{-1}\tau_z\psi\Big)(x),
\end{align*}
and $\frac{\partial}{\partial t}K_{\alpha}^t\tw K_{\alpha}^t(x,t)\delta(y-t)+K_{\alpha}^t(x,-t)\delta(y+t)$. What results is summarized below.
\begin{prop} For any $z,t,\alpha>0$,
\begin{eqnarray}
	\frac{\partial}{\partial t}q_{\alpha}(t,z)\!\!\!&=&\ \ \ \,zp_{\alpha}(t,z)-2R_{\alpha}^t(t,-t)q_{\alpha}(t,z),\label{i13}\\
	\frac{\partial}{\partial t}p_{\alpha}(t,z)\!\!\!&=&\!\!\!-z\pi^2q_{\alpha}(t,z)+2R_{\alpha}^t(t,-t)p_{\alpha}(t,z).\label{i14}
\end{eqnarray}
which we can compare with \cite[$(6.24)$]{TW0}.
\end{prop}
We are now prepared to derive the integro-differential dynamical system for $D(t,\alpha)$ and precisely at this moment our steps start to deviate from \cite{TW0}. Namely, we do not generalize \cite[(6.29),(6.30),(6.31)]{TW0} to \eqref{h9}, but instead \cite[$(21.3.6)$]{M} or equivalently \cite[$(3.6.35),(3.6.36)$]{AGZ}.
\subsection{Final steps} Introduce
\begin{equation*}
	r_{\alpha}(t,z):=\int_{-t}^t\rho(t,y)\e^{-\im\pi zy}\,\d y,\ \ \ \ \ z,t,\alpha>0
\end{equation*}
in terms of the distributional kernel $\rho(x,y)=\delta(x-y)+R_{\alpha}^t(x,y)$ of $(I-K_{\alpha}^t)^{-1}$. Now Fourier-decompose the real-valued $q_{\alpha}(t,z),p_{\alpha}(t,z)$,
\begin{equation*}
	q_{\alpha}(t,z)=\frac{1}{2\pi\im}\Big[\overline{r_{\alpha}(t,z)}-r_{\alpha}(t,z)\Big],\ \ \ \ \ \ \ p_{\alpha}(t,z)=\frac{1}{2}\Big[\overline{r_{\alpha}(t,z)}+r_{\alpha}(t,z)\Big],
\end{equation*}
and obtain from \eqref{i13},\eqref{i14} that
\begin{equation}\label{i15}
	\frac{\partial}{\partial t}r_{\alpha}(t,z)=\frac{\partial}{\partial t}\big(p_{\alpha}(t,z)-\im\pi q_{\alpha}(t,z)\big)=-\im\pi zr_{\alpha}(t,z)+2R_{\alpha}^t(t,-t)\overline{r_{\alpha}(t,z)}.
\end{equation}
Also, \eqref{i11} and \eqref{i12} in terms of $r_{\alpha}(t,z)$ become, when using \eqref{i12} in \eqref{i11},
\begin{eqnarray}
	R_{\alpha}^t(t,t)&=&-\int_0^{\infty}z\Big[p_{\alpha}^2(t,z)+\pi^2q_{\alpha}^2(t,z)\Big]\d w_{\alpha}(z)+2R_{\alpha}^t(t,-t)\int_0^{\infty}q_{\alpha}(t,z)p_{\alpha}(t,z)\,\d w_{\alpha}(z)\nonumber\\
	&=&-\int_0^{\infty}z\big|r_{\alpha}(t,z)\big|^2\d w_{\alpha}(z)-\frac{1}{2\pi^2 t}\left[\int_0^{\infty}\Im\big(r_{\alpha}^2(t,z)\big)\,\d w_{\alpha}(z)\right]^2,\label{i16}\\
	tR_{\alpha}^t(t,-t)&=&\frac{1}{2\pi}\int_0^{\infty}\Im\big(r_{\alpha}^2(t,z)\big)\,\d w_{\alpha}(z).\nonumber
\end{eqnarray}
Combining \eqref{i11},\eqref{i12},\eqref{i13},\eqref{i14},\eqref{i15} and \eqref{i16} we arrive at the following integro-differential Gaudin-Mehta identities, cf. \cite[$(6.26),(6.27),(6.28)$]{TW0} or \cite[$(21.1.11)-(21.1.13)$]{M} for the classical versions of these identities when $w_{\alpha}$ is the indicator function on $(-1,1)$.
\begin{prop}[Integro-differential Gaudin-Mehta identities] For any $t,\alpha>0$,
\begin{equation}\label{i17}
	\frac{\partial}{\partial t}\big(tR_{\alpha}^t(t,-t)\big)=-\int_0^{\infty}z\,\Re\big(r_{\alpha}^2(t,z)\big)\,\d w_{\alpha}(z),
\end{equation}
followed by
\begin{equation}\label{i18}
	\frac{\partial}{\partial t}\big(tR_{\alpha}^t(t,t)\big)=-\int_0^{\infty}z\big|r_{\alpha}(t,z)\big|^2\d w_{\alpha}(z),
\end{equation}
and concluding with
\begin{equation}\label{i19}
	\frac{\partial}{\partial t}R_{\alpha}^t(t,t)=2\big(R_{\alpha}^t(t,-t)\big)^2.
\end{equation}
\end{prop}
Finally we return to the Fredholm determinant $D(t,\alpha)\in(0,1)$ in \eqref{h13} and compute with \eqref{i2},\eqref{i19},
\begin{equation}\label{i19a}
	\frac{\partial^2}{\partial t^2}\ln D(t,\alpha)=-4\big(R_{\alpha}^t(t,-t)\big)^2=-\left[\frac{1}{\pi t}\int_0^{\infty}\Im\big(r_{\alpha}^2(t,z)\big)\,\d w_{\alpha}(z)\right]^2,\ \ \ t,\alpha>0.
\end{equation}
We now introduce
\begin{equation}\label{i20}
	u_{\alpha}(t,z):=-\frac{1}{\pi t}\Im\big(r_{\alpha}^2(t,z)\big)=\frac{2}{t}p_{\alpha}(t,z)q_{\alpha}(t,z)\sim 2z,\ \ \ \ t\downarrow 0,
\end{equation}
where the small $t$-behavior holds pointwise in $z,\alpha>0$ since $\|K_{\alpha}^t\|\rightarrow 0$ as $t\downarrow 0$ in operator norm for any $\alpha>0$. Moving ahead, \eqref{i15} yields, for any $z,\alpha>0$,
\begin{equation}\label{i21}
	\frac{\partial}{\partial t}\big(tu_{\alpha}(t,z)\big)=-\frac{1}{\pi}\frac{\partial}{\partial t}\Im\big(r_{\alpha}^2(t,z)\big)\stackrel{\eqref{i15}}{=}2z\,\Re\big(r_{\alpha}^2(t,z)\big),
\end{equation}
and thus
\begin{align}
	\frac{\partial^2}{\partial t^2}\big(tu_{\alpha}(t,z)\big)\stackrel{\eqref{i15}}{=}-4(\pi z)^2tu_{\alpha}(t,z)&\,+8zR_{\alpha}^t(t,-t)\big|r_{\alpha}(t,z)\big|^2\nonumber\\
	&\,=-4(\pi z)^2tu_{\alpha}(t,z)-4z\big|r_{\alpha}(t,z)\big|^2\int_0^{\infty}u_{\alpha}(t,z)\,\d w_{\alpha}(z).\label{i22}
\end{align}
Now square \eqref{i22} to obtain
\begin{equation}\label{i23}
	\left[\frac{\partial^2}{\partial t^2}\big(tu_{\alpha}(t,z)\big)+4(\pi z)^2tu_{\alpha}(t,z)\right]^2=16z^2\big|r_{\alpha}(t,z)\big|^4\left[\int_0^{\infty}u_{\alpha}(t,z)\,\d w_{\alpha}(z)\right]^2,
\end{equation}
and since $|r_{\alpha}|^4=|r_{\alpha}^2|^2=(\Re(r_{\alpha}^2))^2+(\Im(r_{\alpha}^2))^2$ we can use \eqref{i21} to replace $\Re(r_{\alpha}^2)$ and \eqref{i20} to replace $\Im(r_{\alpha}^2)$. This yields the following integro-differential representation formula for $D(t,\alpha)$.
\begin{theo}\label{main} For any $t,\alpha>0$,
\begin{equation}\label{i24}
	D(t,\alpha)=\exp\left[-2t\int_0^{\infty}w_{\alpha}(z)\,\d z-\int_0^t(t-s)\left\{\int_0^{\infty}u_{\alpha}(s,z)\,\d w_{\alpha}(z)\right\}^2\d s\right]
\end{equation}
where $u_{\alpha}=u_{\alpha}(t,z)$ solves the integro-differential equation
\begin{equation*}
	\left[\frac{\partial^2}{\partial t^2}\big(tu_{\alpha}(t,z)\big)+4(\pi z)^2tu_{\alpha}(t,z)\right]^2=4\left[\int_0^{\infty}u_{\alpha}(t,z)\,\d w_{\alpha}(z)\right]^2\left[\left(\frac{\partial}{\partial t}\big(tu_{\alpha}(t,z)\big)\right)^2+\big(2\pi ztu_{\alpha}(t,z)\big)^2\right]
\end{equation*}
and obeys the boundary condition $u_{\alpha}(t,z)\sim 2z$ as $t\downarrow 0$ for any fixed $z,\alpha>0$.
\end{theo}
\begin{proof} We have already explained how the integro-differential equation for $u_{\alpha}(t,z)$ follows from \eqref{i23}. Likewise, the behavior of $u_{\alpha}(t,z)$ near $t=0$ had been previously established, so it remains to justify \eqref{i24}. By \eqref{i19a} and \eqref{i20},
\begin{equation}\label{i25}
	\frac{\partial^2}{\partial t^2}\ln D(t,\alpha)=-\left\{\int_0^{\infty}u_{\alpha}(t,z)\,\d w_{\alpha}(z)\right\}^2\ \ \ \ \forall\ t,\alpha>0.
\end{equation}
But 
\begin{equation*}
	(0,\infty)\ni t\mapsto\int_0^{\infty}u_{\alpha}(t,z)\d w_{\alpha}(z)
\end{equation*}
is continuous and integrable at $t=0$ by the asymptotic and analytic properties of $u_{\alpha}$ and $w_{\alpha}$. Hence, integrating \eqref{i25} once, we find
\begin{equation*}
	\frac{\partial}{\partial t}\ln D(t,\alpha)=f(\alpha)-\int_0^t\left\{\int_0^{\infty}u_{\alpha}(s,z)\,\d w_{\alpha}(z)\right\}^2\d s,\ \ \ \ \ t,\alpha>0,
\end{equation*}
where $f(\alpha)$ only depends on $\alpha>0$, but not on $t>0$. However,
\begin{equation*}
	(0,\infty)\ni t\mapsto\int_0^t\left\{\int_0^{\infty}u_{\alpha}(s,z)\,\d w_{\alpha}(z)\right\}^2\d s,\ \ \ \alpha>0,
\end{equation*}
is also continuous and integrable at $t=0$, i.e. by another integration, Fubini's theorem and since $D(0,\alpha)=1$,
\begin{eqnarray}
	\ln D(t,\alpha)\!\!&=&\!\!tf(\alpha)-\int_0^t\left[\int_0^{\lambda}\left\{\int_0^{\infty}u_{\alpha}(s,z)\,\d w_{\alpha}(z)\right\}^2\d s\right]\d\lambda\nonumber\\
	&=&\!\!tf(\alpha)-\int_0^t(t-s)\left\{\int_0^{\infty}u_{\alpha}(s,z)\,\d w_{\alpha}(z)\right\}^2\d s,\ \ \ t,\alpha>0.\label{i26}
\end{eqnarray}
In order to compute the outstanding function $f$ we shall use the Fredholm series for $D(t,\alpha)$,
\begin{equation}\label{i27}
	D(t,\alpha)=1-2t\int_0^{\infty}w_{\alpha}(z)\,\d z+\sum_{\ell=2}^{\infty}\frac{(-1)^{\ell}}{\ell!}\int_{(-t,t)^{\ell}}\det\big[K_{\alpha}(x_j,x_k)\big]_{j,k=1}^{\ell}\d x_1\cdots\d x_{\ell}.
\end{equation}
Here the remainder is easily seen to be of order $\mathcal{O}(t^2)$ as $t\downarrow 0$ by Hadamard's inequality and thus
\begin{equation*}
	f(\alpha)=\frac{\partial}{\partial t}D(t,\alpha)\bigg|_{t=0}\stackrel{\eqref{i27}}{=}-2\int_0^{\infty}w_{\alpha}(z)\,\d z,
\end{equation*}
which, once substituted into \eqref{i26}, completes our proof of the Theorem.
\end{proof}
Equipped with Theorem \ref{main} we are prepared to derive Theorem \ref{theo:3}.
\begin{proof}[Proof of Theorem \ref{theo:3}] Simply choose
\begin{equation*}
	w_{\alpha}(z)=\Phi\big(\alpha(z+1)\big)-\Phi\big(\alpha(z-1)\big),\ \ \ z\in\mathbb{R},\ \alpha>0
\end{equation*}
in Theorem \ref{main}, compare \eqref{h8}, and evaluate \eqref{i24} at $\alpha=t/\sigma>0$. What results from \eqref{i24} is precisely \eqref{r37} and \eqref{r38} for $D_{\sigma}(t)$ as in \eqref{r36}.
\end{proof}



\begin{bibsection}
\begin{biblist}


\bib{ACQ}{article}{
AUTHOR = {Amir, Gideon},
author={Corwin, Ivan},
author={Quastel, Jeremy},
     TITLE = {Probability distribution of the free energy of the continuum
              directed random polymer in {$1+1$} dimensions},
   JOURNAL = {Comm. Pure Appl. Math.},
  FJOURNAL = {Communications on Pure and Applied Mathematics},
    VOLUME = {64},
      YEAR = {2011},
    NUMBER = {4},
     PAGES = {466--537},
      ISSN = {0010-3640},
   MRCLASS = {60K35 (60B20 60F05 60H15 82C22 82C44)},
  MRNUMBER = {2796514},
MRREVIEWER = {Timo Sepp\"{a}l\"{a}inen},
       DOI = {10.1002/cpa.20347},
       URL = {https://doi-org.bris.idm.oclc.org/10.1002/cpa.20347},
}


\bib{ACV}{article}{
AUTHOR = {Akemann, Gernot},
author={Cikovic, Milan},
author={Venker, Martin},
     TITLE = {Universality at weak and strong non-{H}ermiticity beyond the
              elliptic {G}inibre ensemble},
   JOURNAL = {Comm. Math. Phys.},
  FJOURNAL = {Communications in Mathematical Physics},
    VOLUME = {362},
      YEAR = {2018},
    NUMBER = {3},
     PAGES = {1111--1141},
      ISSN = {0010-3616},
   MRCLASS = {60B20 (82B05)},
  MRNUMBER = {3845296},
MRREVIEWER = {Dominique L\'{e}pingle},
       DOI = {10.1007/s00220-018-3201-1},
       URL = {https://doi-org.bris.idm.oclc.org/10.1007/s00220-018-3201-1},
}

\bib{AGZ}{book}{
AUTHOR = {Anderson, Greg W.},
author={Guionnet, Alice},
author={Zeitouni, Ofer},
     TITLE = {An introduction to random matrices},
    SERIES = {Cambridge Studies in Advanced Mathematics},
    VOLUME = {118},
 PUBLISHER = {Cambridge University Press, Cambridge},
      YEAR = {2010},
     PAGES = {xiv+492},
      ISBN = {978-0-521-19452-5},
   MRCLASS = {60B20 (46L53 46L54)},
  MRNUMBER = {2760897},
MRREVIEWER = {Terence Tao},
}

\bib{Ben}{article}{
AUTHOR = {Bender, Martin},
     TITLE = {Edge scaling limits for a family of non-{H}ermitian random
              matrix ensembles},
   JOURNAL = {Probab. Theory Related Fields},
  FJOURNAL = {Probability Theory and Related Fields},
    VOLUME = {147},
      YEAR = {2010},
    NUMBER = {1-2},
     PAGES = {241--271},
      ISSN = {0178-8051},
   MRCLASS = {60B20 (60G55 60G70)},
  MRNUMBER = {2594353},
       DOI = {10.1007/s00440-009-0207-9},
       URL = {https://doi-org.bris.idm.oclc.org/10.1007/s00440-009-0207-9},
}


\bib{Bo0}{article}{
AUTHOR = {Bothner, T.},
     TITLE = {On the origins of {R}iemann-{H}ilbert problems in mathematics},
   JOURNAL = {Nonlinearity},
  FJOURNAL = {Nonlinearity},
    VOLUME = {34},
      YEAR = {2021},
    NUMBER = {4},
     PAGES = {R1--R73},
      ISSN = {0951-7715},
   MRCLASS = {30E25 (01A60 45M05 60B20)},
  MRNUMBER = {4246443},
       DOI = {10.1088/1361-6544/abb543},
       URL = {https://doi-org.bris.idm.oclc.org/10.1088/1361-6544/abb543},
}

\bib{Bo}{article}{
AUTHOR = {Bothner, T.},
TITLE = {A Riemann-Hilbert approach to Fredholm determinants of Hankel composition operators: scalar-valued kernels},
YEAR = {2022},
eprint={https://arxiv.org/abs/2205.15007},
      archivePrefix={arXiv},
      primaryClass={math-ph},
}


\bib{BCT}{article}{
AUTHOR = {Bothner, Thomas},
author={Cafasso, Mattia},
author={Tarricone, Sofia},
     TITLE = {Momenta spacing distributions in anharmonic oscillators and
              the higher order finite temperature {A}iry kernel},
   JOURNAL = {Ann. Inst. Henri Poincar\'{e} Probab. Stat.},
  FJOURNAL = {Annales de l'Institut Henri Poincar\'{e} Probabilit\'{e}s et
              Statistiques},
    VOLUME = {58},
      YEAR = {2022},
    NUMBER = {3},
     PAGES = {1505--1546},
      ISSN = {0246-0203},
   MRCLASS = {45J05 (30E25 33C10 35J10 42A38 81V70)},
  MRNUMBER = {4452641},
       DOI = {10.1214/21-aihp1211},
       URL = {https://doi-org.bris.idm.oclc.org/10.1214/21-aihp1211},
}

\bib{BL}{article}{
  eprint = {https://arxiv.org/abs/2208.04684},
  archivePrefix={arXiv},
    primaryClass={math.PH},
  author = {Bothner, Thomas},
  author={Little, Alex},
  title = {The complex elliptic Ginibre ensemble at weak non-Hermiticity: edge spacing distributions},
  year = {2022},
  copyright = {arXiv.org perpetual, non-exclusive license}
}



\bib{BG}{article}{
AUTHOR = {Boyer, Robert},
author={Goh, William M. Y.},
     TITLE = {On the zero attractor of the {E}uler polynomials},
   JOURNAL = {Adv. in Appl. Math.},
  FJOURNAL = {Advances in Applied Mathematics},
    VOLUME = {38},
      YEAR = {2007},
    NUMBER = {1},
     PAGES = {97--132},
      ISSN = {0196-8858},
   MRCLASS = {30C15 (33C45)},
  MRNUMBER = {2288197},
MRREVIEWER = {Heinrich Begehr},
       DOI = {10.1016/j.aam.2005.05.008},
       URL = {https://doi.org/10.1016/j.aam.2005.05.008},
}

\bib{CCR}{article}{
AUTHOR = {Cafasso, Mattia},
author={Claeys, Tom},
author={Ruzza, Giulio},
     TITLE = {Airy kernel determinant solutions to the {K}d{V} equation and
              integro-differential {P}ainlev\'{e} equations},
   JOURNAL = {Comm. Math. Phys.},
  FJOURNAL = {Communications in Mathematical Physics},
    VOLUME = {386},
      YEAR = {2021},
    NUMBER = {2},
     PAGES = {1107--1153},
      ISSN = {0010-3616},
   MRCLASS = {35Q53 (33C10 34M50 34M55 45K05)},
  MRNUMBER = {4294287},
       DOI = {10.1007/s00220-021-04108-9},
       URL = {https://doi-org.bris.idm.oclc.org/10.1007/s00220-021-04108-9},
}

 \bib{CESX}{article}{
AUTHOR = {Cipolloni, Giorgio},
author={Erd\"os, L\'{a}szl\'{o}},
author={Schr\"{o}der, Dominik},
author={Xu, Yuanyuan},
     TITLE = {Directional extremal statistics for {G}inibre eigenvalues},
   JOURNAL = {J. Math. Phys.},
  FJOURNAL = {Journal of Mathematical Physics},
    VOLUME = {63},
      YEAR = {2022},
    NUMBER = {10},
     PAGES = {Paper No. 103303, 11},
      ISSN = {0022-2488},
   MRCLASS = {60B20 (15A18 15B52)},
  MRNUMBER = {4496015},
       DOI = {10.1063/5.0104290},
       URL = {https://doi-org.bris.idm.oclc.org/10.1063/5.0104290},
}

\bib{CESX2}{article}{
AUTHOR = {Cipolloni, Giorgio},
author={Erd\"os, L\'{a}szl\'{o}},
author={Schr\"{o}der, Dominik},
author={Xu, Yuanyuan},
title={On the rightmost eigenvalue of non-Hermitian random matrices},
YEAR = {2022},
eprint={https://arxiv.org/abs/2206.04448},
      archivePrefix={arXiv},
      primaryClass={math.PR},
}


\bib{FGIL}{article}{
AUTHOR = {Di Francesco, P.},
author={Gaudin, M.},
author={Itzykson, C.},
author={Lesage, F.},
     TITLE = {Laughlin's wave functions, {C}oulomb gases and expansions of
              the discriminant},
   JOURNAL = {Internat. J. Modern Phys. A},
  FJOURNAL = {International Journal of Modern Physics A. Particles and
              Fields. Gravitation. Cosmology},
    VOLUME = {9},
      YEAR = {1994},
    NUMBER = {24},
     PAGES = {4257--4351},
      ISSN = {0217-751X},
   MRCLASS = {81V70 (05E10 22E70 52B11 82D10)},
  MRNUMBER = {1289574},
MRREVIEWER = {Peter N. Zhevandrov},
       DOI = {10.1142/S0217751X94001734},
       URL = {https://doi-org.bris.idm.oclc.org/10.1142/S0217751X94001734},
}


%
\bib{DDMS}{article}{
  title = {Noninteracting fermions at finite temperature in a $d$-dimensional trap: Universal correlations},
  author = {Dean, David S.},
  author={Le Doussal, Pierre},
  author={Majumdar, Satya N.},
  author={Schehr, Gr\'egory},
  journal = {Phys. Rev. A},
  volume = {94},
  issue = {6},
  pages = {063622},
  numpages = {41},
  year = {2016},
  month = {Dec},
  publisher = {American Physical Society},
  doi = {10.1103/PhysRevA.94.063622},
  url = {https://link.aps.org/doi/10.1103/PhysRevA.94.063622}
}
%
%


\bib{DKMVZ}{article}{
AUTHOR = {Deift, P.},
author={Kriecherbauer, T.},
author={McLaughlin, K. T.-R.},
author={Venakides, S.},
author={Zhou, X.},
     TITLE = {Uniform asymptotics for polynomials orthogonal with respect to
              varying exponential weights and applications to universality
              questions in random matrix theory},
   JOURNAL = {Comm. Pure Appl. Math.},
  FJOURNAL = {Communications on Pure and Applied Mathematics},
    VOLUME = {52},
      YEAR = {1999},
    NUMBER = {11},
     PAGES = {1335--1425},
      ISSN = {0010-3640},
   MRCLASS = {42C05 (15A52 41A60 82B41)},
  MRNUMBER = {1702716},
MRREVIEWER = {D. S. Lubinsky},
       DOI =
              {10.1002/(SICI)1097-0312(199911)52:11<1335::AID-CPA1>3.0.CO;2-1},
       URL =
              {https://doi-org.bris.idm.oclc.org/10.1002/(SICI)1097-0312(199911)52:11<1335::AID-CPA1>3.0.CO;2-1},
}


\bib{Dy1}{article}{
AUTHOR = {Dyson, Freeman J.},
     TITLE = {The threefold way. {A}lgebraic structure of symmetry groups
              and ensembles in quantum mechanics},
   JOURNAL = {J. Mathematical Phys.},
  FJOURNAL = {Journal of Mathematical Physics},
    VOLUME = {3},
      YEAR = {1962},
     PAGES = {1199--1215},
      ISSN = {0022-2488},
   MRCLASS = {81.22 (22.57)},
  MRNUMBER = {177643},
MRREVIEWER = {A. S. Wightman},
       DOI = {10.1063/1.1703863},
       URL = {https://doi-org.bris.idm.oclc.org/10.1063/1.1703863},
}

\bib{F1}{book}{
AUTHOR = {Forrester, P. J.},
     TITLE = {Log-gases and random matrices},
    SERIES = {London Mathematical Society Monographs Series},
    VOLUME = {34},
 PUBLISHER = {Princeton University Press, Princeton, NJ},
      YEAR = {2010},
     PAGES = {xiv+791},
      ISBN = {978-0-691-12829-0},
   MRCLASS = {82-02 (33C45 60B20 82B05 82B41 82B44)},
  MRNUMBER = {2641363},
MRREVIEWER = {Steven Joel Miller},
       DOI = {10.1515/9781400835416},
       URL = {https://doi-org.bris.idm.oclc.org/10.1515/9781400835416},
}


\bib{FJ}{article}{
AUTHOR = {Forrester, P. J.},
author={Jancovici, B.},
     TITLE = {Two-dimensional one-component plasma in a quadrupolar field},
   JOURNAL = {Internat. J. Modern Phys. A},
  FJOURNAL = {International Journal of Modern Physics A. Particles and
              Fields. Gravitation. Cosmology},
    VOLUME = {11},
      YEAR = {1996},
    NUMBER = {5},
     PAGES = {941--949},
      ISSN = {0217-751X},
   MRCLASS = {82B05 (82B23 82D10)},
  MRNUMBER = {1371262},
MRREVIEWER = {M. Lawrence Glasser},
       DOI = {10.1142/S0217751X96000432},
       URL = {https://doi-org.bris.idm.oclc.org/10.1142/S0217751X96000432},
}


\bib{FKS1}{article}{
AUTHOR = {Fyodorov, Yan V.},
author={Khoruzhenko, Boris A.},
author={Sommers, Hans-J\"{u}rgen},
     TITLE = {Almost-{H}ermitian random matrices: eigenvalue density in the
              complex plane},
   JOURNAL = {Phys. Lett. A},
  FJOURNAL = {Physics Letters. A},
    VOLUME = {226},
      YEAR = {1997},
    NUMBER = {1-2},
     PAGES = {46--52},
      ISSN = {0375-9601},
   MRCLASS = {82B05 (82B44)},
  MRNUMBER = {1431718},
MRREVIEWER = {Estelle L. Basor},
       DOI = {10.1016/S0375-9601(96)00904-8},
       URL = {https://doi-org.bris.idm.oclc.org/10.1016/S0375-9601(96)00904-8},
}




\bib{FKS2}{article}{
AUTHOR = {Fyodorov, Yan V.},
author={Khoruzhenko, Boris A.},
author={Sommers, Hans-J\"{u}rgen},
     TITLE = {Almost {H}ermitian random matrices: crossover from
              {W}igner-{D}yson to {G}inibre eigenvalue statistics},
   JOURNAL = {Phys. Rev. Lett.},
  FJOURNAL = {Physical Review Letters},
    VOLUME = {79},
      YEAR = {1997},
    NUMBER = {4},
     PAGES = {557--560},
      ISSN = {0031-9007},
   MRCLASS = {82B41},
  MRNUMBER = {1459918},
       DOI = {10.1103/PhysRevLett.79.557},
       URL = {https://doi-org.bris.idm.oclc.org/10.1103/PhysRevLett.79.557},
}

\bib{FSK}{article}{
AUTHOR = {Fyodorov, Yan V.},
author={Sommers, Hans-J\"{u}rgen},
author={Khoruzhenko, Boris A.},
     TITLE = {Universality in the random matrix spectra in the regime of
              weak non-{H}ermiticity},
      NOTE = {Classical and quantum chaos},
   JOURNAL = {Ann. Inst. H. Poincar\'{e} Phys. Th\'{e}or.},
  FJOURNAL = {Annales de l'Institut Henri Poincar\'{e}. Physique Th\'{e}orique},
    VOLUME = {68},
      YEAR = {1998},
    NUMBER = {4},
     PAGES = {449--489},
      ISSN = {0246-0211},
   MRCLASS = {60F99 (15A52 60B15 82B41)},
  MRNUMBER = {1634312},
MRREVIEWER = {Oleksiy Khorunzhiy},
       URL = {http://www.numdam.org/item?id=AIHPA_1998__68_4_449_0},
}

\bib{GA}{article}{
title={Connectance of large dynamic (cybernetic) systems: critical values for stability},
  author={Gardner, Mark R},
  author={Ashby, W Ross},
  journal={Nature},
  volume={228},
  number={5273},
  pages={784--784},
  year={1970},
  publisher={Nature Publishing Group}
}

\bib{Gi}{article}{
AUTHOR = {Ginibre, Jean},
     TITLE = {Statistical ensembles of complex, quaternion, and real
              matrices},
   JOURNAL = {J. Mathematical Phys.},
  FJOURNAL = {Journal of Mathematical Physics},
    VOLUME = {6},
      YEAR = {1965},
     PAGES = {440--449},
      ISSN = {0022-2488},
   MRCLASS = {22.60 (53.90)},
  MRNUMBER = {173726},
MRREVIEWER = {J. Dieudonn\'{e}},
       DOI = {10.1063/1.1704292},
       URL = {https://doi-org.bris.idm.oclc.org/10.1063/1.1704292},
}

\bib{Gir}{article}{
AUTHOR = {Girko, V. L.},
     TITLE = {The elliptic law},
   JOURNAL = {Teor. Veroyatnost. i Primenen.},
  FJOURNAL = {Akademiya Nauk SSSR. Teoriya Veroyatnoste\u{\i} i ee Primeneniya},
    VOLUME = {30},
      YEAR = {1985},
    NUMBER = {4},
     PAGES = {640--651},
      ISSN = {0040-361X},
   MRCLASS = {60F99 (81G45 82A31)},
  MRNUMBER = {816278},
MRREVIEWER = {Nina B. Maslova},
}


\bib{GGK}{book}{
AUTHOR = {Gohberg, Israel},
author={Goldberg, Seymour},
author={Krupnik, Nahum},
     TITLE = {Traces and determinants of linear operators},
    SERIES = {Operator Theory: Advances and Applications},
    VOLUME = {116},
 PUBLISHER = {Birkh\"{a}user Verlag, Basel},
      YEAR = {2000},
     PAGES = {x+258},
      ISBN = {3-7643-6177-8},
   MRCLASS = {47B10 (45B05 45P05 47A53 47G10 47L10)},
  MRNUMBER = {1744872},
MRREVIEWER = {Hermann K\"{o}nig},
       DOI = {10.1007/978-3-0348-8401-3},
       URL = {https://doi-org.bris.idm.oclc.org/10.1007/978-3-0348-8401-3},
}


\bib{HKPV}{book}{
AUTHOR = {Hough, J. Ben},
author={Krishnapur, Manjunath},
author={Peres, Yuval},
author={Vir\'{a}g, B\'{a}lint},
     TITLE = {Zeros of {G}aussian analytic functions and determinantal point
              processes},
    SERIES = {University Lecture Series},
    VOLUME = {51},
 PUBLISHER = {American Mathematical Society, Providence, RI},
      YEAR = {2009},
     PAGES = {x+154},
      ISBN = {978-0-8218-4373-4},
   MRCLASS = {60G55 (30B20 30C15 60B20 60F10 60G15 65H04 82B31)},
  MRNUMBER = {2552864},
MRREVIEWER = {Dmitry Beliaev},
       DOI = {10.1090/ulect/051},
       URL = {https://doi.org/10.1090/ulect/051},
}


\bib{IIKS}{article}{
AUTHOR = {Its, A. R.},
AUTHOR={Izergin, A. G.},
AUTHOR={Korepin, V. E.},
AUTHOR={ Slavnov, N. A.},
     TITLE = {Differential equations for quantum correlation functions},
 BOOKTITLE = {Proceedings of the {C}onference on {Y}ang-{B}axter
              {E}quations, {C}onformal {I}nvariance and {I}ntegrability in
              {S}tatistical {M}echanics and {F}ield {T}heory},
   JOURNAL = {Internat. J. Modern Phys. B},
  FJOURNAL = {International Journal of Modern Physics B},
    VOLUME = {4},
      YEAR = {1990},
    NUMBER = {5},
     PAGES = {1003--1037},
      ISSN = {0217-9792},
   MRCLASS = {82B10 (35Q40 58G40 82C10)},
  MRNUMBER = {1064758},
MRREVIEWER = {Anatoliy Prykarpatsky},
       DOI = {10.1142/S0217979290000504},
       URL = {https://doi-org.bris.idm.oclc.org/10.1142/S0217979290000504},
}

\bib{JMMS}{article}{
AUTHOR = {Jimbo, Michio},
author={Miwa, Tetsuji},
author={M\^{o}ri, Yasuko},
author={Sato, Mikio},
     TITLE = {Density matrix of an impenetrable {B}ose gas and the fifth
              {P}ainlev\'{e} transcendent},
   JOURNAL = {Phys. D},
  FJOURNAL = {Physica D. Nonlinear Phenomena},
    VOLUME = {1},
      YEAR = {1980},
    NUMBER = {1},
     PAGES = {80--158},
      ISSN = {0167-2789},
   MRCLASS = {82A15 (14D05 58F07)},
  MRNUMBER = {573370},
       DOI = {10.1016/0167-2789(80)90006-8},
       URL = {https://doi-org.bris.idm.oclc.org/10.1016/0167-2789(80)90006-8},
}
%

\bib{Joh0}{article}{
author={Johansson, Kurt},
 TITLE={Random matrices and determinantal processes},
    YEAR={2005},
    journal={Lecture notes from the Les Houches summer school on Mathematical Statistical Physics},
    eprint={https://arxiv.org/abs/math-ph/0510038v1},
    archivePrefix={arXiv},
    primaryClass={math.PH},
}


\bib{Joh}{article}{
AUTHOR = {Johansson, K.},
     TITLE = {From {G}umbel to {T}racy-{W}idom},
   JOURNAL = {Probab. Theory Related Fields},
  FJOURNAL = {Probability Theory and Related Fields},
    VOLUME = {138},
      YEAR = {2007},
    NUMBER = {1-2},
     PAGES = {75--112},
      ISSN = {0178-8051},
   MRCLASS = {60G70 (15A52 60G07 62G32 82B41)},
  MRNUMBER = {2288065},
MRREVIEWER = {Alexander Roitershtein},
       DOI = {10.1007/s00440-006-0012-7},
       URL = {https://doi-org.bris.idm.oclc.org/10.1007/s00440-006-0012-7},
}

\bib{Kos}{article}{
AUTHOR = {Kostlan, Eric},
     TITLE = {On the spectra of {G}aussian matrices},
      NOTE = {Directions in matrix theory (Auburn, AL, 1990)},
   JOURNAL = {Linear Algebra Appl.},
  FJOURNAL = {Linear Algebra and its Applications},
    VOLUME = {162/164},
      YEAR = {1992},
     PAGES = {385--388},
      ISSN = {0024-3795},
   MRCLASS = {62H10},
  MRNUMBER = {1148410},
MRREVIEWER = {N. Giri},
       DOI = {10.1016/0024-3795(92)90386-O},
       URL = {https://doi.org/10.1016/0024-3795(92)90386-O},
}

\bib{Kra}{article}{
AUTHOR = {Krajenbrink, Alexandre},
     TITLE = {From {P}ainlev\'{e} to {Z}akharov-{S}habat and beyond: {F}redholm
              determinants and integro-differential hierarchies},
   JOURNAL = {J. Phys. A},
  FJOURNAL = {Journal of Physics. A. Mathematical and Theoretical},
    VOLUME = {54},
      YEAR = {2021},
    NUMBER = {3},
     PAGES = {Paper No. 035001, 51},
      ISSN = {1751-8113},
   MRCLASS = {37K10 (34M55 37J65)},
  MRNUMBER = {4209129},
       DOI = {10.1088/1751-8121/abd078},
       URL = {https://doi-org.bris.idm.oclc.org/10.1088/1751-8121/abd078},
}

\bib{LS}{article}{
AUTHOR = {Lehmann, Nils},
author={Sommers, Hans-J\"{u}rgen},
     TITLE = {Eigenvalue statistics of random real matrices},
   JOURNAL = {Phys. Rev. Lett.},
  FJOURNAL = {Physical Review Letters},
    VOLUME = {67},
      YEAR = {1991},
    NUMBER = {8},
     PAGES = {941--944},
      ISSN = {0031-9007},
   MRCLASS = {82B41 (15A18 15A52 82C32)},
  MRNUMBER = {1121461},
       DOI = {10.1103/PhysRevLett.67.941},
       URL = {https://doi-org.bris.idm.oclc.org/10.1103/PhysRevLett.67.941},
}


\bib{LW}{article}{
AUTHOR = {Liechty, Karl},
author={Wang, Dong},
     TITLE = {Asymptotics of free fermions in a quadratic well at finite
              temperature and the {M}oshe-{N}euberger-{S}hapiro random
              matrix model},
   JOURNAL = {Ann. Inst. Henri Poincar\'{e} Probab. Stat.},
  FJOURNAL = {Annales de l'Institut Henri Poincar\'{e} Probabilit\'{e}s et
              Statistiques},
    VOLUME = {56},
      YEAR = {2020},
    NUMBER = {2},
     PAGES = {1072--1098},
      ISSN = {0246-0203},
   MRCLASS = {60B20 (15B52 82B23)},
  MRNUMBER = {4076776},
       DOI = {10.1214/19-AIHP994},
       URL = {https://doi-org.bris.idm.oclc.org/10.1214/19-AIHP994},
}

\bib{Ma}{article}{
AUTHOR={May, R.},
TITLE={Will a Large Complex System be Stable?},
JOURNAL={Nature},
year={1972},
volume={238},
number={5364},
pages={413--414},
doi={10.1038/238413a0},
url={https://doi.org/10.1038/238413a0},
}


\bib{M}{book}{
AUTHOR = {Mehta, Madan Lal},
     TITLE = {Random matrices},
    SERIES = {Pure and Applied Mathematics (Amsterdam)},
    VOLUME = {142},
   EDITION = {Third},
 PUBLISHER = {Elsevier/Academic Press, Amsterdam},
      YEAR = {2004},
     PAGES = {xviii+688},
      ISBN = {0-12-088409-7},
   MRCLASS = {82-02 (15-02 15A52 60B99 60K35 82B41)},
  MRNUMBER = {2129906},
}


\bib{NIST}{book}{
TITLE = {N{IST} handbook of mathematical functions},
    EDITOR = {Olver, Frank W. J.}
    editor={Lozier, Daniel W.}
    editor={Boisvert, Ronald F.}
    editor={Clark, Charles W.},
 PUBLISHER = {U.S. Department of Commerce, National Institute of Standards
              and Technology, Washington, DC; Cambridge University Press,
              Cambridge},
      YEAR = {2010},
     PAGES = {xvi+951},
      ISBN = {978-0-521-14063-8},
   MRCLASS = {33-00 (00A20 65-00)},
  MRNUMBER = {2723248},
}


\bib{PS}{book}{
AUTHOR = {Pastur, Leonid},
author={Shcherbina, Mariya},
     TITLE = {Eigenvalue distribution of large random matrices},
    SERIES = {Mathematical Surveys and Monographs},
    VOLUME = {171},
 PUBLISHER = {American Mathematical Society, Providence, RI},
      YEAR = {2011},
     PAGES = {xiv+632},
      ISBN = {978-0-8218-5285-9},
   MRCLASS = {60B20 (15A18 15B52 60F05 62H10 62H99)},
  MRNUMBER = {2808038},
MRREVIEWER = {Terence Tao},
       DOI = {10.1090/surv/171},
       URL = {https://doi-org.bris.idm.oclc.org/10.1090/surv/171},
}

\bib{Po}{book}{
author={Porter, C.E.},
title={Fluctuations of quantal spectra},
series={Statistical Theories of Spectra: Fluctuations},
PUBLISHER={Academic Press, New York},
YEAR={1965},
}

\bib{Sim}{book}{
AUTHOR = {Simon, Barry},
     TITLE = {Trace ideals and their applications},
    SERIES = {Mathematical Surveys and Monographs},
    VOLUME = {120},
   EDITION = {Second},
 PUBLISHER = {American Mathematical Society, Providence, RI},
      YEAR = {2005},
     PAGES = {viii+150},
      ISBN = {0-8218-3581-5},
   MRCLASS = {47L20 (47A40 47A55 47B10 47B36 47E05 81Q15 81U99)},
  MRNUMBER = {2154153},
MRREVIEWER = {Pavel B. Kurasov},
       DOI = {10.1090/surv/120},
       URL = {https://doi-org.bris.idm.oclc.org/10.1090/surv/120},
}

\bib{Ta}{article}{
AUTHOR = {Tao, Terence},
     TITLE = {The asymptotic distribution of a single eigenvalue gap of a
              {W}igner matrix},
   JOURNAL = {Probab. Theory Related Fields},
  FJOURNAL = {Probability Theory and Related Fields},
    VOLUME = {157},
      YEAR = {2013},
    NUMBER = {1-2},
     PAGES = {81--106},
      ISSN = {0178-8051},
   MRCLASS = {60B20},
  MRNUMBER = {3101841},
MRREVIEWER = {Nizar Demni},
       DOI = {10.1007/s00440-012-0450-3},
       URL = {https://doi-org.bris.idm.oclc.org/10.1007/s00440-012-0450-3},
}

\bib{TW0}{article}{
AUTHOR = {Tracy, Craig A.},
author={Widom, Harold},
     TITLE = {Introduction to random matrices},
 BOOKTITLE = {Geometric and quantum aspects of integrable systems
              ({S}cheveningen, 1992)},
    SERIES = {Lecture Notes in Phys.},
    VOLUME = {424},
     PAGES = {103--130},
 PUBLISHER = {Springer, Berlin},
      YEAR = {1993},
   MRCLASS = {82B41},
  MRNUMBER = {1253763},
       DOI = {10.1007/BFb0021444},
       URL = {https://doi-org.bris.idm.oclc.org/10.1007/BFb0021444},
}

\bib{TW}{article}{
AUTHOR = {Tracy, Craig A.}
author={Widom, Harold},
     TITLE = {Level-spacing distributions and the {A}iry kernel},
   JOURNAL = {Comm. Math. Phys.},
  FJOURNAL = {Communications in Mathematical Physics},
    VOLUME = {159},
      YEAR = {1994},
    NUMBER = {1},
     PAGES = {151--174},
      ISSN = {0010-3616},
   MRCLASS = {82B05 (33C90 47A75 47G10 47N55 82B10)},
  MRNUMBER = {1257246},
MRREVIEWER = {Estelle L. Basor},
       URL = {http://projecteuclid.org/euclid.cmp/1104254495},
}





\end{biblist}
\end{bibsection}
\end{document}